\documentclass[11pt]{article}
\usepackage[hmargin=2.35cm,vmargin=2.35cm]{geometry}
\usepackage{url}
\usepackage{amssymb,latexsym,amsmath,amsthm,psfrag,array,cite,comment,bigints,mathrsfs,bbm,setspace}
\usepackage{amsfonts}
\usepackage{array}
\usepackage{csquotes}
\usepackage{graphicx, caption,array,rotating}
\usepackage[colorlinks,linkcolor={blue},urlcolor ={blue},citecolor={blue}]{hyperref}
\usepackage{authblk}
\usepackage{booktabs}
\usepackage[longnamesfirst]{natbib}
\usepackage{graphics}
\usepackage{pdfpages}
\usepackage{mathtools}
\usepackage{booktabs,colortbl,tabularx}
\usepackage{algorithm2e}
\usepackage[framemethod=TikZ]{mdframed}
\usepackage{graphicx,subcaption,float}
\usepackage{adjustbox}
\usepackage{lscape}
\usepackage{caption} % 加在前言區
\usepackage{etoolbox}
\AtBeginEnvironment{align}{
  \setlength{\abovedisplayskip}{3pt} % space above the equation
  \setlength{\belowdisplayskip}{3pt} % space below the equation
}
\AtBeginEnvironment{align*}{
  \setlength{\abovedisplayskip}{3pt} % space above the equation
  \setlength{\belowdisplayskip}{3pt} % space below the equation
}

\AtBeginEnvironment{eqnarray}{
  \setlength{\abovedisplayskip}{3pt} % space above the equation
  \setlength{\belowdisplayskip}{3pt} % space below the equation
}
\AtBeginEnvironment{eqnarray*}{
  \setlength{\abovedisplayskip}{3pt} % space above the equation
  \setlength{\belowdisplayskip}{3pt} % space below the equation
}
\usepackage{titlesec}
\titlespacing*{\section}
  {0pt}{8pt plus 2pt minus 2pt}{6pt plus 2pt minus 2pt}  % {左邊距}{標題與上方段落之間的距離（前距），基本為 10pt，可伸展或收縮 2pt。}{標題與下方內文之間的距離（後距），基本為 6pt，可伸展或收縮。}
\titlespacing*{\subsection}
  {0pt}{8pt plus 2pt minus 2pt}{4pt plus 2pt minus 2pt}
\usepackage{placeins}
\DeclareMathOperator*{\argmin}{arg\,min}
\DeclareMathOperator*{\argmax}{arg\,max}
\DeclareMathOperator*{\median}{median}
\SetAlgoCaptionSeparator{}
\newtheorem{thm}{Theorem}

\newtheorem{aplemma}{Lemma}[section]

\theoremstyle{definition}

\newtheorem{rem}{Remark}
\numberwithin{equation}{section}
\newcommand{\E}{\mathrm{E}}
\newcommand{\vvec}{\mathrm{vec}}
\newcommand{\var}{\mathrm{var}}

\begin{document}

\title{Unsupervised Learning Under a General Semiparametric Clusterwise Elliptical Distribution: Efficient Estimation, Optimal Clustering, and Consistent Cluster Selection
}

\author[1]{Jen-Chieh Teng}
\author[2]{Sheng-Hsin Fan}
\author[1]{Chin-Tsang Chiang}
\author[3]{Ming-Yueh Huang}
\author[4]{Alvin Lim}

\affil[1]{Institute of Applied Mathematical Sciences, National Taiwan University, Taipei, Taiwan}
\affil[2]{Department of Mathematics, National Taiwan University, Taipei, Taiwan} 
\affil[3]{Institute of Statistical Science, Academia Sinica, Taipei, Taiwan} 
\affil[4]{Goizueta Business School, Emory University, Atlanta, GA, USA} 

%\thanks{}%
\date{\today}%
\maketitle
 
% ----------------------------------------------------------------
\begin{abstract}
A general semiparametric clusterwise elliptical distribution is introduced for assessing how latent cluster structure shapes continuous outcomes. Using a subjectwise representation, cluster-specific mean vectors are estimated by minimizing a weighted sum of squares criterion augmented with a separation penalty. An initialization scheme based on iterative Mahalanobis reassignment and a computational algorithm with guaranteed convergence accompany the estimation procedure. 
Under the localization conditions, the separation penalty criterion admits a local oracle property, and the corresponding localized partition recovers the realized latent cluster partition. The practical iterative Mahalanobis reassignment does not impose these truth-centered neighborhoods and is evaluated separately through numerical studies. The resulting first-phase partition seeds a second phase of pseudo-likelihood-based estimation and cluster reassignment.
A semiparametric information criterion is also proposed for selecting the number of clusters. 
Monte Carlo simulations identify settings in which the proposed procedures improve on the evaluated parametric benchmarks, while additional experiments reveal sensitivity to cluster imbalance and departures from ellipticity. Applications illustrate the resulting customer and patient profiles.
\end{abstract}
% ----------------------------------------------------------------
\begin{spacing}{1.9}
\begin{section}{Introduction}

Cluster analysis is a core tool across disciplines---from bioinformatics \citep{liu2008statistical} to marketing \citep{gomes2023customer}---for uncovering latent subpopulations when class labels are unavailable. Also called segmentation in marketing research, clustering is increasingly used not only to decide the number of groups and group observations by similarity but also to relate latent cluster structure to observed variables. Researchers have employed both distribution-free methods \citep[e.g.,][]{mclachlan19829,mclachlan1988mixture} and distribution-based methods \citep[e.g.,][]{gordon1999classification,hastie2009elements} to explore such relationships. 

Two motivating applications illustrate these needs. In personalized marketing, firms increasingly build customer representations from implicit behavior---product views, purchase sequences, and RFM (Recency, Frequency, Monetary) summaries---and then segment customers to enable targeted treatment assignment. The default $k$-means algorithm \citep{macqueen1967some} effectively assumes spherical, equally sized clusters, an assumption often violated in heterogeneous retail data. In the present retail application, the objective is to describe purchasing profiles rather than to identify life stages or evaluate targeting policies. In healthcare, patient stratification with the Pima Indians Diabetes dataset offers a complementary view. While supervised models can achieve high accuracy with careful preprocessing and feature engineering like principal component analysis and random forests \citep{salih2024diabetic}, they depend on labels. The clinical application describes biomarker-age profiles. Whether these profiles improve prediction or intervention decisions is not evaluated in this study.
Given noise, missingness, and heterogeneity---ubiquitous in clinical data---elliptical clusters capture elongated, correlated patterns (e.g., glucose-BMI axes) that spherical models miss. Long-term findings in the Pima population underscore the stakes, with diabetes a key driver of kidney failure and other complications \citep{nelson2021pima}.

This article introduces a semiparametric clusterwise elliptical distribution (SCED) for continuous data and develops an efficient estimation-clustering procedure. The framework encompasses a broad class of elliptical distributions, including clusterwise multivariate normal distributions and clusterwise multivariate 
$t$ distributions with degrees of freedom satisfying the requisite finite-moment conditions, thereby offering greater modeling flexibility than conventional parametric mixtures. Leveraging a subjectwise representation of the latent structure, we construct a weighted least-squares objective with a separation penalty that estimates cluster-specific mean vectors and a cluster-invariant scatter matrix while producing cluster assignments. Under the localization and regularity conditions stated in Section \ref{sec:sec3_1}, the theoretical pseudo-maximum likelihood estimator attains the semiparametric efficiency bound, and the corresponding refined clustering rule asymptotically maximizes the probability of correct membership. 
The nonconvex program in the separation penalty estimation procedure is solved via the difference of convex functions programming (DCFP) \citep{an2005dc} combined with the alternating direction method of multipliers algorithm (ADMM) \citep{boyd2011distributed}, with a tailored initialization to improve stability. 
Relative to pairwise-fusion penalties \citep{chi2015splitting}, the approach reduces computational complexity \citep[cf.][]{tang2021individualized}. Under the localization conditions stated in Section \ref{sec:sec3_1.1}, the separation penalty criterion also admits a local oracle characterization.

This work builds on and connects two major strands of the literature. Distribution-free clustering includes hierarchical and non-hierarchical families. Hierarchical approaches comprise agglomerative procedures \citep{lance1966generalized,lance1967general, jambu1978classification} and divisive schemes \citep{macnaughton1964dissimilarity}. Among non-hierarchical methods, 
$k$-means \citep{macqueen1967some}---with antecedents in \cite{steinhaus1956division} and \cite{forgy1965cluster} and efficient updates by \cite{lloyd1982least}---is canonical. Invariance-based criteria leveraging within- and between-cluster covariance under linear transformations were proposed by \cite{friedman1967some}; \cite{marriott1971practical} suggested selecting 
the number of clusters by minimizing the determinant of the within-cluster covariance scaled by the square of the number of clusters. Comprehensive reviews appear in \cite{cormack1971review} and \cite{gordon:1987}. However, these procedures generally provide only partial guarantees; convex relaxations \citep{chi2015splitting} improve tractability but do not fully settle the statistical consistency of the resulting estimated cluster partitions. 
More generally, errors in an estimated cluster partition may induce bias in subsequent parameter estimation.
Distribution-based clustering identifies latent groups via mixtures of cluster-specific distributions. \cite{wolfe1965computer} studied univariate Gaussian mixtures; \cite{day1969estimating} extended to multivariate Gaussians and documented likelihood singularities when covariances are unrestricted. Robustness concerns motivated multivariate $t$ mixtures \citep{mclachlan1998robust}; see \cite{mclachlan2000finite} for a comprehensive treatment. Closer to our approach, subjectwise representations for likelihood-based estimation were developed by \cite{symons1981clustering} and extended via the classification EM of \cite{celeux1992classification}, with covariance reparameterizations (orientation, volume, and shape) by \cite{banfield1993model} and \cite{celeux1995gaussian}. Bayesian model selection for clustering is discussed in \citet{fraley1998many,fraley2002model}. 
Our semiparametric elliptical specification leaves the density generator unspecified while preserving the interpretation of cluster means and scatter. This flexibility accommodates elliptical distributions beyond the normal and \(t\) specifications. The simulations compare the proposed procedures with parametric benchmarks under the specified data-generating mechanisms.

Our contributions are fourfold: (1) a general SCED framework that leaves the elliptical density generator unspecified; (2) a local oracle characterization of the separation penalty criterion under the localization conditions in Section \ref{sec:sec3_1.1}, followed by semiparametric estimation and an asymptotically optimal clustering rule; (3) a scalable DCFP+ADMM implementation with a data-driven iterative Mahalanobis reassignment whose finite-sample performance is evaluated numerically; and (4) a semiparametric information criterion for data-adaptive determination of the number of clusters.
The simulations evaluate estimation, clustering, and cluster-number selection against the specified benchmarks. The applications illustrate fitted profiles and agreement between estimated partitions; they do not establish comparative accuracy relative to unknown latent memberships or population parameters.

The remainder of the article proceeds as follows. Section \ref{sec:sec2} formalizes the SCED model and its subjectwise representation. Section \ref{sec:sec3} develops a novel method for estimation and clustering. Section \ref{sec:sec4} details the initialization and DCFP+ADMM algorithm and outlines the full computational procedure. Section \ref{sec:sec3_1} establishes the asymptotic properties of the proposed estimators. Section \ref{sec:sec5} reports simulation evidence on finite-sample estimation and clustering accuracy. Section \ref{sec:sec6} presents applications in personalized marketing and patient stratification. Section \ref{sec:sec7} concludes with main findings and future directions. The Supplementary Material contains proofs of the main results and the technical lemma in Section \ref{sec:sec3_1}, together with additional simulation and application results.

\end{section}{}

\begin{section}{Clusterwise and Subjectwise Elliptical Distributions}
\label{sec:sec2}

Let $X$ be a $p \times 1$ vector of continuous variables with support $\mathcal{X}$, and let $C$ be a latent cluster variable taking values in support $\mathcal{C} = \{1, \dots, k \}$, where $k$ denotes the number of clusters. Define the cluster probabilities $\omega_{c}=P(C=c)>0$, $c\in\mathcal{C}$, with $\sum^{k}_{c=1}\omega_{c}=1$.

We consider a semiparametric clusterwise elliptical distribution (SCED) framework for the conditional density of $X$ given $C$:
\begin{align}
f(x\mid  c;\theta)=\frac{1}{|\Sigma |^{\frac{1}{2}}}f_{p}\big((x-\mu_{c})^{\top}\Sigma^{-1}(x-\mu_{c})\big),~ x\in\mathcal{X}, ~ c\in \mathcal{C},\label{2.1}
\end{align}
where $\theta$ is a column vector comprising the mean vectors $\mu_{1},\dots, \mu_{k}$ and the upper triangular entries of the scatter matrix $\Sigma$, and
$f_{p}$ is an unknown non-negative density generator satisfying the normalization condition
$2\pi^{p/2} \int^{\infty}_{0}r^{p-1}f_{p}(r^2)dr/\Gamma\big(p/2)=1$. The use of a common scatter matrix and density generator across clusters precludes the singularity problem in our pseudo-maximum marginal likelihood (PMML) estimation that is analogous to that arising in likelihood estimation for multivariate Gaussian mixtures \citep[cf.][]{day1969estimating}.
The model is identifiable up to a permutation of the cluster labels. Since $f_{p}$ is unspecified, we fix the first diagonal entry of $\Sigma$ at one for identification.
Given estimators of $f_{p}$ and $\Sigma$, the cluster-invariant covariance matrix of $X$, $\Sigma_{x}$, is estimated using the identity
$\Sigma_{x}=(\pi^{p/2}  \int^{\infty}_{0}r^{p+1}f_{p}(r^2)dr/\Gamma(p/2+1))\Sigma$.
Although the monotonicity assumption on $f_p$ has been studied in the literature \citep[cf.][]{BB2014}, this restriction offers no inferential advantage within the semiparametric framework under consideration. Our objective is to estimate the cluster-specific mean vectors, common within-cluster covariance matrix, cluster probabilities, and density generator based on the unsupervised data $\{X_{i}\}^{n}_{i=1}$.

Under the SCED, the observed variables $X$ can be expressed as
\begin{align}
X=\sum^{k}_{c=1}I(C=c)\mu_{c}+\Sigma^{\frac{1}{2}}U, \label{2.2}
\end{align}
where $I(C=c)$ is an indicator that equals 1 if $C=c$ and 0 otherwise, and $U$ is a $p\times 1$ spherical random vector. 
Because $f_{p}$ is unknown and $\{C_{i}\}^{n}_{i=1}$ are latent in the observed sample $\{X_{i}\}^{n}_{i=1}$, standard pseudo-maximum likelihood (PML) and PMML methods are not directly applicable to the joint density of $(X,C)$ and the $k$-component mixture density, respectively, given by
\begin{align}
f(x,c;\eta)=\omega_{c}f(x|c;\theta)\text{ and }f(x;\eta)=\sum^{k}_{c=1} f(x,c;\eta),\label{2.2.1}
\end{align}
where $\eta=(\theta^{\top},\omega^{\top})^{\top}$ and $\omega=(\omega_{1},\dots,\omega_{k-1})^{\top}$. 
We adopt a subjectwise representation of model (\ref{2.2}) in the first-phase estimation:
\begin{align}
X_{i}=\beta_{i}+\Sigma^{\frac{1}{2}}U_{i},~ i=1,\dots, n,\label{2.3}
\end{align}
where $\beta_{1},\dots,\beta_{n}\in\{\mu_{1},\dots,\mu_{k}\}$ and $\beta_{i}=\mu_{c}$ if and only if $C_{i}=c$. The model systematically incorporates the information of the realized latent cluster partition of the individual-level set $\{1,\dots,n\}$ into the estimation of $f_{p}$.

The second-phase pseudo-likelihood development uses a first-phase cluster partition to estimate the model parameters and construct a classifier that maximizes the probability of correct membership. The theoretical analysis in Section \ref{sec:sec3_1} characterizes this construction under the stated localization conditions, whereas the practical first-phase partition is initialized from iterative Mahalanobis reassignment
in the separation penalty procedure without asserting that iterative Mahalanobis reassignment starting values satisfy assumption {A2}.
The method reduces the original $p$-dimensional estimation problem to the one-dimensional problem. Specifically, we estimate the density generator $f_{p}$ through the density of the transformed variable
$Y=\sum^{k}_{c=1}I(C=c) Y_{c}$,
\begin{align}
g(y)=\frac{\pi^{\frac{p}{2}}}{\Gamma\big(\frac{p}{2}\big)}(\psi(y))^{\frac{p}{2}-1}\psi^{(1)}(y)
f_{p}\big(\psi(y)\big),\label{2.4}
\end{align}
where $y$ is a realization of $Y$, $Y_{c}=\Psi((X-\mu_c)^{\top}\Sigma^{-1}(X-\mu_c))$, and $\Psi(v)=-1+(1+v^{p/2})^{2/p}$, with 
$\psi^{(1)}$ denoting the derivative of the inverse function $\psi=\Psi^{-1}$. The conditional density in the SCED admits the representation
\begin{align}
f(x\mid c;\theta)=w(y_{c}) g(y_{c}),~x\in\mathcal{X},~c\in\mathcal{C},\text{ with }w(y_{c})=\frac{\Gamma\big(\frac{p}{2}\big)(\psi(y_{c}))^{1-\frac{p}{2}}}{|\pi\Sigma|^{\frac{1}{2}}\psi^{(1)}(y_{c})},\label{2.6}
\end{align}
where $y_{c}$ is a realization of $Y_{c}$. Notably,
the transformation $\Psi$ improves the numerical stability of estimating $f_{p}$, particularly when its argument is close to zero or becomes extremely large \citep[see][]{L2005}. Although the choice of transformation affects the leading constant in the asymptotic mean integrated squared
error of the estimator of $f_{p}$, no practical data-driven procedure is currently available for selecting an optimal transformation. More importantly, the asymptotic properties of the proposed estimator of $\theta$ are invariant to the  choice of $\Psi$.

\end{section}
\begin{section}{Estimation and Clustering}
\label{sec:sec3}
Let $\mathcal{G} = \{ \mathcal{G}_1, \dots, \mathcal{G}_k \}$ denote a partition of the individual-level index set $\{1, \dots, n\}$, with the realized latent cluster partition 
$\mathcal{G}^{\text{o}}$. For models (\ref{2.2}) and (\ref{2.3}), let
$\mu=(\mu^{\top}_{1},\dots,\mu^{\top}_{k})^{\top}$ and $\beta=(\beta^{\top}_{1},\dots,\beta^{\top}_{n})^{\top}$, where the true values $\mu^{\text{o}}$ and
$\beta^{\text{o}}$ satisfy $\beta^{\text{o}}_i=\mu^{\text{o}}_{c}$ for all $i\in \mathcal{G}^{\text{o}}_{c}$, 
$i=1,\dots, n$, $c\in\mathcal{C}$. The estimation procedure consists of two phases. Section \ref{sec:sec3.1} first obtains initial parameter estimators and an estimated cluster partition based on the subjectwise representation in (\ref{2.3}). Section \ref{sec:sec3.2} then develops pseudo-likelihood-based estimation procedures that incorporate these estimators and refines the estimated cluster partition via a Bayes assignment rule. Finally, Section \ref{sec:sec3.4} considers the selection of the number of clusters using the semiparametric information criterion.

\begin{subsection}{First-Phase: Separation-Penalty Estimation}
\label{sec:sec3.1}

Motivated by the subjectwise formulation adopted for estimation, the parameters $\beta^{\text{o}}$ and $\mu^{\text{o}}$ are estimated by minimizing the following weighted sum of squares objective with an $\ell_{1}$-based separation penalty:
\begin{align}
\textsc{SS}_{\text{sp}}(\beta,\mu;\lambda) 
= \frac{1}{2} \sum_{i=1}^n (X_{i}-\beta_{i})^{\top}W(X_{i}-\beta_{i}) + \lambda\sum^{n}_{i=1}\min_{c} \Big\|W^{\frac{1}{2}}(\beta_{i} -\mu_{c})\Big\|_{1},\label{3.1.1} 
\end{align}
where $W$ is a $p\times p$ positive-definite matrix, $\lambda>0$ is a shrinkage parameter, and $\|\cdot\|_{1}$ denotes the $\ell_{1}$-norm. The first term measures the weighted quadratic discrepancy between $\beta_i$ and $X_i$, while the separation penalty shrinks each $\beta_i$ toward its nearest cluster center $\mu_c$. Thus, $\lambda$ controls the degree of within-cluster homogeneity: larger values impose stronger shrinkage toward the cluster centers, whereas smaller values allow greater within-cluster variability. When $W=I_{p}$, the first term in (\ref{3.1.1}) coincides with the classical $k$-means objective. However, the estimated cluster partition induced by minimizing this criterion may not consistently recover $\mathcal{G}^{\text{o}}$. In our estimation, $W$ is chosen as a consistent estimator of $\Sigma^{-1}_{x}$. 

For a given $\lambda$, the separation-penalty (SP) estimator of $(\beta^{\text{o}},\mu^{\text{o}})$ is defined as
\vspace{-0.05in}
\begin{align}
(\hat{\beta}^{\textsc{sp}},\hat{\mu}^{\textsc{sp}}) \in \argmin_{\{\beta,\mu\}} \textsc{SS}_{\text{sp}}(\beta,\mu;\lambda).\label{3.1.2}
\end{align}
The estimator in (\ref{3.1.2}) is the unrestricted SP estimator used in the practical procedure. The local theoretical analysis in Section \ref{sec:sec3_1.1} considers a restriction of this criterion to the truth-centered parameter region specified in assumption {A2}. We do not assume that iterative Mahalanobis reassignment starting values are  known to lie in that region.
The estimated cluster partition $\widehat{\mathcal{G}}^{\,\textsc{sp}}= \{\widehat{\mathcal{G}}^{\,\textsc{sp}}_1, \dots, \widehat{\mathcal{G}}^{\,\textsc{sp}}_{k}\}$ of $\mathcal{G}^{\text{o}}$ is obtained via the rule:
\vspace{-0.05in}
\begin{align*}
\text{Assigning each observation $i$ to }  \widehat{\mathcal{G}}^{\,\textsc{sp}}_{c}\text{ if } 
c\in\argmin_{\{c_1\in\mathcal{C}\}} \big\| W^{\frac{1}{2}}(\hat{\beta}^{\textsc{sp}}_{i} - \hat{\mu}^{\textsc{sp}}_{c_1})\big\|_1, ~i=1,\dots, n.
\end{align*}
Moreover, the variance matrix $\Sigma_{x}$ is estimated by
\begin{align}
\widehat{\Sigma}^{\textsc{sp}}_{x}=\frac{1}{n}\sum^{n}_{i=1}(X_{i}-\hat{\beta}^{\textsc{sp}}_{i})(X_{i}-\hat{\beta}^{\textsc{sp}}_{i})^{\top}
\text{ or } \widehat{\Sigma}^{\textsc{sp}}_{x}=\frac{1}{n}\sum^{n}_{i=1}\sum^{k}_{c=1}I\big(i\in\widehat{\mathcal{G}}^{\,\textsc{sp}}_{c}\big)  (X_{i}-\hat{\mu}^{\textsc{sp}}_{c})(X_{i}-\hat{\mu}_{c}^{\textsc{sp}})^{\top}.\label{3.1.4}
\end{align}
Although the asymptotic behavior of 
$(\hat{\beta}^{\textsc{sp}},\hat{\mu}^{\textsc{sp}})$ is invariant to the choice of $W$,
empirical results show that setting $W=\big(\widehat{\Sigma}^{\textsc{sp}}_{x}\big)^{-1}$ in (\ref{3.1.1}) yields better agreement between $\widehat{\mathcal{G}}^{\,\textsc{sp}}$ and $\mathcal{G}^{\text{o}}$ than using $W=I_{p}$.

The tuning parameter $\lambda$ is selected from the shrinkage sequence $\lambda_{1} < \dots < \lambda_{J}$, where 
\vspace{-0.05in}
\begin{align*}
\lambda_{1} =n^{\alpha} \min_{\{i,j,c: \hat{\beta}^{(0)}_{i} \neq \hat{\mu}^{(0)}_{c}\}}\big |\big(W^{1/2}\big(\hat{\beta}^{(0)}_{i}-\hat{\mu}^{(0)}_{c}\big)\big)_{j} \big |\text{ and } 
\lambda_{J}=n^{\alpha}\max_{\{i,j,c\}}\big |\big(W^{1/2}\big(\hat{\beta}^{(0)}_{i}-\hat{\mu}^{(0)}_{c}\big)\big)_{j} \big |,
\end{align*} 
with $\alpha>0$ chosen sufficiently close to zero and $\big(\hat{\beta}^{(0)},\hat{\mu}^{(0)}\big)$ denoting the initial values for the SP estimation procedure. For $\lambda>\lambda_{J}$, the SP estimator remains unchanged from that at $\lambda=\lambda_{J}$. Within this admissible range, $\lambda$ is set to a minimizer of 
\vspace{-0.05in}
\begin{align*}
\sum_{i=1}^{n}\sum_{c=1}^{k} I\big(i\in\widehat{\mathcal{G}}^{\,\textsc{sp}}_{c}\big)(X_{i}-\hat{\mu}^{\textsc{sp}}_{c})^{\top}(X_{i}-\hat{\mu}^{\textsc{sp}}_{c}).
\end{align*}
An exploratory comparison based on ten replications per configuration compares the data-driven tuning rule with the first, fifth, and tenth candidates in the shrinkage sequence with $J=10$ (Table \ref{tab:lam_ari} and Figures \ref{fig:lam_res}--\ref{fig:lam_res_1}). The reported SP-stage results are broadly competitive, but the selected rule is not uniformly best. These alternatives are fixed positions in a data-dependent sequence, not fixed numerical values of $\lambda$ across samples. The comparison does not establish insensitivity to the candidate sequence or initialization, or general insensitivity of the subsequent optimal clustering, PML, PMML, and semiparametric information criterion procedure.

\begin{rem}
\label{rem1}
An alternative to the SP term in (\ref{3.1.1}) is a pairwise fusion penalty (FP) of the form $\sum_{i < j} p\big(\big\|W^{\frac{1}{2}} \big(\beta_{i} - \beta_{j}\big) \big\|; \lambda\big)$, where  $p(t; \lambda)$ is a non-decreasing and concave function and $\|\cdot\|$ denotes the Frobenius norm. Although numerical results indicate that the two methods exhibit comparable empirical performance, they differ substantially in computational complexity. The SP method has computational complexity of order $O(kn)$, whereas the FP method has computational complexity of order $O(n^2)$.\qed
\end{rem}

\end{subsection}

\begin{subsection}{Second-Phase: Semiparametric Efficient Estimation}
\label{sec:sec3.2}
 
For notational convenience, define 
$\widehat{Y}_{i}=\sum^{k}_{c=1}I\big(i\in \widehat{\mathcal{G}}^{\textsc{sp}}_{c}\big)Y_{ic}$, $i=1,\dots, n$.
To estimate the density $g(y)$ of $Y$, we adopt the reflection technique of \cite{cowling1996pseudodata} and \cite{zhang1999improved} to reduce
boundary bias near zero. The resulting boundary-corrected kernel is
\hspace{-0.05in}
 \begin{align*}
 K_{h}(v_{1},v_{2})=\frac{1}{h}K\Big(\frac{v_{1}-v_{2}}{h}\Big)+\frac{1}{h}K\Big(\frac{-v_{1}-v_{2}}{h}\Big),
\end{align*}
where $K$ is a symmetric, compactly supported, and twice continuously differentiable second-order kernel with its second derivative satisfying
$\int K^{(2)}(u)du=0$, and $h>0$ denotes the bandwidth.

Given a fixed value of $\theta$, we estimate $g(y)$ using the kernel estimator
\vspace{-0.05in}
\begin{align}
\hat{g}_{h}(y;\theta )=\frac{1}{n}\sum^{n}_{i=1}K_{h}\big(\widehat{Y}_{i},y\big), \label{3.2.1}
\end{align}
and construct a corresponding plug-in estimators of 
$f(x,c;\eta)$ and $f(x;\eta)$ in (\ref{2.2.1}) as
\vspace{-0.05in}
\begin{align}
\hat{f}_{h}(x, c;\eta)=\omega_{c} w(y_{c})\hat{g}_{h}(y_{c};\theta)\text{ and } \hat{f}_{h}(x;\eta)=\sum^{k}_{c=1}\hat{f}_{h}(x, c;\eta),~ c\in\mathcal{C}. \label{3.2.2}
\end{align}
By replacing $f(x,c;\eta)$ with $\hat{f}_{h}(x, c;\eta)$ and $I\big(i\in\mathcal{G}^{\text{o}}_{c}\big)$ with $I\big(i\in\widehat{\mathcal{G}}^{\textsc{sp}}_{c}\big)$ in the log-likelihood function
$\ell(\eta)=\sum^{n}_{i=1}\sum^{k}_{c=1}I\big(i\in\mathcal{G}^{\text{o}}_{c}\big)\log f(X_{i}, c;\eta)$,
we obtain the log-pseudo-likelihood function
\vspace{-0.05in}
\begin{align}
p\ell_{1}(\eta)=\sum^{n}_{i=1}\sum^{k}_{c=1}I\big(i\in\widehat{\mathcal{G}}^{\,\textsc{sp}}_{c}\big)\log  \hat{f}_{h}(X_{i}, c;\eta). \label{3.2.4}
\end{align}
The maximizer $\hat{\eta}^{\textsc{pml}}$ of $p\ell_{1}(\eta)$, for a fixed bandwidth $h$, is defined as the PML estimator of the true parameter vector 
$\eta^{\text{o}}=(\theta^{\text{o}\top},\omega^{\text{o}\top})^{\top}$. The PML estimator of $\omega^{\text{o}}$ is explicitly given by
$\hat{\omega}^{\textsc{pml}}=(|\widehat{\mathcal{G}}^{\,\textsc{sp}}_{1}|,\cdots, |\widehat{\mathcal{G}}^{\,\textsc{sp}}_{k-1}|)^{\top}/n$.
To initialize the maximization of $p\ell_{1}(\eta)$ with respect to $\theta$, we use the SP estimator $\hat{\mu}^{\textsc{sp}}$ and the SP scatter matrix estimator $\widehat{\Sigma}^{\textsc{sp}}=\widehat{\Sigma}^{\textsc{sp}}_{x}/\widehat{\sigma}^{2}_{x}$, where $\widehat{\sigma}^{2}_{x}$ denotes the first diagonal element of  
$\widehat{\Sigma}_{x}$. The bandwidth $h$ is selected as $\hat{h}= n^{3/80}\tilde{h}$, where $\tilde{h}$ minimizes the cross-validation criterion of \cite{bowman1984alternative},
\vspace{-0.075in}
\begin{align}
\textsc{CV}(h) = \frac{1}{n}\sum\limits^{n}_{i=1}\bigg[\int \big(\hat{g}^{\,-i}_{h}\big(y;\hat{\theta}^{\textsc{sp}}\,\big)\big)^{2}dy-2\hat{g}^{\,-i}_{h}\big(\widehat{Y}_{i};\hat{\theta}^{\textsc{sp}}\,\big)\bigg],\label{3.2.5}
\end{align}
with $\hat{g}^{\,-i}_{h}(y;\theta)$ denoting the leave-one-out version of the estimator in (\ref{3.2.1}). The conventional choice $\tilde{h}=O_{p}(n^{-1/5})$ is unsuitable in the current setting, as it violates the condition required for the $\sqrt{n}$-consistency of $\hat{\eta}^{\textsc{pml}}$.

In the context of unsupervised learning under parametric clusterwise elliptical distributions, the underlying density of $X$ is given by the $k$-component mixture density in (\ref{2.2.1}).
As an alternative to $p\ell_{1}(\eta)$ in (\ref{3.2.4}), we further consider the log-pseudo-marginal-likelihood function
\begin{align}
p\ell_{2}(\eta)= \sum^{n}_{i=1}\log \hat{f}_{h}(X_{i};\eta).\label{3.2.7}
\end{align}
The maximizer $\hat{\eta}^{\textsc{pmml}}$ of $p\ell_{2}(\eta)$, for a fixed bandwidth $h$, is defined as the PMML estimator of $\eta^{\text{o}}$. 
Our research findings underscore the importance of considering application contexts when choosing between $\hat{\eta}^{\textsc{pml}}$ and $\hat{\eta}^{\textsc{pmml}}$.

From the SCED and the corresponding cluster probabilities, the posterior probability that a data point with observation $x$ belongs to Cluster $c$ is given by
\begin{align*}
\pi(c\mid x;\eta^{\text{o}})=\frac{f(x,c;\eta^{\text{o}})}{f(x;\eta^{\text{o}})},~ c\in\mathcal{C}.
\end{align*} 
Let $C^{\text{o}}$ denote the cluster membership associated with a new observation $X_{0}$. The Bayes classifier $\widehat{C}$ assigns
$X_{0}$ to the cluster with the highest posterior probability:
\begin{align}
\widehat{C}=c_{0} \text{ if } \pi(c_{0}\mid X_{0};\eta^{\text{o}})=\max_{\{c\in\mathcal{C}\}} \pi(c\mid X_{0};\eta^{\text{o}}).\label{3.3.4}
\end{align}
For any classifier $\bar{C}$ based on $X_{0}$, the probability of correct membership is

\begin{align*}
P(\bar{C}=C^{\text{o}})=&E\Bigg[\sum^{k}_{c_{0}=1}E[I(\bar{C}=c_{0})I(C^{\text{o}}=c_{0})\mid X_{0}]\Bigg] 
= E\Bigg[\sum^{k}_{c_{0}=1}I(\bar{C}=c_{0})\pi(c_{0}\mid X_{0};\eta^{\text{o}})\Bigg].
\end{align*}
By construction, for any $c_{1},c_{2}\in\mathcal{C}$ with $\widehat{C}=c_{1}$ and $\bar{C}=c_{2}$,  
$\pi (c_{1} \mid X_{0};\eta^{\text{o}})-\pi(c_{2} \mid X_{0};\eta^{\text{o}})\geq 0$.
It follows that $\widehat{C}$ maximizes the probability of correct membership.

By replacing $(f(x,c;\eta^{\text{o}}),f(x;\eta^{\text{o}}))$ with $\big(\hat{f}_{\tilde{h}}(x,c;\hat{\eta}^{\textsc{pml}}),\hat{f}_{\tilde{h}}(x;\hat{\eta}^{\textsc{pml}})\big)$, we estimate $\pi(c\mid x;\eta^{\text{o}})$ by
\begin{align}
\hat{\pi}\big(c\mid x;\hat{\eta}^{\textsc{pml}}\big)=\frac{\hat{f}_{\tilde{h}}(x,c;\hat{\eta}^{\textsc{pml}})}{\hat{f}_{\tilde{h}}(x;\hat{\eta}^{\textsc{pml}})},~c\in\mathcal{C}.\label{3.3.2}
\end{align}\vspace{-0.05in}Using these posterior cluster probability estimators, an alternative to the SP-estimated cluster partition $\widehat{\mathcal{G}}^{\,\textsc{sp}}$ is $\widehat{\mathcal{G}}^{\,\textsc{oc}}$, obtained via the following optimal clustering (OC) rule:
\begin{align}
\text{Assigning observation $i$ to } \widehat{\mathcal{G}}^{\,\textsc{oc}}_{c_{0}}\text{ if } \hat{\pi}\big(c_{0}\mid X_{i};\hat{\eta}^{\textsc{pml}}\big)=\max_{\{c\in\mathcal{C}\}}\hat{\pi}\big(c\mid X_{i};\hat{\eta}^{\textsc{pml}}\big),\, i=1,\dots, n.\label{3.3.3}
\end{align}
In practice, $\hat{\eta}^{\textsc{pml}}$ in (\ref{3.3.2}) and (\ref{3.3.3}) may be replaced by $\hat{\eta}^{\textsc{pmml}}$. Replacing $\widehat{\mathcal{G}}^{\,\textsc{sp}}$ in (\ref{3.2.4}) and (\ref{3.2.7}) with $\widehat{\mathcal{G}}^{\,\textsc{oc}}$ yields alternative versions of the PML and PMML estimators. These estimators can then be used to iteratively update
$\widehat{\mathcal{G}}^{\,\textsc{oc}}$ until convergence.
\end{subsection}

\begin{subsection}{Selecting the Number of Clusters}
\label{sec:sec3.4}

We consider the $k$-component mixture density $p\ell(k)=\sum^{n}_{i=1}\log \hat{f}^{-i}_{\tilde{h}}(X_{i};\hat{\eta}^{\textsc{pml}})$ and select the number of clusters, $\hat{k}$,
as the minimizer of the semiparametric information criterion (SPIC):
\vspace{-0.075in}
\begin{align}
\text{SPIC}(k)=-\frac{1}{n}p\ell(k)+\frac{k\log n}{2n^{\frac{4}{5}}}+\frac{k(p+1) \log n}{2n}.\label{3.4.5}
\end{align}
The leading term in the asymptotic expansion of $(p\ell(k)-p\ell_{2}(\eta^{\text{o}}))$ suggests a penalty of order $k \log n/n^{4/5}$ for estimating the densities $(f(x\mid 1;\theta^{\text{o}}),\dots, f(x\mid k;\theta^{\text{o}}))$ given $\theta^{\text{o}}$. The subsequent $O_{p}(n^{-1})$ term yields an additional correction of order $(k(p+1)+p(p+1)/2-2)\log n/n$, which accounts for the estimation of $\eta^{\text{o}}$. When $n$ is small relative to $p$, the correction may substantially underestimate the number of clusters and is therefore omitted, as illustrated by $(n,p)=(250,15)$. An alternative to $p\ell(k)$ in (\ref{3.4.5}) is given by $\sum^{n}_{i=1}\sum^{k}_{c=1}I(i\in\widehat{\mathcal{G}}^{\,\textsc{oc}}_{c})\log  \hat{f}^{-i}_{\tilde{h}}(X_{i}\mid c;\hat{\eta}^{\textsc{pml}})$. Nonetheless, numerical results indicate that the associated SPIC tends to overestimate the number of clusters, even as the sample size grows.

\end{subsection}
 
\end{section}

\begin{section}{Estimation Implementation}
\label{sec:sec4}

This section introduces iterative Mahalanobis reassignment (IMR) as an initialization strategy to enhance convergence of the SP estimation procedure and presents its implementation algorithm. We then outline the computational procedure for the proposed method.

\begin{subsection}{Iterative Mahalanobis Reassignment}
\label{subsec:sec4.1}

To implement the SP estimation procedure, we first apply $k$-means to $\{X_{i}\}^{n}_{i=1}$ to obtain an initial partition 
$\widehat{\mathcal{G}}^{\,{\text{$k$m}},\ell} = \big\{\widehat{\mathcal{G}}^{\,{\text{$k$m}},\ell}_{1}, \dots, \widehat{\mathcal{G}}^{\,{\text{$k$m}},\ell}_{\ell} \big\}$ of $\{1, \dots, n \}$ into $\ell$ clusters for each $\ell=2,\dots, \bar{k}$. where $\bar{k}$ is a sufficiently large pre-specified upper bound; for example, we set $\lfloor n^{1/2} \rfloor$. 
Other distribution-free clustering methods could be employed for initialization. The reported experiments use the $k$-means-IMR initialization described here and do not establish that alternative initial partitions or repeated random starts would yield equivalent final estimates and partitions.

Given the estimated partition $\widehat{\mathcal{G}}^{\,{\text{$k$m}},\ell}$ for each $\ell$, we compute the estimator  
\begin{align}
\hat{\mu}^{{\text{$k$m}},\ell}=\Bigg(\frac{1}{\big|\widehat{\mathcal{G}}^{\,{\text{$k$m}},\ell}_{1}\big|}\sum^{n}_{i=1}I\big(i\in\widehat{\mathcal{G}}^{\,{\text{$k$m}},\ell}_{1}\big)X^{\top}_{i},\dots, \frac{1}{\big|\widehat{\mathcal{G}}^{\,{\text{$k$m}},\ell}_{\ell}\big|}\sum^{n}_{i=1}I\big(i\in\widehat{\mathcal{G}}^{\,{\text{$k$m}},\ell}_{\ell}\big)X^{\top}_{i} \Bigg)^{\top}\label{4.1.1}
\end{align} 
by minimizing the following within-cluster sum of squares with $\mathcal{G}^{\ell}=\widehat{\mathcal{G}}^{\,{\text{$k$m}},\ell}$:
\begin{align}
\textsc{SS}(\mu^{\ell};\mathcal{G}^{\ell})
= \frac{1}{2} \sum_{i=1}^n \sum^{\ell}_{c=1}I\big(i\in\mathcal{G}^{\ell}_{c}\big)\big(X_{i}-\mu^{\ell}_{c}\big)^{\top}\big(X_{i}-\mu^{\ell}_{c}\big). 
\label{4.1.2}
\end{align}\vspace{-0.05in}The resulting estimator of $\beta^{\text{o}}$ is then given by
\begin{align}
\hat{\beta}^{{\text{$k$m}},\ell}= \sum_{c = 1}^{\ell} \big(I\big(1 \in \widehat{\mathcal{G}}^{\,{\text{$k$m}},\ell}_{c}\big), \dots, I\big(n \in \widehat{\mathcal{G}}^{\,{\text{$k$m}},\ell}_{c}\big) \big)^{\top} \otimes \hat{\mu}^{{\text{$k$m}},\ell}_{c}.\label{4.1.3}
\end{align} 

Within the considered semiparametric framework, $\hat{\mu}^{{\text{$k$m}},\ell}$ and $\hat{\beta}^{{\text{$k$m}},\ell}$, $\ell=2,\dots,\bar{k}$, do not provide a direct advantage.
To improve the accuracy of $\widehat{\mathcal{G}}^{\,{\text{$k$m}},\ell}$ and its corresponding estimator $\hat{\mu}^{{\text{$k$m}},\ell}$, we refine the clustering of data points based on 
\vspace{-0.05in}
\begin{align}
\textsc{SD}_{ic}=\big(X_{i}-\hat{\mu}^{{\text{$k$m}},\ell}_{c}\big)^{\top}\big(\widehat{\Sigma}^{\,{\text{$k$m}},\ell}_x\big)^{-1}\big(X_{i}-\hat{\mu}^{{\text{$k$m}},\ell}_{c}\big), i=1,\dots,n, c=1,\dots,\ell, \label{4.1.4}
\end{align}
with $\widehat{\Sigma}^{\,{\text{$k$m}},\ell}_{x}=\sum^{n}_{i=1}\sum^{\ell}_{c=1} I\big(i \in \widehat{\mathcal{G}}^{\,{\text{$k$m}},\ell}_{c}\big)\big(X_{i}-\hat{\mu}^{{\text{$k$m}},\ell}_{c}\big)\big(X_{i}-\hat{\mu}^{{\text{$k$m}},\ell}_{c}\big)^{\top}/n$.
Each data point is then reassigned to the cluster attaining the minimum squared distance:
\vspace{-0.05in}
\begin{align}
\text{Assign observation $i$ to } \widehat{\mathcal{G}}^{\,\textsc{imr},\ell}_{c_{0}} \text{ if } \textsc{SD}_{ic_{0}}=\min_{\{1\leq c \leq \ell \}} \,\textsc{SD}_{ic}, i=1,\dots, n.
\label{4.1.6}
\end{align}
Replacing $\widehat{\mathcal{G}}^{\,{\text{$k$m}},\ell}$ with
 $\widehat{\mathcal{G}}^{\,\textsc{imr},\ell}$ in (\ref{4.1.1}) and (\ref{4.1.3}) yields the updated estimators $\hat{\mu}^{\textsc{imr},\ell}$ and $\hat{\beta}^{\textsc{imr},\ell}$. Similarly, replacing $\widehat{\mathcal{G}}^{\,{\text{$k$m}},\ell}$ and $\hat{\mu}^{\,{\text{$k$m}},\ell}$ with $\widehat{\mathcal{G}}^{\,\textsc{imr},\ell}$ and $\hat{\mu}^{\,\textsc{imr},\ell}$, respectively, in (\ref{4.1.4}) yields the updated estimator $\widehat{\Sigma}^{\,\textsc{imr},\ell}_{x}$. 
This refinement process is iterated by alternating between (\ref{4.1.4}) and (\ref{4.1.6}) until either a local minimum of $\textsc{SS}(\mu^{\ell}; \widehat{\mathcal{G}}^{\,\textsc{imr},\ell})$ is attained or a specified convergence criterion is met. 
\end{subsection}

\begin{subsection}{Computational Algorithm for the SP Estimation}
\label{subsec:sec4.2}

For a fixed $\lambda\in\{\lambda_{1}, \dots ,\lambda_{J}\}$, the minimization of $\textsc{SS}_{\text{sp}}(\beta,\mu;\lambda)$ is carried out via a block coordinate descent algorithm.
Let $(\hat{\beta}^{(m)},\hat{\mu}^{(m)})$ denote 
the solution obtained at iteration $m\geq 0$. The algorithm is initialized with $\hat{\beta}^{(0)}$ selected from the IMR starting values $\{\hat{\beta}^{\textsc{imr},\ell}: \ell=k,\dots,\bar{k}\}$ and $\hat{\mu}^{(0)}=\hat{\mu}^{\textsc{imr},k}$. 

To minimize $\textsc{SS}_{\text{sp}}(\beta,\hat{\mu}^{(m)};\lambda)$ with respect to $\beta$, we first express the objective function as the difference of two convex functions in $\beta$:
\begin{align*}
\bigg[\frac{1}{2} \sum_{i=1}^n (X_{i}-\beta_{i})^{\top}W(X_{i}-\beta_{i}) + \lambda\sum^{n}_{i=1}\sum_{c=1}^k \big\|W^{\frac{1}{2}}(\beta_i-\hat{\mu}^{(m)}_c)\big\|_{1}\bigg]-\lambda\sum^{n}_{i=1}\,\max_c \,\sum_{\{c_1\neq c\}} \big\|W^{\frac{1}{2}}(\beta_i-\hat{\mu}^{(m)}_{c_1})\big\|_{1}.
\end{align*}
Applying the DCFP of \cite{an2005dc}, the update $\hat{\beta}^{(m+1)}$ is obtained by minimizing the following strictly convex surrogate function:
\begin{align}
\textsc{SS}\big(\beta;\hat{\beta}^{(m)},\hat{\mu}^{(m)}\big)= \textsc{SS}_{w}(\beta)  + \lambda\sum^{n}_{i=1}\sum_{c=1}^k \big\|W^{\frac{1}{2}}(\beta_i-\hat{\mu}^{(m)}_c)\big\|_{1} 
 - \lambda\, \partial_{\beta}^{\top} \textsc{S}\big( \hat{\beta}^{(m)}, \hat{\mu}^{(m)} \big) \big(\beta - \hat{\beta}^{(m)}\big),\label{DCFP1}
\end{align}
where
\begin{align*}
 &\textsc{SS}_{w}(\beta)= \frac{1}{2} \sum_{i=1}^n (X_{i}-\beta_{i})^{\top}W(X_{i}-\beta_{i}),\, \textsc{S}( \beta, \mu) = \sum_{i=1}^n\,\max_{c} \,\sum_{\{c_1 \neq c\} } \big\| W^{\frac{1}{2}}\big( \beta_i - \mu_{c_1} \big) \big\|_{1},\text{ and}\\
 &\partial_{\beta_{i}} \textsc{S} (\beta, \mu ) = \sum_{\big\{c: c \neq \argmax\limits_{c_1} \sum\limits_{ \{ c_2: c_2 \neq c_1 \} } \big\| W^{\frac{1}{2}}( \beta_i - \mu_{c_2} ) \big\|_{1} \big\} } W^{\frac{1}{2}}  \big( \text{sgn}\big(\big(W^{\frac{1}{2}}(\beta_{i} - \mu_{c})\big)_{j}\big) \big), ~i = 1,\dots, n.
\end{align*}
To facilitate computation of $\hat{\beta}^{(m+1)}$, we introduce an auxiliary parameter vector $\delta=(\delta_{11},\dots,\delta_{1k},\dots,$ $\delta_{n1},\dots,\delta_{nk})^{\top}$, and reformulate the minimization problem in \eqref{DCFP1} as
\vspace{-0.05in}
\begin{align*}
 &\text{Minimize }\textsc{SS}\big(\beta,\delta;\hat{\beta}^{(m)},\hat{\mu}^{(m)}\big)= \big[\textsc{SS}_{w}(\beta) - \lambda\, \partial_{\beta}^{\top} \textsc{S}\big( \hat{\beta}^{(m)}, \hat{\mu}^{(m)} \big) \big(\beta - \hat{\beta}^{(m)}\big) \big]+ \lambda\sum^{n}_{i=1}\sum_{c=1}^k \big\|\delta_{ic}\big\|_{1}\nonumber \\
&\text{subject to } W^{\frac{1}{2}}\big(\beta_{i}-\hat{\mu}^{(m)}_c \big) = \delta_{ic},~ i = 1, \dots,n, ~c \in \mathcal{C}.
\end{align*}\vspace{-0.15in}Since the objective function is separable in $\beta$ and $\delta$,
we solve the constrained optimization problem using the ADMM.
The corresponding augmented Lagrangian function is 
\begin{align*}
\textsc{SS}\big(\beta,\delta;\hat{\beta}^{(m)},\hat{\mu}^{(m)}\big)+ \sum_{i = 1}^{n} \sum_{c = 1}^{k} \nu_{ic}^{\top} \Big(W^{\frac{1}{2}}(\beta_i -\hat{\mu}^{(m)}_c) - \delta_{ic}\Big) + \frac{\nu_0}{2} \sum_{i = 1}^{n} \sum_{c = 1}^{k} \Big\| W^{\frac{1}{2}}(\beta_i -\hat{\mu}^{(m)}_c )- \delta_{ic} \Big\|^2,
\end{align*}
where $\nu_{ic}, i = 1, \dots,n, c \in \mathcal{C}$, are Lagrange multipliers, $\nu_0>0$ is the augmented Lagrangian penalty parameter (set to 1 throughout). Given $\hat{\beta}^{(m+1)}$, the update for $\mu$ is obtained by minimizing $\textsc{SS}_{\text{sp}}(\hat{\beta}^{(m+1)},\mu;\lambda)$ with respect to $\mu$, yielding
\begin{align}
\hat{\mu}^{(m+1)}_{c}= &~ W^{-\frac{1}{2}}  \Bigg( \median_{\big\{i: c \in \argmin\limits_{c_1} \big\| W^{\frac{1}{2}} \big(\hat{\beta}^{(m+1)}_i - \hat{\mu}^{(m)}_{c_1}\big)  \big\|_1 \big\} } \Big(W^{\frac{1}{2}} \hat{\beta}^{(m+1)}_{i}\Big)_j \Bigg),~ c \in \mathcal{C}.\label{4.2.9}
\end{align}

The SP estimator is obtained by alternating between the following two updates until a prescribed convergence criterion is satisfied.
\begin{itemize}
\item \textbf{($\boldsymbol{\beta}$-update)} Given $\hat{\mu}^{(m)}$, obtain $\hat{\beta}^{(m+1)}$ through the following inner iterations:
 \begin{align}
\beta_{i}^{(s+1)} =&~\frac{1}{k+1} \bigg[ X_i + \sum_{c=1}^{k} \Big(\hat{\mu}_c^{(m)} + W^{-\frac{1}{2}} \delta_{i c}^{(s)} - W^{-\frac{1}{2}}\nu_{i c}^{(s)}\Big)  + W^{-1} \lambda\,\partial_{\beta}^{\top} \textsc{S}\big(\hat{\beta}^{(m)}, \hat{\mu}^{(m)}\big)\bigg], \label{4.2.6}\\
\delta_{i c}^{(s+1)} = & ~\text{diag}\Bigg(\max\Bigg\{0, 1 -\frac{\lambda}{\Big|\Big(W^{\frac{1}{2}}\Big(\beta_{i}^{(s+1)} - \hat{\mu}_{c}^{(m)}\Big) +  \nu_{i c}^{(s)}\Big)_1\Big| } \Bigg\},\dots,\max\Bigg\{0,  \nonumber \\
&~ 1 -\frac{\lambda}{\Big|\Big(W^{\frac{1}{2}}\Big(\beta_{i}^{(s+1)} - \hat{\mu}_{c}^{(m)}\Big) +  \nu_{i c}^{(s)}\Big)_p\Big| } \Bigg\} \Bigg)  \Big( W^{\frac{1}{2}}\Big(\beta_{i}^{(s+1)} - \hat{\mu}_{c}^{(m)} \Big)+  \nu_{i c}^{(s)}\Big), \label{4.2.7}\\
  \nu_{i c}^{(s+1)} = &~\nu_{i c}^{(s)} +  \Big(W^{\frac{1}{2}}\Big(\beta_{i}^{(s+1)} - \hat{\mu}_{c}^{(m)} \Big) - \delta_{i c}^{(s+1)}\Big), ~ s\geq 0,\label{4.2.8} 
\end{align}
where $\delta_{i c}^{(0)} = W^{\frac{1}{2}} \Big(\hat{\beta}_{i}^{(m)} - \hat{\mu}_{c}^{(m)}\Big)$ and $\nu_{i c}^{(0)} = 0$, $i=1,\dots, n$, $c\in\mathcal{C}$.
\item \textbf{($\boldsymbol{\mu}$-update)} Given $\hat{\beta}^{(m+1)}$, obtain $\hat{\mu}^{(m+1)}$ from (\ref{4.2.9}).
\end{itemize}
Convergence of the inner iterations in the $\beta$-update is ensured by the arguments in \cite{boyd2011distributed} and \cite{an2018dc}, while convergence of the overall alternating procedure to a stationary point follows from \cite{attouch2013convergence}.
This convergence result concerns the numerical iterates and their convergence to a stationary point. It does not assert that the limiting stationary point is the global minimizer in (\ref{3.1.2}), that it lies in $\mathcal{B}\times\mathcal{U}$, or that using the IMR starting values satisfies assumption {A2}. The local oracle result in Section \ref{sec:sec3_1.1} is a separate theoretical statement.
The following pseudocode summarizes the complete computational procedure.
\RestyleAlgo{ruled}

\begin{algorithm}
\caption{Pseudocode of the computational algorithm for the SP estimation}
\begin{minipage}{\linewidth}
Initialize $\hat{\beta}^{(0)}$ and $\hat{\mu}^{(0)}$\;
\Begin{
 \For{$m = 0, 1, 2, \dots$}{
 Set $\beta^{(0)} \gets \hat{\beta}^{(m)}$, $\delta^{(0)}_{ic} \gets W^{\frac{1}{2}}\big(\hat{\beta}_{i}^{(m)} - \hat{\mu}_{c}^{(m)}\big)$, $\nu^{(0)} \gets 0$, and $\varepsilon\gets \text{ predefined tolerance}$\;
  \For{$s = 0, 1, 2, \dots$}{
  Compute $\beta^{(s+1)}$ using (\ref{4.2.6})\;
  Compute $\delta^{(s+1)}$ using (\ref{4.2.7})\;
  Compute $\nu^{(s+1)}$ using (\ref{4.2.8})\;
          \eIf{$\frac{1}{n}\sum_{i=1}^n \sum_{c=1}^k \Big\| \beta^{(s+1)}_i - \hat{\mu}^{(m)}_{c} - \delta^{(s+1)}_{ic} \Big\| < \varepsilon$}{
            Set $\hat{\beta}^{(m+1)} \leftarrow \beta^{(s+1)}$ and exit the loop\;
            }{$s = s+1$\;
            }{
            }
        }
        Compute $\hat{\mu}^{(m+1)}$ using (\ref{4.2.9})\;
        \eIf{$\textsc{SS}_{\mathrm{sp}}\big(\hat{\beta}^{(m+1)},\hat{\mu}^{(m+1)};\lambda \big) > \textsc{SS}_{\mathrm{sp}}\big(\hat{\beta}^{(m)},\hat{\mu}^{(m)};\lambda \big)$,}{
            Set $\hat{\beta}^{\textsc{sp}} \leftarrow \hat{\beta}^{(m)}$ and $\hat{\mu}^{\textsc{sp}} \gets \hat{\mu}^{(m)}$; terminate\;
            }{$m = m+1$\;
            }{
            }
   }
   }
   \end{minipage}
\end{algorithm}

\end{subsection}

\begin{subsection}{Computational Procedure}
\label{subsec:sec4.3}

The proposed method is described with the following flowchart:

\begin{figure*}[htp]
\centering
 \includegraphics[width=15cm,height=13.9cm]{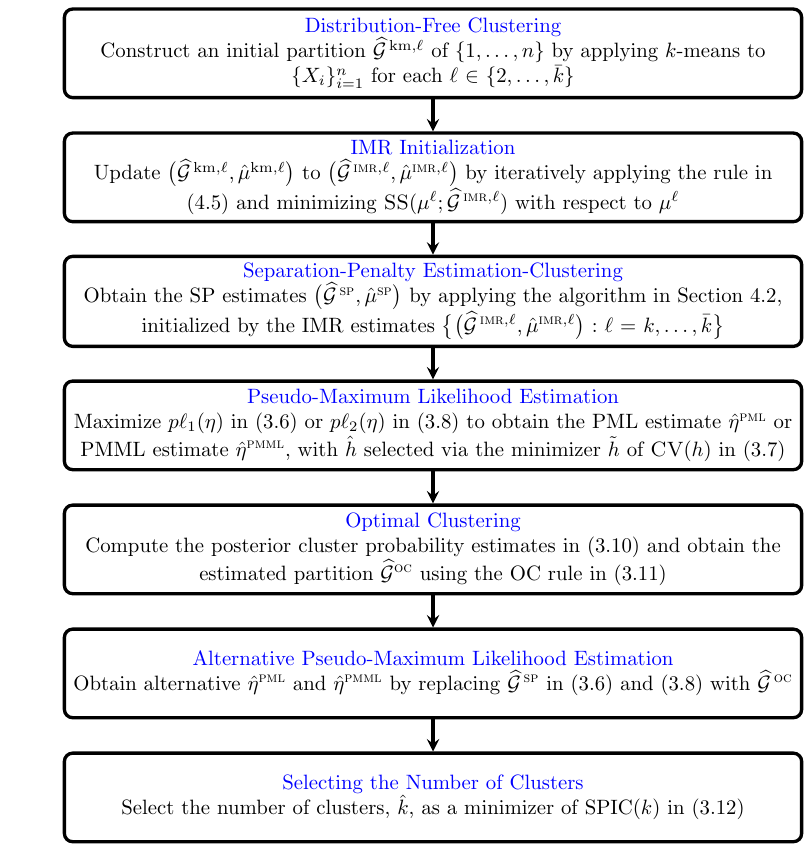}
 \label{fig:flowchart} 
\end{figure*}

\end{subsection}

\end{section}

\begin{section}{Asymptotic Properties}
\label{sec:sec3_1}

In this section, we establish the oracle property of the localized SP estimator, the consistency and asymptotic normality of the PML and PMML estimators, and the consistency of the estimator of the number of clusters.

\begin{subsection}{Local Oracle Property of the Localized SP Estimator}
\label{sec:sec3_1.1}

Let $\mathcal{B}$ and $\mathcal{U}$ denote the compact parameter spaces of $\beta$ and $\mu$, respectively, with the true parameters $\beta^{\text{o}}$ and $\mu^{\text{o}}$ lying in their interiors, and let $D_{\mu}=\min_{\{c_{1} \neq c_{2}\}} \|W^{1/2}(\mu^{\text{o}}_{c_{1}} - \mu^{\text{o}}_{c_{2}})\|_{1}$. The oracle estimators of $\mu^{\text{o}}$ and $\beta^{\text{o}}$ are defined as $\tilde{\mu}^{\text{or}}$, a minimizer of $ \sum_{i=1}^n \sum^{k}_{c=1}I(C_{i}=c)(X_{i}-\mu_{c})^{\top}(X_{i}-\mu_{c})$, and $\tilde{\beta}^{\text{or}} = \sum_{c = 1}^{k} (I(1 \in \mathcal{G}^{\text{o}}_{c}), \dots, I(n \in \mathcal{G}^{\text{o}}_{c}))^{\top} \otimes \tilde{\mu}^{\text{or}}_{c}$, respectively, where $\otimes$ denotes the Kronecker product.
The following assumptions are imposed to establish the properties of the localized SP estimator:
\begin{itemize}
 \item[A1.] $\E\big[U_{1}^{2} \big]<\infty$.
 \item[A2.] $\mathcal{B}=\{\beta:\max_{\{1\leq i\leq n\}} \|W^{1/2}(\beta_{i}-\beta^{\text{o}}_{i} )\|_{1} \leq D_{\mu}/4\}$ and $\mathcal{U}=\{\mu:\max_{\{c\in\mathcal{C}\}}\|W^{1/2}(\mu_{c}-\mu^{\text{o}}_{c})\|_{1}\leq D_{\mu}/4\}$.
  \item[A3.] $\lambda \geq  n^{\alpha}$ for each fixed $\alpha>0$.
\end{itemize}
\noindent Assumption {A1} imposes the moment condition used in the oracle analysis. Assumption {A2} is a local theoretical condition: it restricts the candidate subject parameters and cluster centers to neighborhoods of their corresponding true values. Because these neighborhoods depend on unknown true parameters, assumption {A2} is not a condition that is imposed or verified in the practical implementation. Assumption {A3} requires the separation penalty to be sufficiently strong within this localized parameter region. Together, assumptions {A2} and {A3} yield the objective comparison used in the local oracle result below. No claim is made that IMR starting values necessarily satisfy assumption {A2}.

For the local theoretical analysis, define
$(\tilde{\beta}^{\textsc{sp}},\tilde{\mu}^{\textsc{sp}}) \in \argmin_{(\beta,\mu)\in\mathcal{B}\times\mathcal{U}} \text{SS}_{\text{sp}}(\beta,\mu;\lambda),$
and let $\widetilde{\mathcal{G}}^{\textsc{sp}}$ denote the corresponding induced partition. This localized estimator is a theoretical object distinct from the unrestricted estimator $(\hat{\beta}^{\textsc{sp}},\hat{\mu}^{\textsc{sp}})$ defined in (\ref{3.1.2}) and computed using IMR starting values.
The following theorem establishes the local oracle property of the SP criterion.
\begin{thm} 
\label{Thm3.1}
Under assumptions {A1}--{A3},
\begin{align*}
P\big(\big(\tilde{\beta}^{\textsc{sp}},\tilde{\mu}^{\textsc{sp}}\big)= \big(\tilde{\beta}^{\text{or}}, \tilde{\mu}^{\text{or}}\big)\big) {\longrightarrow} 1\text{ as } n {\longrightarrow} \infty,
\end{align*} 
where $(\tilde{\beta}^{\textsc{sp}},\tilde{\mu}^{\textsc{sp}})$ denotes the localized SP estimator defined over $\mathcal{B}\times\mathcal{U}$.
\begin{proof}
See Appendix \ref{Thm3.1pf}.
\end{proof}
\end{thm}
Since $(\tilde{\beta}^{\textsc{sp}},\tilde{\mu}^{\textsc{sp}})=(\tilde{\beta}^{\text{or}},\tilde{\mu}^{\text{or}})$ implies $\widetilde{\mathcal{G}}^{\textsc{sp}}=\mathcal{G}^{\text{o}}$ under the label alignment used in assumption {A2}, Theorem \ref{Thm3.1} yields
\begin{align} 
P(\widetilde{\mathcal{G}}^{\textsc{sp}}=\mathcal{G}^{\text{o}})\longrightarrow 1 \text{ as } n\longrightarrow \infty. \label{3.1.7_0}
\end{align}
Furthermore, the asymptotic normality of $\tilde{\mu}^{\text{or}}$ yields
\begin{align}
\sqrt{n}\big(\tilde{\mu}^{\textsc{sp}} - \mu^{\text{o}}\big)\stackrel{d}{\longrightarrow} N\Big(0,\text{diag}\Big(\frac{1}{\omega_{1}}\Sigma_{x},\dots, \frac{1}{\omega_{k}}\Sigma_{x}\Big)\Big) \text{ as } n \longrightarrow \infty. \label{3.1.7}
\end{align} 

\begin{rem}
\label{rem2}
Theorem \ref{Thm3.1} is a local oracle result under the truth-centered localization in assumption {A2}. Assumption {A2} is introduced for theoretical analysis and is neither observable from the data nor imposed in the practical computation. In particular, we do not claim that using IMR starting values satisfies assumption {A2} with probability converging to one. The unrestricted SP estimator in (\ref{3.1.2}) is initialized from IMR, and its finite-sample performance is evaluated separately in Section \ref{sec:sec5}.

The exact latent-partition recovery in Theorem \ref{Thm3.1} should likewise not be interpreted as a general perfect-classification result for continuous mixtures with fixed nonvanishing overlap. Under such overlap, a positive per-observation Bayes classification error generally precludes recovery of every realized latent label with probability converging to one. Within the localized problem, assumption {A2} places each candidate subject parameter in the neighborhood of its corresponding true center, while assumption {A3} makes the separation penalty sufficiently strong for the oracle configuration to dominate local alternatives. Thus, assumption {A2} supplies the local separation structure and assumption {A3} supplies the required penalty strength.

Theorem \ref{Thm3.1} also does not assert that IMR starting values enter $\mathcal{B}\times\mathcal{U}$ or that an algorithm initialized in this region necessarily remains there and converges to the oracle solution. Those are separate algorithmic questions.
\qed
\end{rem}

\end{subsection}

\begin{subsection}{Asymptotic Properties of the PML and PMML Estimators}
\label{sec:sec3_1.2}

Let $\mathcal{H}$ denote the compact parameter space, with the true value $\eta^{\text{o}}$ lying in its interior, and let $\mathcal{X}_{n}$ denote $\{X_1,\dots,X_n\}$. Define
\vspace{-0.1in}
\begin{align*}
 M_{q,r}^{[m]}(X,x,c)=&~\big(\partial_{\eta}\big(Y+(-1)^r y_{c}\big)\big)^{\otimes q} I(q=m)+\partial_{\eta}^2\big(Y+(-1)^r y_{c}\big) I((q,m)=(1,2)),\\
 f^{[m]}(x,c;\eta)=&~\sum_{s=0}^m \binom{m}{s}\big(\partial_{\eta}^s \,(\omega_c w(y_c))\big) \sum_{q=0}^{m-s} \sum_{r=0}^1 \partial_{y_c}^q \Big(\E\Big[M_{q,r}^{[m-s]}(X,x,c)\Big|Y=y_c\Big]g(y_c;\theta)\Big)\text{ and}\\
 f^{[m]}(x;\eta)=&~\sum_{c=1}^k f^{[m]}(x,c;\eta), ~ q=0,\dots,m,~r=0,1,~ m=0,1,2.
\end{align*}
The following additional assumptions are imposed for the properties of $\tilde{\eta}^{\textsc{pml}}$ and $\tilde{\eta}^{\textsc{pmml}}$ based on $\widetilde{\mathcal{G}}^{\textsc{sp}}$:
\begin{itemize}
\item [A1$^{\prime}$.]  $\E\big[U_{1}^{8} \big]<\infty$.
\item[A4.] $h\in \big(h_{1}n^{-1/5},h_{2} n^{-1/8}\big)$ for some constants $h_{1}, h_{2}>0$. 
\item[A5.] $\partial_{y}^m (\E[M_{q,r}^{m}(X,x,c) \mid Y=y]g(y;\theta))$ is bounded for $x\in \mathcal{X}_{n}$ and Lipschitz continuous in $(y,\eta)$ for each $q=0,\dots,m$, $r=0,1$, and $m=0,1,2$, with Lipschitz constants independent of $x$.
\item[A6.] $\inf_{\{\eta \in \mathcal{H}\}} f^{[0]}(x,c;\eta)$ is far away from zero for $(x,c) \in \mathcal{X}_{n} \times \mathcal{C}$.
\item[A7.] $V_{1}=\var\big(\sum^{k}_{c=1}I(C=c)f^{[1]}(X,c;\eta^{\text{o}})/f(X,c;\eta^{\text{o}})\big)$ is positive definite.
\item[A8.] $V_{2}=\var( f^{[1]}(X;\eta^{\text{o}})/f(X;\eta^{\text{o}}))$ is positive definite.
\end{itemize}
Assumption {A1$^{\prime}$} imposes a stronger moment condition than {A1}, which is required to establish the uniform consistency of the second-order partial derivatives of the conditional density estimator with respect to the scalar parameters. In parametric clusterwise elliptical distributions, asymptotic normality of the corresponding maximum likelihood estimators requires only a finite fourth moment.
Assumption {A4} specifies the bandwidth rate required for the $\sqrt{n}$-consistency of $\tilde{\eta}^{\textsc{pml}}$ and $\tilde{\eta}^{\textsc{pmml}}$. Assumption {A5} imposes smoothness conditions that guarantee the uniform consistency of the conditional density estimator and its partial derivatives with respect to the model parameters. The uniform positivity condition over the observed sample $\mathcal{X}_{n}$ in assumption {A6} is a mild regularity condition for establishing the consistency and asymptotic normality of $\tilde{\eta}^{\textsc{pml}}$ and $\tilde{\eta}^{\textsc{pmml}}$. Assumptions {A7} and {A8} provide additional conditions required for asymptotic normality.

For the asymptotic results in this subsection, whenever the oracle recovery result from Theorem \ref{Thm3.1} is invoked, the first-phase partition refers to the localized theoretical partition $\widetilde{\mathcal{G}}^{\textsc{sp}}$. The practical estimators initialized from IMR are computational counterparts whose finite-sample performance is assessed numerically.
The following theorem establishes the consistency and asymptotic normality of $\tilde{\eta}^{\textsc{pml}}$ and $\tilde{\eta}^{\textsc{pmml}}$.
\vspace{-0.05in}
\begin{thm} 
\label{Thm3.2}
Under assumptions {A1$^{\prime}$} and {A2}--{A8},
\begin{align*}
\tilde{\eta}^{\textsc{pml}}\stackrel{p}{\longrightarrow}\eta^{\text{o}}, ~\tilde{\eta}^{\textsc{pmml}}\stackrel{p}{\longrightarrow}\eta^{\text{o}}, ~
\sqrt{n}\big(\tilde{\eta}^{\textsc{pml}}-\eta^{\text{o}}\big)\stackrel{d}{\longrightarrow} N\big(0,V_{1}^{-1}\big), \text{ and } \sqrt{n}\big(\tilde{\eta}^{\textsc{pmml}}-\eta^{\text{o}}\big)\stackrel{d}{\longrightarrow} N\big(0,V_{2}^{-1}\big)
\end{align*} 
as $n\longrightarrow\infty$ with $V_{1}^{-1}$ attaining the semiparametric efficiency bound.
\vspace{-0.1in}
\begin{proof}
See Appendix \ref{Thm3.2-3.pf}.
\end{proof} 
\end{thm}
\noindent As shown in \cite{chiang2024general}, $V_{1}$ and $V_{2}$ are invariant to the specification of the transformation $\Psi$. Following the approach of \cite{BKRW1998Efficient}, we show in Appendix \ref{speff.pf} that $V_{1}^{-1}$ is the semiparametric efficiency bound for the present model.

\end{subsection}

\begin{subsection}{Consistency of the Estimator of the Number of Clusters}
\label{sec:sec3_1.3}

Assume that $\lim_{n\rightarrow\infty}P\big(\widetilde{\mathcal{G}}^{\,\textsc{sp}}=\mathcal{G}^{\text{o}}\big)=1$ for every fixed $k \geq 1$.
Let $k^*$ denote the true number of clusters, and define $\mathcal{G}^{*\text{o}} = \{\mathcal{G}^{*\text{o}}_{1}, \dots, \mathcal{G}^{*\text{o}}_{k^{*}}\}$ and 
$\eta^{*\text{o}}=(\eta^{*\text{o}}_{1}, \dots, \eta^{*\text{o}}_{k^{*}})^{\top}$, where $\mathcal{G}^{*\text{o}}_{\ell}=\mathcal{G}^{\text{o}}_{\ell}$ and $\eta^{*\text{o}}_{\ell}=\eta^{\text{o}}_{\ell}$ for $\ell=1,\dots,k^{*}$.
Let $C^{*}$ be the latent cluster variable associated with $\mathcal{G}^{*}$, supported on
$\mathcal{C}^{*} = \{1, \dots, k^{*}\}$.
Furthermore, let $\hat{\eta}^{*\textsc{pml}}$ denote the PML estimator of 
$\eta^{*\text{o}}$, and let $\widehat{\mathcal{G}}^{*\textsc{sp}}$ denote
the SP-estimated cluster partition of ${\mathcal{G}}^{*\text{o}}$.

Define $\mathcal{C}_{\mathcal{G}}=\big\{ \mathcal{G}^{\text{o}}: \mathcal{G}_{\ell}^{*\text{o}} = \bigcup_{\{\ell_1:\mathcal{G}^{\text{o}}_{\ell_1} \subseteq \mathcal{G}_{\ell}^{*\text{o}},\theta^{\text{o}}_{\ell_1}=\theta^{*\text{o}}_{\ell}\}}\mathcal{G}^{\text{o}}_{\ell_1}~\forall \ell\in\mathcal{C}^{*} \big\}$. Building on the arguments of Theorem \ref{Thm3.2} and the associated technical lemma, we show that
\begin{align}
&\frac{1}{n}  p\ell(k^{*})- \frac{1}{n} p\ell(k)= \left\{
\begin{array}{ll}
O_p\Big(\frac{1}{n^{\frac{4}{5}}}\Big)+O_{p}\big(\frac{1}{n}\big)
& \mathcal{G}^{\text{o}} \in \mathcal{C}_{\mathcal{G}}, \\
b^{2}(k) +o_p(1)
& \text{otherwise,}
\end{array}\right. \label{3.4.3}
\end{align}
where $b^{2}(k)>0$ is a constant. The following theorem establishes the consistency of $\hat{k}$ to $k^{*}$.
\begin{thm}
 \label{Thm3.4}
Under assumptions {A1$^{\prime}$} and {A2}--{A8}, 
\begin{align*}
 \hat{k}\stackrel{p}{\longrightarrow} k^{*} \text{ as } n {\longrightarrow} \infty. 
\end{align*}
\begin{proof}
See Appendix \ref{Thm3.4.pf}.
\end{proof}
 \end{thm}
\end{subsection}

\end{section}

\begin{section}{Monte Carlo Simulations}
\label{sec:sec5}
The main simulation study evaluates the compared methods at sample sizes \(n\in\{250,500,1000\}\), using 500 replications per configuration. The configurations reflect the empirical applications and settings with larger numbers of variables and clusters. The separate exploratory tuning comparison in Table \ref{tab:lam_ari} and Figures \ref{fig:lam_res}--\ref{fig:lam_res_1} uses ten replications per configuration. Section \ref{subsec:sec5.1} describes the simulation settings and assessment metrics. Section \ref{subsec:sec5.2} assesses the competing estimated cluster partitions, parameter estimators, and criteria for selecting the number of clusters. Section \ref{subsec:sec5.3} further examines the effects of cluster imbalance and model misspecification.
R code and computational times for all estimation and clustering methods are available at \url{https://github.com/kendanny1/SCED}; computations were performed on a Fujitsu RX4770 server with four 2.20 GHz Intel Xeon E7-8890 v4 CPUs and 1 TB of RAM.

\begin{subsection}{Simulation Design and Performance Metrics}
\label{subsec:sec5.1}

The data were generated from clusterwise elliptical distributions under four scenarios: $(p,k)\in\{(6,2),(10,2),(10,3),(15,7)\}$. Although $(p,k)=(15,7)$ extends beyond the fixed-$(p,k)$ framework considered in our theoretical results, it assesses the performance of the proposed methodology in a higher-dimensional setting with a larger number of variables and clusters. Throughout, $0_{p}$ and $1_{p}$ denote the $p\times 1$ vectors of zeros and ones, respectively, and $I_{p}$ denotes the $p\times p$ identity matrix. The cluster membership probabilities were set to
$(\omega_{1},\omega_{2})=(0.6,0.4)$ for $k=2$, $(\omega_1,\omega_{2}, \omega_3)=(0.4,0.3,0.3)$ for $k=3$, and $(\omega_1,\omega_{2}, \dots, \omega_7)=(0.1,0.15 1_{6}^{\top})$ for $k=7$.
The cluster-specific mean vectors were given by
$\big\{0_{p},(1.5,0,\dots,1.5,0)^{\top}\big\}$ for $k=2$, $\big\{0_{p},1.5 1_{p},(1.5,0,\dots,1.5,0)^{\top}\big\}$ for $k=3$, and $\big\{0_{p},(1.5 1_{5}^{\top}, 0_{10}^{\top})^{\top},(0_{5}^{\top},1.5 1_{5}^{\top},0_{5}^{\top})^{\top},(0_{10}^{\top},1.5 1_{5}^{\top})^{\top},$ $(-1.5 1_{5}^{\top}, 0_{10}^{\top})^{\top},(0_{5}^{\top},-1.5 1_{5}^{\top},0_{5}^{\top})^{\top},(0_{10}^{\top},-1.5 1_{5}^{\top})^{\top}\big\}$ for $k=7$.
The cluster-invariant covariance matrix was specified as 
$\Sigma_{x}=\sigma^{2}(0.175I_{p}+0.0751_{p}1^{\top}_{p})$.
The scalar $\sigma\in\{1.2,1.6\}$ controls the degree of separation between clusters, with smaller values corresponding to more distinct clusters and larger values inducing greater overlap. Two density generators were considered:
\vspace{-0.05in}
\begin{align*}
\text{M1. }&f_{p}(y)=\gamma_{\text{o}} y^5\Big(\frac{45}{4}-y\Big)^{\frac{1}{4}}I\Big(0\leq y\leq \frac{45}{4}\Big),\\
\text{M2. }& f_{p}(y) = \frac{1}{(2\pi)^{\frac{p}{2}}}\exp\Big(-\frac{y}{2}\Big)I(y>0),
\end{align*}
where $\gamma_{\text{o}}=\Gamma(p/2))/[2\pi^{p/2}\int_0^{3\sqrt{5}/2} y^{p+9}(45/4-y^2)^{1/4}dy]$.

Agreement between the realized latent cluster partition and each estimated cluster partition was
quantified using the adjusted Rand index (ARI) \cite[]{hubert1985comparing}.
The performance of an estimator $\hat{\alpha}_{g}$ of a generic parameter vector $\alpha^{\text{o}}$ was assessed using the normalized root squared error (RSE), defined as
\vspace{-0.05in}
\begin{align*}
\sqrt{\frac{(\hat{\alpha}_{g}-\alpha^{\text{o}})^{\top}(\hat{\alpha}_{g}-\alpha^{\text{o}})}{d^{\alpha}}},
\end{align*}
where $d^{\alpha}$ denotes the number of parameters in $\alpha^{\text{o}}$. 
Because the SCED is identifiable up to a permutation of the cluster labels, the overall RSE of the estimated cluster-specific mean vectors is defined as the minimum over all permutations of the cluster labels.

The simulations assess the finite-sample performance of the unrestricted data-driven procedures. They are not intended to verify assumption {A2}, which depends on unknown true parameter neighborhoods, or to establish that using IMR starting values satisfies the conditions of Theorem \ref{Thm3.1}.
\end{subsection}

\begin{subsection}{Assessment of Estimation and Clustering} \label{subsec:sec5.2}

The clustering methods included in the comparison were $k$-means, IMR in Section \ref{subsec:sec4.1}, SP in Section \ref{sec:sec3.1}, and OC in Section \ref{sec:sec3.2}. 
The reported implementation uses the $k$-means-IMR initialization described in Section \ref{subsec:sec4.1}.
For OC, we considered two variants: one based on clusterwise multivariate normal distributions, denoted by $\text{OC}_{\text{n}}$, and the other based on clusterwise multivariate $t$-distributions, denoted by $\text{OC}_{\text{t}}$. 
Table \ref{tab:RI_p6} summarizes the performance of these methods based on their estimated cluster partitions.

Our IMR achieves substantially higher ARI values than
$k$-means, with the discrepancy widening as $n$ and $\sigma$ increase. Its performance is generally comparable to that of the proposed SP clustering method, except at $(p,k)=(15,7)$, where the ARI values are lower. 
We do not claim that using the IMR-estimated cluster partition satisfies the localization condition in assumption {A2} or that Theorem \ref{Thm3.1} applies to the IMR initialization. Empirically, however, IMR provides an effective initialization for the unrestricted SP procedure in the numerical settings considered.
Under model {M1}, the OC method outperforms the SP and $\text{OC}_{\text{n}}$ methods, achieving substantially higher ARI values for all $(p,k)$ combinations. Notably, the $\text{OC}_{\text{n}}$ method generally outperforms the $\text{OC}_{\text{t}}$ method. Under model {M2}, the $\text{OC}$ method yields higher ARI values than the SP method when $k=2$, while performing similarly or slightly worse when $k>2$. Moreover, the $\text{OC}_{\text{t}}$ method performs comparably to or slightly better than the $\text{OC}_{\text{n}}$ method, which in turn performs comparably to or slightly better than the $\text{OC}$ method. 

\begin{table}[htbp]
  \centering
  \caption{Means of 500 ARI values (scaled by $10^2$) of the realized latent cluster partition and various estimated cluster partitions under models {M1} and {M2}.}
    \begin{adjustbox}{max width=1\linewidth}
    \begin{tabular}{cccccccccccccccccccccc}
    \toprule
    \multicolumn{15}{c}{M1}\\
    \midrule
        \multicolumn{2}{c}{$(p,k)$}&\multicolumn{6}{c}{(6,2)} & &\multicolumn{6}{c}{(10,2)} \\
        \cmidrule(rl){1-2} \cmidrule(rl){3-8}  \cmidrule(rl){10-15} 
   $\sigma$      &  $n$     & $k$-means & IMR    & SP       & $\text{OC}$  & $\text{OC}_{\text{n}}$ & $\text{OC}_{\text{t}}$ & & $k$-means & IMR    & SP       & $\text{OC}$  & $\text{OC}_{\text{n}}$ & $\text{OC}_{\text{t}}$  \\
    \midrule
     1.2   & 250   & 78.92 & 95.25 & 95.26 & 97.40 & 95.58 & 95.50 &       & 86.41 & 99.15 & 99.16 & 99.46 & 99.28 & 99.27 \\
          & 500   & 79.67 & 95.88 & 95.88 & 97.78 & 96.01 & 95.93 &       & 86.47 & 99.34 & 99.34 & 99.57 & 99.37 & 99.35 \\
          & 1000  & 80.23 & 96.04 & 96.04 & 97.81 & 96.12 & 96.06 &       & 86.94 & 99.41 & 99.41 & 99.59 & 99.42 & 99.42 \\
     \midrule
    1.6   & 250   & 44.40 & 56.75 & 56.75 & 82.02 & 58.54 & 57.72 &       & 54.84 & 82.64 & 82.67 & 90.41 & 84.28 & 84.10 \\
          & 500   & 44.70 & 67.42 & 67.41 & 88.65 & 69.64 & 68.91 &       & 55.45 & 89.19 & 89.20 & 94.30 & 89.96 & 89.83 \\
          & 1000  & 45.15 & 73.26 & 73.26 & 89.89 & 75.34 & 74.92 &       & 55.61 & 90.13 & 90.13 & 94.75 & 90.57 & 90.49 \\
     \midrule 
     \multicolumn{2}{c}{$(p,k)$}&\multicolumn{6}{c}{(10,3)} & &\multicolumn{6}{c}{(15,7)} \\
        \cmidrule(rl){1-2} \cmidrule(rl){3-8}  \cmidrule(rl){10-15} 
   $\sigma$      &  $n$     & $k$-means & IMR    & SP       & $\text{OC}$  & $\text{OC}_{\text{n}}$ & $\text{OC}_{\text{t}}$ & & $k$-means & IMR    & SP       & $\text{OC}$  & $\text{OC}_{\text{n}}$ & $\text{OC}_{\text{t}}$  \\
    \midrule
    1.2   & 250   & 86.04 & 95.07 & 95.06 & 96.45 & 95.04 & 95.01 &       & 80.56 & 93.95 & 98.89 & 99.01 & 99.00 & 98.55 \\
          & 500   & 86.75 & 95.62 & 95.62 & 97.34 & 95.62 & 95.55 &       & 84.25 & 98.60 & 99.51 & 99.64 & 99.52 & 99.52  \\
          & 1000  & 87.25 & 95.90 & 95.90 & 97.59 & 95.91 & 95.88 &       & 87.13 & 99.55 & 99.63 & 99.72 & 99.64 & 99.64 \\
    \midrule
    1.6   & 250   & 53.99 & 63.01 & 63.01 & 68.54 & 62.75 & 62.46 &       & 61.26 & 69.86 & 70.34 & 71.57 & 70.94 & 71.06 \\
          & 500   & 54.33 & 70.27 & 70.27 & 78.78 & 70.13 & 69.55 &       & 63.41 & 84.87 & 85.70 & 89.43 & 86.36 & 86.50\\
          & 1000  & 54.87 & 75.44 & 75.44 & 85.22 & 75.41 & 75.07 &       & 64.00 & 91.45 & 91.73 & 94.34 & 91.91 & 91.94 \\
          \midrule \midrule
          \multicolumn{15}{c}{M2} \\
    \midrule
    \multicolumn{2}{c}{$(p,k)$}&\multicolumn{6}{c}{(6,2)} & &\multicolumn{6}{c}{(10,2)} \\
        \cmidrule(rl){1-2} \cmidrule(rl){3-8}  \cmidrule(rl){10-15} 
   $\sigma$      &  $n$     & $k$-means & IMR    & SP       & $\text{OC}$  & $\text{OC}_{\text{n}}$ & $\text{OC}_{\text{t}}$ & & $k$-means & IMR    & SP       & $\text{OC}$  & $\text{OC}_{\text{n}}$ & $\text{OC}_{\text{t}}$   \\
    \midrule
    1.2   & 250   & 79.76 & 91.17 & 91.20 & 91.37 & 91.46 & 91.45 &       & 86.08 & 97.38 & 97.40 & 97.52 & 97.52 & 97.53 \\
          & 500   & 79.94 & 91.65 & 91.66 & 91.88 & 91.91 & 91.95 &       & 86.04 & 97.77 & 97.78 & 97.86 & 97.87 & 97.89 \\
          & 1000  & 80.31 & 92.20 & 92.20 & 92.44 & 92.42 & 92.43 &       & 86.18 & 97.91 & 97.91 & 97.96 & 97.96 & 97.97 \\
    \midrule
    1.6   & 250   & 52.75 & 69.00 & 69.00 & 70.40 & 70.40 & 70.63 &       & 59.10 & 84.41 & 84.42 & 85.49 & 85.54 & 85.81 \\
          & 500   & 52.72 & 73.84 & 73.84 & 74.98 & 74.99 & 75.07 &       & 59.61 & 88.18 & 88.18 & 88.60 & 88.61 & 88.63 \\
          & 1000  & 53.21 & 75.82 & 75.82 & 76.85 & 76.93 & 76.96 &       & 59.71 & 88.77 & 88.77 & 89.20 & 89.21 & 89.24 \\
\midrule 
    \multicolumn{2}{c}{$(p,k)$}&\multicolumn{6}{c}{(10,3)} & &\multicolumn{6}{c}{(15,7)} \\
        \cmidrule(rl){1-2} \cmidrule(rl){3-8}  \cmidrule(rl){10-15} 
    $\sigma$      &  $n$     & $k$-means & IMR    & SP       & $\text{OC}$  & $\text{OC}_{\text{n}}$ & $\text{OC}_{\text{t}}$ & & $k$-means & IMR    & SP       & $\text{OC}$  & $\text{OC}_{\text{n}}$ & $\text{OC}_{\text{t}}$   \\
    \midrule
    1.2   & 250   & 85.87 & 93.05 & 93.04 & 92.98 & 93.04 & 93.03 &       & 80.89 & 92.14 & 97.57 & 97.68 & 97.68 & 97.31\\
          & 500   & 86.42 & 93.63 & 93.63 & 93.59 & 93.63 & 93.63 &       & 85.13 & 97.40 & 98.53 & 98.56 & 98.56 & 98.57 \\
          & 1000  & 86.68 & 93.91 & 93.91 & 93.91 & 93.92 & 93.96 &       & 87.75 & 98.73 & 98.84 & 98.75 & 98.75 & 98.76 \\
    \midrule
    1.6   & 250   & 59.29 & 70.98 & 70.98 & 70.91 & 70.88 & 71.09 &       & 62.61 & 71.42 & 72.07 & 72.71 & 72.36 & 72.97\\
          & 500   & 60.69 & 77.24 & 77.24 & 77.21 & 77.31 & 77.30 &       & 64.72 & 85.02 & 86.60 & 87.18 & 87.42 & 87.42 \\
          & 1000  & 61.17 & 78.82 & 78.82 & 78.92 & 78.85 & 78.91 &       & 65.44 & 90.26 & 90.28 & 90.63 & 90.39 & 90.45 \\
          \bottomrule
    \end{tabular}%
   \end{adjustbox}
  \label{tab:RI_p6}%
\end{table}%

The estimation methods compared for the cluster-specific mean vectors and the common covariance matrix were $k$-means, SP, PML, PMML. 
Under clusterwise multivariate normality with a common covariance matrix, the IMR and SP estimators are asymptotically equivalent to the maximum likelihood estimator of \cite{banfield1993model} and the oracle maximum likelihood estimator that treats the realized latent cluster partition as known, respectively.
Since the SP-estimated cluster partition performs comparably to or better than the IMR-estimated cluster partition, the RSEs of the IMR estimator are omitted. The PML and PMML estimators are based on the OC-estimated cluster partition, which outperforms the SP-estimated cluster partition.
We also considered maximum marginal likelihood estimators under the normal and $t$-mixture models, denoted by $\text{MML}_{\text{n}}$ and $\text{MML}_{\text{t}}$, respectively. These procedures provide parametric benchmarks for assessing the proposed unspecified-generator approach under {M1} and {M2}. The comparisons evaluate the implemented procedures under a nonnormal elliptical generator and a correctly specified normal generator; they are not a comprehensive comparison of clustering methods. Performance relative to other robust or non-elliptical approaches is not established here.
Figures \ref{fig:RMSE_p6} and \ref{fig:RMSE_p7} present the average RSEs for the estimated mean vectors and covariance matrix.

Across models {M1} and {M2}, the SP estimator yields substantially lower RSEs than the $k$-means estimator for the mean vectors and covariance matrix, with the difference becoming more pronounced as $n$ and $\sigma$ increase. Under model {M1}, the PML estimator substantially outperforms the SP estimator in RSE for both quantities, except that the reduction in the RSE of the mean vectors is only marginal at $(k,\sigma)=(3,1.2)$. Compared with the PMML estimator, the PML estimator generally yields lower RSEs, except for $\sigma=1.6$ and $k\geq 3$, where its performance deteriorates owing to the inferior OC-estimated cluster partition. Under model {M2}, the PMML estimator achieves RSEs for the mean vectors that are comparable to or lower than those of the SP and PML estimators. For the covariance matrix, the SP, PML, and PMML estimators exhibit comparable RSEs at $\sigma=1.2$, although the SP estimator tends to achieve slightly lower RSEs. At $\sigma=1.6$, the PMML estimator achieves substantially lower RSEs than the SP and PML estimators because the latter rely on less accurate estimated cluster partitions. 
 
Compared with the parametric benchmarks, Figures \ref{fig:RMSE_p6} and \ref{fig:RMSE_p7} show that under model M1, the PML or PMML estimator generally achieves substantially lower RSEs than the $\text{MML}_{\text{n}}$ and $\text{MML}_{\text{t}}$ estimators, with a marginal gain at $(k,\sigma)=(3,1.2)$. Under model M2, however, the RSEs of the PMML estimator are generally comparable to or slightly lower than those of the $\text{MML}_{\text{n}}$ and $\text{MML}_{\text{t}}$ estimators.
Across all scenarios, the RSEs decrease with increasing $n$ and $p$ and decreasing $\sigma$ and $k$. For models {M1} and {M2}, the average RSEs of the PML and PMML estimators approach the asymptotically semiparametric efficient (ASPE) benchmark as $n$ increases.
Under model {M2}, the average RSE of the $\text{MML}_{\text{n}}$ estimator also approaches the asymptotically efficient (AE) benchmark.

\begin{sidewaysfigure}[p]  % [p] 讓它獨立佔一頁，或用 [htbp]
    \centering
    \includegraphics[width=23cm,height=14cm]{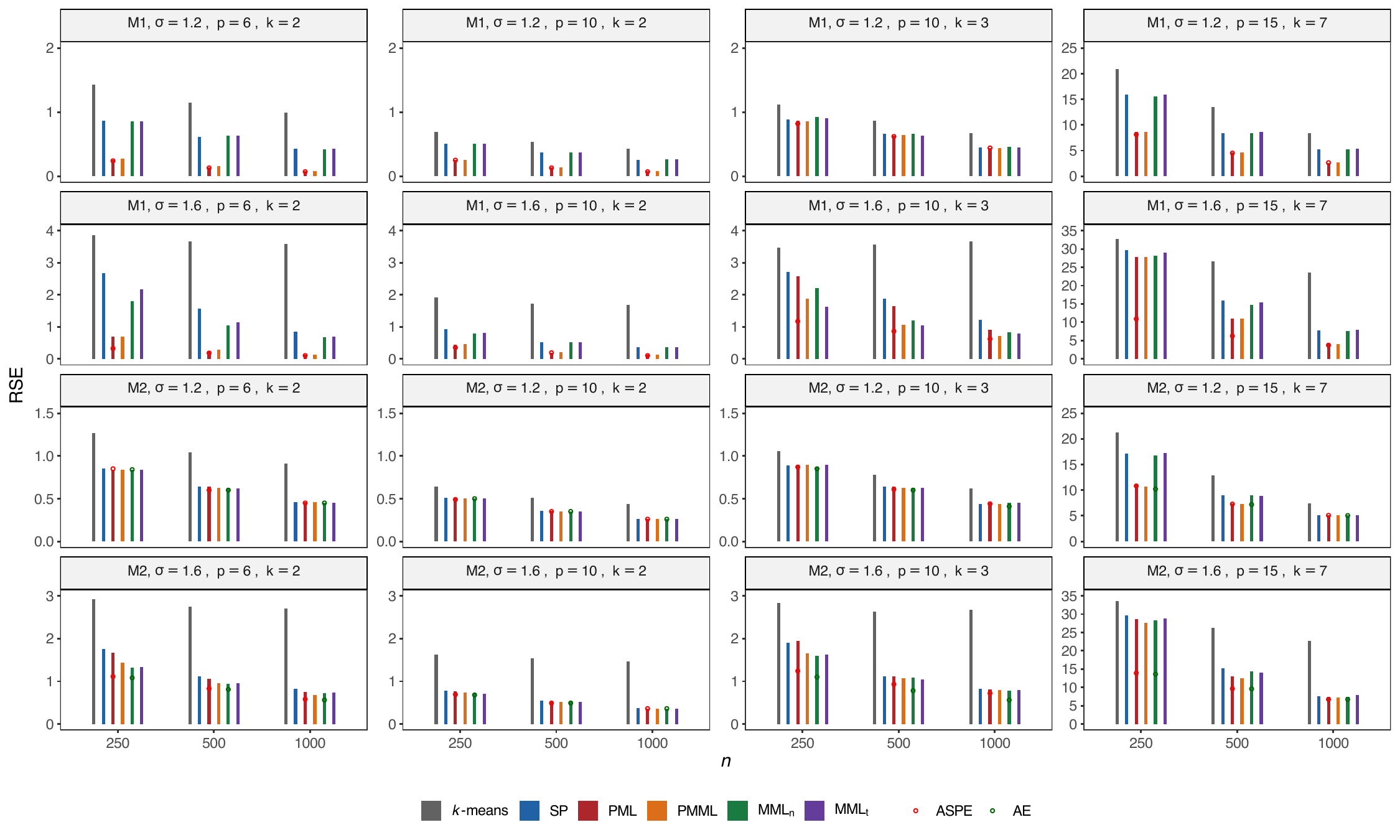}  
    \caption{\scriptsize{Means of 500 RSEs (scaled by $10^2$) of the proposed, competing, and oracle estimates of the mean vectors under models {M1} and {M2}.}} 
    \label{fig:RMSE_p6}  
\end{sidewaysfigure}

\begin{sidewaysfigure}[p]
    \centering
    \includegraphics[width=23cm,height=14cm]{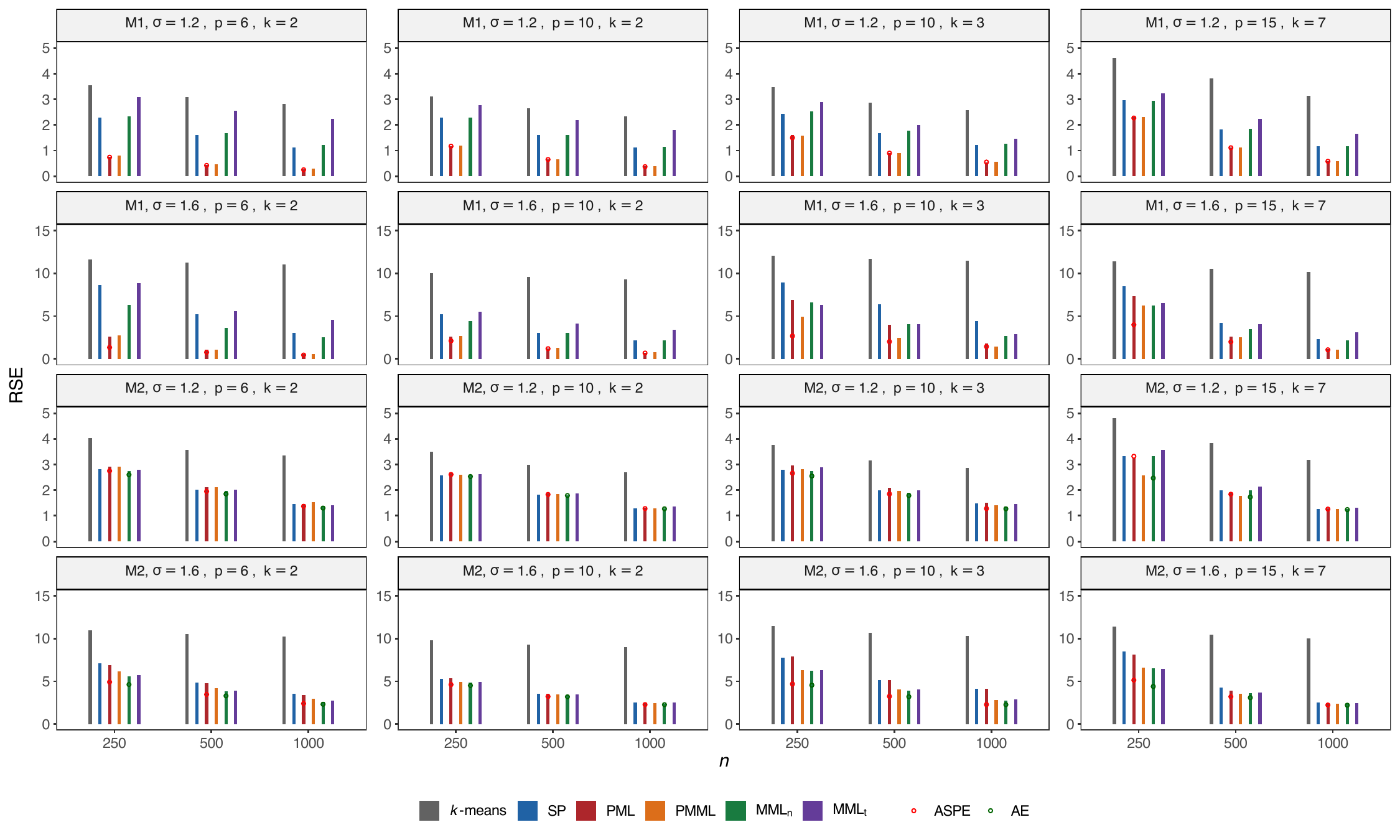}  
    \caption{\scriptsize{Means of 500 RSEs (scaled by $10^2$) of the proposed, competing, and oracle estimates of the covariance matrix under models {M1} and {M2}.}} 
    \label{fig:RMSE_p7}  
\end{sidewaysfigure}

Table \ref{tab:SPIC} compares the proposed SPIC with Bayesian information criteria for mixtures of normal and $t$ distributions, denoted by $\text{BIC}_{\text{n}}$ and $\text{BIC}_{\text{t}}$, respectively, in selecting the number of clusters. Cluster-number comparisons are restricted to SPIC,  $\text{BIC}_{\text{n}}$ and $\text{BIC}_{\text{t}}$. Distribution-free criteria such as the Elbow method \citep{thorndike1953belongs, cormack1971review} and the Silhouette method \citep{rousseeuw1987silhouettes} are not evaluated.

Under model {M1}, the SPIC outperforms both competing criteria, whereas $\text{BIC}_{\text{n}}$ generally outperforms $\text{BIC}_{\text{t}}$. For $(p,k)=(6,2)$, $\text{BIC}_{\text{t}}$ tends to overestimate the number of clusters, and this overestimation persists as $n$ increases. Under {M2}, the relative ranking of the criteria varies across configurations. For example, at $(p,k)=(15,7)$, $\sigma=1.6$, and $n=500$, the correct-selection proportions are 0.510 for SPIC, 0.726 for 
$\text{BIC}_{\text{n}}$, and 0.714 for $\text{BIC}_{\text{t}}$. Although SPIC's mean selected cluster count is closest to seven in this configuration, its correct-selection proportion is lower. The correct-selection results therefore qualify conclusions based on the mean selected number alone.
To assess the performance of the SPIC under a homogeneous population ($k=1$), data were generated from Cluster 1 of the designed clusterwise elliptical distributions with $(p,k)=(6,2)$, using the density generator specified by model {M1}. Across all combinations of $(\sigma,n)$, the SPIC selects the one-cluster solution in every simulation replication.

\begin{table}[htbp]
  \centering
   \caption{Means (correct selection proportions) of 500 estimated number of clusters using different information criteria under models {M1} and {M2}.}
  \begin{adjustbox}{max width=1\linewidth}
    \begin{tabular}{ccccccccccccccccc}
    \toprule
    \multicolumn{17}{c}{M1} \\
    \midrule
    \multicolumn{2}{c}{$(p,k)$}&\multicolumn{3}{c}{(6,2)} & &\multicolumn{3}{c}{(10,2)}&&\multicolumn{3}{c}{(10,3)}&&\multicolumn{3}{c}{(15,7)} \\
    \cmidrule(rl){1-2} \cmidrule(rl){3-5}  \cmidrule(rl){7-9} \cmidrule(rl){11-13}  \cmidrule(rl){15-17} 
     $\sigma$      &  $n$     & $\text{SPIC}$       & $\text{BIC}_{n}$      & $\text{BIC}_{t}$ &  & $\text{SPIC}$      & $\text{BIC}_{n}$      & $\text{BIC}_{t}$  &  & $\text{SPIC}$      & $\text{BIC}_{n}$      & $\text{BIC}_{t}$ &  & $\text{SPIC}$      & $\text{BIC}_{n}$      & $\text{BIC}_{t}$       \\  
     \midrule
    1.2   & 250     & 2.00  & 2.00  & 2.00  &           & 2.00  & 2.00  & 2.00  &           & 3.00  & 3.00  & 3.00 & & 7.05 & 6.93 &6.88\\
                     &     &(1.000)    &(1.000)    &(0.998)&                &(1.000)     &(1.000)      &(1.000)     &           &(0.998)      &(0.998)      &(0.998)    & & (0.946)      &(0.772)      &(0.774)\\
                    & 500   & 2.00  & 2.00  & 2.02  &           & 2.00  & 2.00  & 2.00  &           & 3.00  & 3.00  & 3.00 & & 7.02 & 7.04 &7.04\\
                    &     &(1.000)    &(1.000)    &(0.990)&                &(1.000)     &(1.000)      &(1.000)     &           &(1.000)      &(0.998)      &(1.000)    & & (0.980)      &(0.962)      &(0.962)\\
                    & 1000  & 2.00  & 2.03  & 2.75  &           & 2.00  & 2.00  & 2.00  &           & 3.00  & 3.00  & 3.00 & & 7.00 & 7.00 &7.00\\
                    &     &(1.000)    &(0.986)    &(0.626)&               &(1.000)     &(1.000)      &(1.000)     &           &(1.000)      &(1.000)      &(1.000)    & & (1.000)      &(0.996)      &(0.996)\\
                    \midrule
     1.6   & 250   & 2.02  & 1.40  & 1.43  &           & 2.00  & 1.92  & 1.87  &           & 2.97  & 2.80  & 2.74& & 7.18 & 1.20 &1.19 \\
             &     &(0.976)    &(0.394)    &(0.280)&                &(0.998)     &(0.916)      &(0.874)     &           &(0.814)      &(0.752)      &(0.698)    & & (0.374)      &(0.000)      &(0.000)\\
                    & 500   & 2.01  & 1.97  & 1.78  &           & 2.00  & 2.00  & 2.00  &           & 3.01  & 3.02  & 3.01 & & 6.91 & 6.41 &6.27\\
                    &     &(0.990)    &(0.962)    &(0.582)&                &(1.000)     &(1.000)      &(1.000)     &           &(0.988)      &(0.982)      &(0.956)    & & (0.920)      &(0.580)      &(0.460)\\
                    & 1000  & 2.00  & 2.04  & 2.78  &           & 2.00  & 2.00  & 2.00  &           & 3.02  & 3.01  & 3.02 & & 7.00 & 7.00 &6.99\\
                    &     &(0.998)    &(0.966)    &(0.330)&                &(1.000)     &(1.000)      &(1.000)     &           &(0.988)      &(0.984)      &(0.966)    & & (0.998)      &(0.980)      &(0.978)\\
                    \midrule \midrule
                    \multicolumn{17}{c}{M2} \\
    \midrule
                    \multicolumn{2}{c}{$(p,k)$}&\multicolumn{3}{c}{(6,2)} & &\multicolumn{3}{c}{(10,2)}&&\multicolumn{3}{c}{(10,3)}&&\multicolumn{3}{c}{(15,7)} \\
    \cmidrule(rl){1-2} \cmidrule(rl){3-5}  \cmidrule(rl){7-9} \cmidrule(rl){11-13} \cmidrule(rl){15-17} 
     $\sigma$      &  $n$     & $\text{SPIC}$       & $\text{BIC}_{n}$      & $\text{BIC}_{t}$ &  & $\text{SPIC}$      & $\text{BIC}_{n}$      & $\text{BIC}_{t}$  &  & $\text{SPIC}$      & $\text{BIC}_{n}$      & $\text{BIC}_{t}$ &  & $\text{SPIC}$      & $\text{BIC}_{n}$      & $\text{BIC}_{t}$       \\  
     \midrule
    1.2     & 250   & 2.01  & 1.99  & 2.00  &           & 2.00  & 2.00  & 2.00  &           & 3.01  & 3.00  & 3.00 & & 7.69 & 7.00 & 7.01 \\
             &     &(0.992)    &(0.990)    &(1.000)&                &(1.000)     &(0.994)      &(1.000)     &           &(0.994)      &(0.998)      &(0.998)    & & (0.332)      &(0.756)      &(0.756)\\
                    & 500   & 2.00  & 2.00  & 2.00  &           & 2.00  & 2.00  & 2.00  &           & 3.00  & 3.00  & 3.00& & 7.00 & 7.06 & 7.06  \\
                    &     &(0.998)    &(0.998)    &(1.000)&                &(1.000)     &(1.000)      &(1.000)     &           &(1.000)      &(1.000)      &(1.000)    & & (0.996)      &(0.944)      &(0.944)\\
                    & 1000  & 2.00  & 2.00  & 2.00  &           & 2.00  & 2.00  & 2.00  &           & 3.00  & 3.00  & 3.00& & 7.00 & 7.00 & 7.00  \\
                    &     &(1.000)    &(1.000)    &(1.000)&                &(1.000)     &(1.000)      &(1.000)     &           &(1.000)      &(1.000)      &(1.000)    & & (1.000)      &(1.000)      &(1.000)\\
                    \midrule
        1.6   & 250   & 1.99  & 1.96  & 1.59  &           & 1.96  & 1.96  & 1.98  &           & 2.78  & 2.94  & 2.94& & 7.35 & 1.20 & 1.28  \\
                &     &(0.982)    &(0.956)    &(0.446)&                &(0.958)     &(0.958)      &(0.978)     &           &(0.782)      &(0.938)      &(0.938)    & & (0.468)      &(0.002)      &(0.002)\\
                    & 500   & 2.00  & 1.96  & 2.00  &           & 2.00  & 1.98  & 2.00  &           & 3.00  & 3.00  & 3.00 & & 6.71 & 6.67 & 6.69 \\
                    &     &(0.998)    &(0.962)    &(0.998)&                &(0.996)     &(0.978)      &(1.000)     &           &(0.998)      &(0.996)      &(0.996)    & & (0.510)      &(0.726)      &(0.714)\\
                    & 1000  & 2.00  & 2.00  & 2.00  &           & 2.00  & 2.00  & 2.00  &           & 3.00  & 3.00  & 3.00& & 7.00 & 6.99 & 7.01  \\
                    &     &(1.000)    &(1.000)    &(1.000)&                &(1.000)     &(1.000)      &(1.000)     &           &(0.998)      &(0.998)      &(0.998)    & & (0.998)      &(0.986)      &(0.986)\\
          \bottomrule
    \end{tabular}% 
        \end{adjustbox}
  \label{tab:SPIC}%
\end{table}%

Overall, the OC-estimated cluster partition outperforms its parametric counterparts under misspecified density generators, while performing comparably under correctly specified density generators. Across models {M1} and {M2}, its ARI decreases as $\sigma$ and $k$ increase and as $n$ and the number $p$ of informative variables used for clustering decrease.
The PML or PMML estimator tends to outperform the competing estimators, with the PMML estimator offering substantial finite-sample advantages particularly when $n$ is small, 
$\sigma$ is large, or $p$ is large. The SPIC performs favorably under {M1} and in several {M2} configurations, but the relative performance is configuration-dependent and no criterion is uniformly best across the reported settings.

\end{subsection}

\begin{subsection}{Assessment of Cluster Imbalance and Model Misspecification} \label{subsec:sec5.3}

To assess the effect of cluster imbalance on the OC-estimated cluster partition, the PML and PMML estimators, and the SPIC, we additionally set the cluster membership probabilities to $(\omega_1,\omega_2)=(0.8,0.2)$ under model {M1} with $(p,k)=(6,2)$. The mean ARI values (scaled by $10^2$) at $n=250,500,1000$ are 84.47, 93.73, and 96.15, respectively, for $\sigma=1.2$, and 39.73, 39.92, and 33.93, respectively, for $\sigma=1.6$. For the (PML, PMML) estimators of the mean vectors, the mean RSEs (scaled by $10^2$) at $n=250,500,1000$ are $(1.27,1.27)$, $(0.46,0.46)$, and $(0.23,0.24)$, respectively, for $\sigma=1.2$, and $(5.35,5.30)$, $(5.41,5.38)$, and $(5.81,5.79)$, respectively, for $\sigma=1.6$. For the corresponding estimators of the common covariance matrix, the mean RSEs are $(2.78,2.77)$, $(1.30,1.29)$, and $(0.74,0.76)$, respectively, for $\sigma=1.2$, and $(11.50,11.38)$, $(12.22,12.13)$, and $(14.98,14.87)$, respectively, for $\sigma=1.6$.
The mean estimated numbers of clusters (proportions of correct cluster-number selection) are 2.16 (0.860), 2.06 (0.950), and 2.02 (0.976) for $\sigma=1.2$, and 2.58 (0.510), 2.50 (0.520), and 2.66 (0.414) for $\sigma=1.6$.
Taken together with the results in Tables \ref{tab:RI_p6} and \ref{tab:SPIC} and Figures \ref{fig:RMSE_p6} and \ref{fig:RMSE_p7}, these findings suggest that increasing the degree of cluster imbalance from $(\omega_1,\omega_2)=(0.6,0.4)$ to $(0.8,0.2)$ tends to decrease the ARI of the OC-estimated cluster partition, increase the RSEs of the PML and PMML estimators, and adversely affect SPIC performance. These effects are more pronounced for $\sigma=1.6$ than for $\sigma=1.2$.

The 90/10 configurations show substantial deterioration (Table \ref{tab:label4}). At the largest sample size examined, $n=1000$, mean ARIs on the original, unscaled ARI scale are 0.0040 and 0.0022 for $\sigma=1.2$ and 1.6, respectively; the corresponding correct cluster-number selection proportions are 0.004 and 0.174. Increasing the sample size over the examined range does not remedy these clustering failures. The results do not separate the contributions of initialization and subsequent estimation.

Performance also deteriorates under 80/20 imbalance with greater overlap. At $\sigma=1.6$ and $n=1000$, the mean ARI on the unscaled scale is 0.3393 and the correct-selection proportion is 0.414. The favorable baseline results should therefore not be generalized to strongly imbalanced, overlapping clusters.

To assess the robustness of the OC-estimated cluster partition and the SPIC to departures from the clusterwise elliptical framework, we considered $(p,k)=(6,2)$ and generated $X$ according to  
\begin{align*}  X = \sum_{c=1}^{2}I(C=c) \Big[ Z\Big(\mu_c-\mu_0 1_{6}+\Sigma_x^{\frac{1}{2}}U\Big) + (1-Z)\Big(\mu_c+\mu_0 1_{6}+\sigma_0\Sigma_x^{\frac{1}{2}}U\Big) \Big], \end{align*}  
where $U$ has the density generator specified in model {M1}, $Z\sim\operatorname{Bernoulli}(0.5)$ is independent of $C$, and $(\mu_0,\sigma_0)=(0.25,1.1)$. This two-component mixture combines distinct location shifts and scale factors, inducing non-elliptical distributions within each cluster and allowing us to assess the sensitivity of the proposed method to departures from its modeling assumptions.
The mean ARI values, scaled by $10^2$, at $n=250,500,1000$ are 92.22, 93.88, and 94.24, respectively, for $\sigma=1.2$, and 52.21, 66.43, and 77.70, respectively, for $\sigma=1.6$. The corresponding mean estimated numbers of clusters (proportions of correct cluster-number selection) are 2.02 (0.988), 2.00 (1.00), and 2.00 (1.00) for $\sigma=1.2$, and 2.23 (0.84), 2.08 (0.930), and 2.05 (0.958) for $\sigma=1.6$. Relative to the corresponding baseline setting, this location-and-scale mixture departure reduces OC clustering accuracy and, particularly at $\sigma=1.6$, the frequency with which SPIC selects the generating number of clusters. The experiment characterizes sensitivity to this particular departure, which changes both distributional shape and within-cluster dispersion. It does not isolate the effect of non-ellipticity or compare the relative performance of competing methods outside the elliptical family.

\end{subsection}

\end{section}

\begin{section}{Applications} 
\label{sec:sec6}

The applications illustrate customer purchasing profiles and patient biomarker-age profiles. They compare fitted means, estimated partitions, and reported bootstrap variability. Because latent memberships and population parameters are unknown, these summaries do not establish comparative clustering or parameter-estimation accuracy. In this section, ARI measures agreement between estimated partitions rather than agreement with a known true partition. Continuous variables were standardized to mean zero and variance one prior to model fitting. The PML and PMML estimators use the OC-estimated cluster partition, as in Section  \ref{sec:sec5}. Let $\ell_{0}(k)=\max_{\eta}\ell(\eta)$
and define
 \begin{align*}
 D=\frac{1}{n}\big(p\ell_{1}\big(\hat{k}\big)-\ell_{0}\big(\hat{k}\big)\big)-\frac{\log n}{n^{\frac{4}{5}}}.
 \end{align*}
We report $D$ as an exploratory likelihood-based diagnostic within the clusterwise elliptical framework. The applications do not use it as a calibrated goodness-of-fit test or as evidence of more accurate cluster assignments. Throughout, missingness in the variables is assumed to be completely at random; extensions to other missingness mechanisms are left for future work. 

\begin{subsection}{Application in Customer Segmentation}
\label{subsec:Customer}

Customer segmentation enables differentiated marketing by identifying distinct groups of shoppers with coherent purchasing patterns. To illustrate our methodology, we analyze transaction data from a North American supermarket chain spanning 2017-2020. The source consists of two primary tables: (i) a Transaction table (about 1.2 billion rows) with customer ID, transaction ID, date, product ID, quantity, and unit price; and (ii) a Product table (about 150,000 rows) with description, category, subcategory, brand, manufacturer, and pack size.
For empirical clarity and to avoid pandemic-related distortions, we focus on calendar year 2018---the most complete pre-COVID year in the dataset---and restrict attention to reliably active shoppers. Inclusion criteria are: (a) loyalty-program members; (b) at least 26 transactions in 2018 (multiple transactions within a week counted as a single purchasing occasion); (c) evidence of engagement---at least one transaction every two months; and (d) for the category analysis, participation in product categories where each included customer records at least 26 transactions in 2018. This design balances sample size, purchase regularity, and representativeness while reducing the risk of inadvertently including churned customers.
After filtering, the reduced dataset contains 4,062 customers and 116,426 transactions across nine major product categories. We construct two variable sets per customer: (1) six bimonthly purchase amounts for Jan-Feb, Mar-Apr, May-Jun, Jul-Aug, Sep-Oct, and Nov-Dec; and (2) category-level purchase amounts for Cold Beverages, Fruit, Household, In-Store (prepared foods), Meal Makers, Milk \& Eggs, Snacks, Natural Foods, and Vegetables. The bimonthly profile serves as the clustering feature space. Category totals are not entered as clustering features and are used to describe the estimated groups. They are not an independent validation of comparative clustering accuracy.

Let $T_{1}$ through $T_{6}$ denote the six bimonthly intervals and let $\text{PA}_{1}$ through $\text{PA}_{6}$ be the corresponding log-transformed purchase amounts. We fit the SCED to $\text{PA}_{1}$, $\text{PA}_{2}$, \dots, $\text{PA}_{6}$. Among the evaluated cluster counts, SPIC favors three clusters, whereas $\text{BIC}_{\text{n}}$ and $\text{BIC}_{\text{t}}$ favor six and five, respectively (Table \ref{ab:spic0}). The following partition comparisons condition all methods on three clusters. Agreement with OC, measured by the ARI, is 0.272 for $k$-means, 0.634 for IMR, 0.634 for SP, 0.789 for $\text{OC}_{\text{n}}$, and 0.710 for $\text{OC}_{\text{t}}$. These values describe agreement at a common cluster count; they do not compare the partitions obtained using each criterion's selected number of clusters. Under the OC-estimated cluster partition, Clusters 1 to 3 contain 2,513, 1,081, and 468 customers, with mean purchase amounts of 371.20, 199.10, and 277.75, and mean transaction counts of 28.95, 28.30, and 27.98, respectively.

\begin{table}[htbp]
  \centering
  \caption{Information criterion values versus number of clusters.}
  \begin{adjustbox}{max width=0.5\linewidth}
      \begin{tabular}{cccc}
\toprule
    Number of clusters & SPIC  & $\text{BIC}_{\text{n}}$ & $\text{BIC}_{\text{t}}$ \\
\midrule
    1     & 7.913 & 7.971 & 7.910 \\
    2     & 7.898 & 7.953 & 7.858 \\
    3     & \textbf{7.865} & \textbf{7.893} & \textbf{7.822} \\
    4     & 7.895 & 7.927 & 7.841 \\
    5     & 7.890 & 7.879 & \textbf{7.818} \\
    6     & 7.910 & \textbf{7.870} & 7.825 \\
\bottomrule
    \end{tabular}%
    \end{adjustbox}
      \label{ab:spic0}%
\end{table}%

The exploratory diagnostic values are $D=0.021$ and $D=0.088$ for the clusterwise normal and $t$ specifications, respectively. Table \ref{tab:addlabel0} reports the fitted means and their bootstrap standard errors. Bootstrap evaluation of the PML and PMML estimates used 500 within-cluster resamples. The reported standard errors of their mean estimates are generally smaller than those of the parametric estimators. These entries summarize variability under the implemented resampling procedures; they do not establish smaller bias or greater accuracy relative to the unknown population parameters.
Using the estimated mean vectors with bootstrap standard errors, pairwise comparisons of purchase amounts across the six time intervals show statistically significant differences at the 0.05 level: Cluster 1 exceeds Cluster 2 from $T_1$ to $T_6$; Cluster 1 exceeds Cluster 3 from $T_2$ to $T_6$; and Cluster 3 exceeds Cluster 2 at every time period except $T_5$.
 
Figure \ref{fig:Sindian_mean_D} shows that in $T_5$, Cluster 3 records significantly lower purchase amounts than Clusters 1 and 2 across all nine major product categories. Aside from the In-Store category in $T_1$ and $T_2$, Cluster 2 purchases are significantly lower than Cluster 1’s across the remaining categories. Relative to Cluster 3, Cluster 1 is significantly higher in Household ($T_2$), Milk \& Eggs ($T_4$), and---in $T_6$---In-Store, Milk \& Eggs, Snacks, Natural Foods, and Vegetables. Except for $T_5$, Cluster 2’s purchase amounts are generally comparable to or significantly lower than Cluster 3’s across categories. Table \ref{tab:addlabel} corroborates these patterns for average weekly spend: Cluster 1 exceeds Cluster 2 in every category and exceeds Cluster 3 in all but Cold Beverages. Compared with Cluster 3, Cluster 2 is significantly lower in every category except Cold Beverages, Household, and In-Store. With the exception of Fruit, cross-cluster patterns in average purchase amounts closely mirror those in average weekly purchase amounts. At three clusters, each method yields a fitted profile whose standardized $\text{PA}_5$ mean is markedly lower than those of its other profiles (Table \ref{tab:addlabel0}). Under PML and PMML, that profile has the second-highest $\text{PA}_2$ mean, whereas under $\text{MML}_{\text{n}}$ and $\text{MML}_{\text{t}}$ it has the highest. These are similarities and differences in the within-solution ordering of fitted profiles; they do not establish identical membership across methods. The results do not determine which assignments are more appropriate for customer targeting.

\begin{table}[htbp]
  \centering
  \caption{Mean vector estimates (bootstrap standard errors) from the PML, PMML, $\text{MML}_{\text{n}}$, and $\text{MML}_{\text{t}}$ methods, and covariance matrix estimates (bootstrap standard errors), with the upper and lower triangles corresponding to the PML and PMML estimates, respectively.}
\begin{adjustbox}{max width=1\linewidth}
    \begin{tabular}{ccccccccccccc}
\toprule
      Method            &  \multicolumn{3}{c}{PML}  &      \multicolumn{3}{c}{PMML} &        \multicolumn{3}{c}{$\text{MML}_{\text{n}}$} &      \multicolumn{3}{c}{$\text{MML}_{\text{t}}$}     \\
\cmidrule(rl){1-1} \cmidrule(rl){2-4} \cmidrule(rl){5-7} \cmidrule(rl){8-10} \cmidrule(rl){11-13}
    Variable & $\mu_1$ & $\mu_2$ & $\mu_3$ & $\mu_1$ & $\mu_2$ & $\mu_3$ & $\mu_1$ & $\mu_2$ & $\mu_3$ & $\mu_1$ & $\mu_2$ & $\mu_3$ \\
    \midrule
    $\text{PA}_{1}$ & 0.46  & -1.15 & 0.36  & 0.44  & -1.02 & 0.37  & 0.26  & -1.64 & 0.40   & 0.28  & -1.41 & 0.41 \\
          & (0.012) & (0.019) & (0.025) & (0.015) & (0.027) & (0.023) & (0.026) & (0.179) & (0.071) & (0.041) & (0.203) & (0.087) \\
    $\text{PA}_{2}$ & 0.36  & -0.85 & 0.18  & 0.39  & -0.76 & 0.18  & 0.08  & -0.58 & 0.30   & 0.09  & -0.06 & 0.25 \\
          & (0.013) & (0.033) & (0.039) & (0.016) & (0.037) & (0.032) & (0.048) & (0.160) & (0.088) & (0.054) & (0.131) & (0.072) \\
    $\text{PA}_{3}$ & 0.25  & -0.55 & 0.05  & 0.31  & -0.51 & 0.08  & 0.07  & -0.48 & 0.20   & 0.07  & 0.10   & 0.15 \\
          & (0.018) & (0.038) & (0.029) & (0.020) & (0.026) & (0.030) & (0.028) & (0.086) & (0.069) & (0.031) & (0.063) & (0.069) \\
    $\text{PA}_{4}$ & 0.23  & -0.46 & -0.06 & 0.29  & -0.42 & -0.01 & 0.08  & -0.45 & 0.05  & 0.07  & 0.09  & -0.02 \\
          & (0.018) & (0.037) & (0.034) & (0.022) & (0.023) & (0.041) & (0.032) & (0.101) & (0.075) & (0.033) & (0.085) & (0.080) \\
    $\text{PA}_{5}$ & 0.49  & -0.31 & -1.45 & 0.48  & -0.28 & -1.19 & 0.19  & -0.42 & -2.05 & 0.20   & 0.10   & -1.67 \\
          & (0.008) & (0.028) & (0.033) & (0.024) & (0.023) & (0.026) & (0.035) & (0.091) & (0.143) & (0.047) & (0.093) & (0.173) \\
    $\text{PA}_{6}$ & 0.24  & -0.41 & -0.21 & 0.30   & -0.35 & -0.14 & 0.07  & -0.42 & -0.15 & 0.10   & 0.08  & -0.16 \\
          & (0.017) & (0.029) & (0.032) & (0.024) & (0.027) & (0.028) & (0.032) & (0.074) & (0.112) & (0.035) & (0.072) & (0.086) \\
\bottomrule
    \end{tabular}%
    \end{adjustbox}
      \centering
  \caption*{}
  \begin{adjustbox}{max width=0.6\linewidth}
      \begin{tabular}{cccccccc}
\toprule
       & {\small $\text{PA}_{1}$} &  {\small $\text{PA}_{2}$} &  {\small $\text{PA}_{3}$} &  {\small $\text{PA}_{4}$} &  {\small $\text{PA}_{5}$} &  {\small $\text{PA}_{6}$} &  \\
          \midrule
           & 0.49  & -0.01 & 0.02  & -0.02 & -0.01 & -0.06 & {\small $\text{PA}_{1}$} \\
          & (0.012) & (0.013) & (0.008) & (0.010) & (0.010) & (0.009) &  \\
          &       & 0.72  & 0.26  & 0.17  & 0.14  & 0.07  & {\small $\text{PA}_{2}$} \\
          &       & (0.017) & (0.012) & (0.011) & (0.011) & (0.014) &  \\
    {\small $\text{PA}_{1}$} & 0.49  &       & 0.86  & 0.36  & 0.23  & 0.14  & {\small $\text{PA}_{3}$} \\
          & (0.017) &       & (0.019) & (0.016) & (0.012) & (0.013) &  \\
   {\small $\text{PA}_{2}$} & -0.02 & 0.71  &       & 0.91  & 0.32  & 0.18  & {\small $\text{PA}_{4}$} \\
          & (0.011) & (0.023) &       & (0.022) & (0.012) & (0.016) &  \\
    {\small $\text{PA}_{3}$}& 0.01  & 0.26  & 0.87  &       & 0.53  & 0.24  & {\small $\text{PA}_{5}$} \\
          & (0.012) & (0.023) & (0.027) &       & (0.014) & (0.012) &  \\
    {\small $\text{PA}_{4}$} & -0.02 & 0.16  & 0.36  & 0.90   &       & 0.91  & {\small $\text{PA}_{6}$} \\
          & (0.013) & (0.015) & (0.021) & (0.028) &       & (0.026) &  \\
    {\small $\text{PA}_{5}$} & 0.02  & 0.15  & 0.24  & 0.32  & 0.52  &       &  \\
          & (0.012) & (0.011) & (0.015) & (0.014) & (0.023) &       &  \\
    {\small $\text{PA}_{6}$} & -0.06 & 0.07  & 0.14  & 0.18  & 0.25  & 0.91  &  \\
          & (0.011) & (0.014) & (0.013) & (0.014) & (0.011) & (0.028) &  \\
\bottomrule
    \end{tabular}%
  \label{tab:addlabel0}%
  \end{adjustbox}
  \end{table}%

\end{subsection}

\begin{subsection}{Application in Pima Indians Diabetes Research} 
\label{subsec:Pima}

The Pima Indians Diabetes dataset compiled by the National Institute of Diabetes and Digestive and Kidney Diseases contains records on 768 adult women of Pima Indian heritage ($age \ge 21$ at baseline) from a longitudinal study of incident diabetes over a 1-5 year follow-up. For each participant, we observe $age$; gravidity status $gs$ (coded 0 if the number of pregnancies $\le 2$, 1 if $>2$); diabetes status $ds$ ($0 =$ non-diabetic, $1 =$ diabetic); and six log-transformed clinical measures: two-hour plasma glucose after an oral glucose tolerance test ($gtt$), diastolic blood pressure ($dbp$), triceps skinfold thickness ($tsft$), two-hour serum insulin ($si$), body mass index ($bmi$), and diabetes pedigree function ($dpf$). To mitigate physiologically implausible values, we excluded patients with extreme biomarker readings or zeros for either $bmi$ or $gtt$, yielding a final analytic sample of 388 participants. Of these, 129 met the WHO diagnostic criteria for diabetes and 259 did not. Among participants with recorded gravidity, 210 had at most two pregnancies and 178 had more than two. The study objective is to identify latent clusters from the biomarker-age profile and examine their associations with $ds$ and $gs$.

We fit the SCED to $gtt,dbp,tsft,si,bmi,dpf,$ and $age$. Model selection using the SPIC and two parametric comparators ($\text{BIC}_{\text{n}}$ and $\text{BIC}_{\text{t}}$) consistently favor 
four clusters (Table~\ref{tab:spic}). The degree of agreement with OC varies across clustering methods. The ARI is
0.605 between $k$-means and OC, whereas it is 0.924 between IMR and OC, 0.924 between SP and OC, 0.958 between $\text{OC}_{\text{n}}$ and OC, and 0.924 between $\text{OC}_{\text{t}}$ and OC. Under the OC solution, cluster sizes are 102, 105, 100 and 81, respectively, for Clusters 1 to 4. Diabetes prevalence is low in Clusters 1 and 3 (0.128 and 0.120) and markedly higher in Clusters 2 and 4 (0.543 and 0.580). The proportions with low gravidity are 0.716, 0.571, 0.660, and 0.139 in Clusters 1 to 4, respectively. Comparing the OC-estimated cluster partition to the partition defined jointly by $(ds,gs)$ yields ARI$ = 0.118$, indicating that while disease/gravidity status is informative, the latent structure is not reducible to these two labels.

\begin{table}[htbp]
  \centering
    \caption{Information criterion values versus number of clusters.}
  \begin{adjustbox}{max width=0.5\linewidth}
    \begin{tabular}{cccc}
      \toprule
      Number of clusters & SPIC & $\mathrm{BIC}_{\text{n}}$ & $\mathrm{BIC}_{\text{t}}$ \\
      \midrule
      1     & 9.192 & 9.178 & 9.151 \\
    2     & 9.154 & 9.174 & 9.141 \\
    3     & 9.139 & 9.139 & 9.076 \\
    4     & \textbf{9.046} & \textbf{9.081} & \textbf{9.041} \\
    5     & 9.120 & 9.133 & 9.090 \\
    6     & 9.207 & 9.183 & 9.141 \\
      \bottomrule
    \end{tabular}
     \label{tab:spic}%
       \end{adjustbox}
\end{table}

The exploratory diagnostic values are $D=0.194$ for the clusterwise normal specification and $D=0.144$ for the clusterwise $t$ specification. Table \ref{tab:label1} reports differences in the fitted mean profiles, with generally smaller reported bootstrap standard errors for the PML and PMML mean estimates. These entries summarize variability under the implemented resampling procedures; they do not establish smaller bias or greater accuracy relative to the unknown population parameters.
Group comparisons (Table~\ref{tab:label2}) show that diabetic patients $(ds=1)$ are significantly older and exhibit higher levels across all six biomarkers than non-diabetics. Participants with higher gravidity $(gs=1)$ are significantly older and have elevated 
$gtt$ and $dbp$ relative to those with $gs=0$. Diabetes prevalence is higher in Clusters 2 and 4 than in Clusters 1 and 3. Low gravidity remains more common in Cluster 2, whereas high gravidity predominates in Cluster 4. Notable nuances include: $bmi$ in Cluster 1 is comparable to Clusters 2 and 4 and exceeds Cluster 3; mean $age$ in Cluster 2 is similar to Clusters 1 and 3 but lower than Cluster 4. Tables \ref{tab:label2} and \ref{tab:label3} provide descriptive context for interpreting the OC-estimated patient profiles. Table \ref{tab:label3} reports within-cluster differences associated with diabetes and gravidity status. These descriptive summaries do not establish causal effects or invariance of the differences across clusters.

All three information criteria favor four clusters, and the ARI between OC and \(\mathrm{OC}_n\) is 0.958. Each estimation method also yields an older fitted profile and a profile with comparatively low BMI and skinfold thickness. These shared features and the high partition agreement do not establish identical patient groups or equivalent clinical consequences. The present analyses do not establish additional clinical insight from the discordant assignments or show that those assignments improve prediction or treatment decisions.

\begin{table}[htbp]
  \centering
  \caption{Mean vector estimates (bootstrap standard errors) from the PML, PMML, $\text{MML}_{\text{n}}$, and $\text{MML}_{\text{t}}$
methods, and covariance matrix estimates (bootstrap standard errors), with the upper and lower triangles corresponding to the PML and PMML estimates, respectively.}
\begin{adjustbox}{max width=1\linewidth}
    \begin{tabular}{ccccccccccccccccccccc}
    \toprule
   Method            &  \multicolumn{4}{c}{PML}  &      \multicolumn{4}{c}{PMML} &        \multicolumn{4}{c}{$\text{MML}_{\text{n}}$} &      \multicolumn{4}{c}{$\text{MML}_{\text{t}}$}     \\
   \cmidrule(rl){1-1} \cmidrule(rl){2-5} \cmidrule(rl){6-9} \cmidrule(rl){10-13} \cmidrule(rl){14-17}  
    Variable      & $\mu_1$    & $\mu_2$    & $\mu_3$    & $\mu_4$    & $\mu_1$    & $\mu_2$    & $\mu_3$    & $\mu_4$    & $\mu_1$    & $\mu_2$    & $\mu_3$    & $\mu_4$    & $\mu_1$    & $\mu_2$    & $\mu_3$    & $\mu_4$       \\
    \cmidrule(rl){1-17}
    \textit{gtt}   & -0.87 & 0.73  & -0.29 & 0.47  & -0.62 & 0.50   & -0.27 & 0.47  & -0.29 & 0.25  & -0.27 & 0.40   & -0.35 & 0.31  & -0.30  & 0.41 \\
          & (0.068) & (0.069) & (0.092) & (0.105) & (0.060) & (0.065) & (0.087) & (0.100) & (0.185) & (0.200) & (0.143) & (0.141) & (0.184) & (0.182) & (0.143) & (0.145) \\
    \textit{dpb}   & -0.24 & 0.22  & -0.37 & 0.43  & -0.38 & 0.28  & -0.25 & 0.46  & -0.66 & 0.48  & -0.15 & 0.51  & -0.63 & 0.47  & -0.12 & 0.50 \\
          & (0.100) & (0.112) & (0.107) & (0.100) & (0.088) & (0.099) & (0.099) & (0.084) & (0.224) & (0.212) & (0.171) & (0.098) & (0.204) & (0.186) & (0.168) & (0.101) \\
    \textit{tsft}  & 0.42  & 0.64  & -1.32 & 0.28  & 0.26  & 0.66  & -1.29 & 0.33  & 0.14  & 0.72  & -1.33 & 0.29  & 0.16  & 0.71  & -1.33 & 0.29 \\
          & (0.067) & (0.065) & (0.070) & (0.086) & (0.059) & (0.057) & (0.063) & (0.083) & (0.148) & (0.118) & (0.149) & (0.113) & (0.145) & (0.108) & (0.136) & (0.104) \\
    \textit{si}    & -0.64 & 0.69  & -0.36 & 0.35  & -0.44 & 0.60   & -0.42 & 0.27  & -0.24 & 0.41  & -0.43 & 0.30   & -0.26 & 0.46  & -0.47 & 0.30 \\
          & (0.092) & (0.085) & (0.097) & (0.107) & (0.075) & (0.078) & (0.092) & (0.098) & (0.180) & (0.169) & (0.132) & (0.130) & (0.160) & (0.163) & (0.139) & (0.123) \\
    \textit{bmi}   & 0.29  & 0.55  & -0.96 & 0.17  & 0.20   & 0.49  & -0.92 & 0.18  & 0.06  & 0.59  & -0.92 & 0.14  & 0.10   & 0.56  & -0.96 & 0.16 \\
          & (0.088) & (0.088) & (0.087) & (0.096) & (0.080) & (0.084) & (0.078) & (0.088) & (0.170) & (0.161) & (0.137) & (0.099) & (0.160) & (0.155) & (0.138) & (0.109) \\
    \textit{dpf}   & -0.24 & 0.49  & -0.24 & -0.01 & -0.17 & 0.33  & -0.24 & -0.01 & -0.09 & 0.34  & -0.27 & -0.01 & -0.08 & 0.32  & -0.29 & 0.00 \\
          & (0.105) & (0.099) & (0.103) & (0.129) & (0.104) & (0.093) & (0.092) & (0.115) & (0.194) & (0.216) & (0.157) & (0.142) & (0.170) & (0.195) & (0.156) & (0.141) \\
    \textit{age}   & -0.55 & -0.26 & -0.52 & 1.65  & -0.54 & -0.29 & -0.52 & 1.58  & -0.48 & -0.29 & -0.53 & 1.65  & -0.49 & -0.28 & -0.54 & 1.63 \\
          & (0.053) & (0.055) & (0.054) & (0.074) & (0.042) & (0.050) & (0.049) & (0.072) & (0.077) & (0.101) & (0.073) & (0.124) & (0.078) & (0.105) & (0.063) & (0.138) \\
    \bottomrule
    \end{tabular}%
    \end{adjustbox}  
  \centering
  \caption*{}
  \begin{adjustbox}{max width=0.6\linewidth}
    \renewcommand{\arraystretch}{1.15}
    \begin{tabular}{lccccccccl}
      \toprule
      & \textit{gtt} & \textit{dbp} & \textit{tsft} & \textit{si} & \textit{bmi} & \textit{dpf} & \textit{age} & \\
      \midrule
      & 0.58  & 0.04  & 0.05  & 0.31  & 0.11  & -0.04 & 0.07  & \textit{gtt} \\
          & (0.043) & (0.040) & (0.025) & (0.037) & (0.031) & (0.037) & (0.023) &  \\
          &       & 0.89  & 0.06  & -0.04 & 0.21  & -0.12 & 0.09  & \textit{dbp} \\
          &       & (0.074) & (0.029) & (0.041) & (0.042) & (0.054) & (0.027) &  \\
    \textit{gtt}   & 0.59  &       & 0.37  & 0.02  & 0.21  & -0.04 & 0.00     & \textit{tsft} \\
          & (0.044) &       & (0.041) & (0.032) & (0.034) & (0.034) & (0.022) &  \\
    \textit{dbp}   & 0.02  & 0.88  &       & 0.73  & 0.14  & -0.02 & 0.07  & \textit{si} \\
          & (0.038) & (0.067) &       & (0.070) & (0.041) & (0.048) & (0.024) &  \\
    \textit{tsft}  & 0.05  & 0.05  & 0.38  &       & 0.65  & 0.01  & 0.02  & \textit{bmi} \\
          & (0.023) & (0.029) & (0.041) &       & (0.053) & (0.043) & (0.024) &  \\
    \textit{si}    & 0.30   & -0.05 & 0.01  & 0.71  &       & 0.93  & 0.06  & \textit{dpf} \\
          & (0.038) & (0.041) & (0.032) & (0.072) &       & (0.069) & (0.027) &  \\
    \textit{bmi}   & 0.09  & 0.19  & 0.22  & 0.12  & 0.66  &       & 0.26  & \textit{age} \\
          & (0.030) & (0.039) & (0.035) & (0.041) & (0.055) &       & (0.022) &  \\
    \textit{dpf}   & -0.03 & -0.10  & -0.03 & -0.01 & 0.03  & 0.95  &       &  \\
          & (0.035) & (0.049) & (0.032) & (0.045) & (0.041) & (0.064) &       &  \\
    \textit{age}   & 0.08  & 0.08  & 0.00     & 0.08  & 0.03  & 0.08  & 0.26  &  \\
          & (0.021) & (0.026) & (0.023) & (0.023) & (0.021) & (0.024) & (0.021) &  \\
      \bottomrule
    \end{tabular}
  \end{adjustbox}
     \label{tab:label1}%
    \end{table}

\end{subsection}

\end{section}

\vspace{-0.2in}
\begin{section}{Concluding Remarks and Future Challenges}

\label{sec:sec7}

We study a general SCED for unsupervised learning and develop PML and PMML estimators, a posterior-probability clustering rule, and a SPIC for selecting the number of clusters. The asymptotic results are stated under the localization and regularity conditions in Section \ref{sec:sec3_1} and should be distinguished from the finite-sample performance of the practical IMR-initialized procedure.

The simulations identify configurations in which the proposed procedures perform favorably relative to the evaluated parametric benchmarks, and PMML can offer finite-sample advantages over PML in the examined settings. These advantages are not uniform. Additional experiments reveal sensitivity to the specified non-elliptical departure and substantial deterioration under strong imbalance. The applications provide interpretable fitted profiles without establishing comparative assignment accuracy or downstream decision benefits. Broader method comparisons and a systematic assessment of initialization sensitivity remain outside the empirical scope of this study.

Potential extensions include SCEDs with cluster-specific scatter matrices and density generators, and relationships to contaminated-normal and spurious-outlier approaches such as those of \cite{punzo2016parsimonious} and  \cite{garcia2008general}. Further work is also needed when the number of variables or clusters grows with sample size, and for formal model-adequacy testing. The likelihood-based comparison used in the applications is exploratory. Extending the methodology to an unrestricted multivariate density would require separate estimation and inferential development, including attention to data sparsity.

The Pima application motivates further investigation of whether the estimated patient profiles are useful for subgroup-aware evaluation of predictive models. Combining SCED-based representations with semi-supervised classification is another possible direction. Neither improvement in downstream prediction nor benefit for treatment decisions is evaluated in the present study.

%\cite{bouveyron2007high}
\end{section}
\end{spacing}

\bibliographystyle{biom}
\bibliography{reference}{}

\newpage

\appendix

\section*{Appendix}

\renewcommand{\thesection}{A}
\renewcommand{\theequation}{\thesection.\arabic{equation}}
\renewcommand{\thetable}{S\arabic{table}}
\renewcommand{\thefigure}{S\arabic{figure}}
\renewcommand{\thelem}{\thesection.\arabic{aplemma}}
\setcounter{table}{0}
\setcounter{equation}{0}
\setcounter{thm}{0}

In what follows, $\widehat{\mathcal{G}}_{c}^{\,\textsc{sp}}$ in $\hat{f}_h(x,c;\eta)$ is replaced by $\widetilde{\mathcal{G}}_{c}^{\,\textsc{sp}}$ for simplicity, without introducing additional notation.

\subsection{Technical Lemma}
\label{sec:A.1}

\begin{aplemma} 
    \label{lem1}
Under assumptions {A1$^{\prime}$} and {A5}, as $n\longrightarrow \infty$, $h\longrightarrow 0$, and $nh^{1+2m}\longrightarrow \infty$  for $m=0,1,2$,
    \begin{align*}
   &\sup_{\{(x,c,\eta) \in \mathcal{X}_{n} \times \mathcal{C} \times \mathcal{H}\}}\big\Vert\partial_\eta^m\hat{f}_h(x,c;\eta)-f^{[m]}(x,c;\eta)\big\Vert \\
 &=\sup_{\{(x,c,\eta) \in \mathcal{X}_{n} \times \mathcal{C} \times \mathcal{H}\}} \bigg\Vert \frac{1}{n}\sum_{i=1}^n \xi^{[m]}_{i}(x,c;\eta)  \bigg\Vert=O(h^{2})+o_p\Bigg(\sqrt{\frac{\log n}{nh^{1+2m}}}\,\Bigg),
\end{align*}
where 
$\xi^{[m]}_{i}(x,c;\eta)=\partial_{\eta}^m \big(\omega_{c}w(y_{c})K_{h}(Y_{i},y_{c}) \big)-f^{[m]}(x,c;\eta)$, $i=1,\dots,n$.
\end{aplemma}
\begin{proof}
By the general Leibniz rule,
\begin{align*}
 \partial_{\eta}^m \big(\omega_{c} w(y_{c}) K_h(Y,y_{c})\big) =\sum_{s=0}^m  \binom{m}{s}\big[\partial_{\eta}^s (\omega_{c} w(y_{c}))\big] \sum_{q=0}^{m-s}\sum_{r=0}^1 \frac{1}{h^{q+1}}K^{(q)}\bigg(\frac{Y+(-1)^r y_{c}}{h}\bigg)\big(M_{q,r}^{[m-s]}(X,x,c)\big)^{\top}.
\end{align*}
Since $K^{(q)}$ is bounded and $(1+\psi(y_{c})^{\frac{p}{2}})^a\leq 1$ for $a<0$, there exists a polynomial function $P_m(x)$ of degree $2m$ in $|x_1|,\dots,|x_p|$ such that 
\begin{align}
\sup_{\{(x,c,\eta) \in \mathcal{X}_{n} \times \mathcal{C} \times \mathcal{H}\}}\big\Vert\partial_{\eta}^m \big(\omega_{c} w(y_{c}) K_h(Y,y_{c})\big)\big\Vert_1 \leq \frac{ P_m(X)}{h^{m+1}},~m=0,1,2, \label{A.0.1.5}
\end{align}
where $\vvec(\cdot)$ denotes the vectorization operator. 
For $m=0,1,2$ and $s=1,2$, it follows from (\ref{A.0.1.5}) and H\"{o}lder's inequality  that under assumption {A1$^{\prime}$}, 
\begin{align}
&\sup_{\{(x,c,\eta) \in \mathcal{X}_{n} \times \mathcal{C} \times \mathcal{H}\}}\Big\|\E \Big[ \Big(\vvec \Big(\partial_{\eta}^m \big(\omega_{c} w(y_{c}) K_h(Y,y_{c})\big)\Big)\Big)^{\otimes 2} \Big] \Big\|=O\bigg(\frac{1}{h^{2m+1}}\bigg) \label{A.0.1.7}
\end{align} 
and
\begin{align}
&\sup_{\{(x,c,\ell,\eta) \in \mathcal{X}_{n} \times \mathcal{C} \times \mathcal{C} \times \mathcal{H}\}}\Big\|\E \Big[ \Big(\vvec \Big(\partial_{\eta}^m \big(\omega_{c} w(y_{c}) K_h(Y_{\ell},y_{c})\big)\Big)\Big)^{\otimes s}\Big] \Big\|=O\bigg(\frac{1}{h^{ ms+s -1}}\bigg). \label{A.0.r.1}
\end{align} 
Applying Taylor's theorem, we obtain
\begin{align}
    &\E\Big[ \partial_{\eta}^m  \big(\omega_c w(y_{c})K_{h}(Y,y_{c})\big)\Big]  \nonumber \\
    &=\E \Bigg[\sum_{s=0}^m  \binom{m}{s}\big(\partial_{\eta}^s \omega_c w(y_c)\big) \sum_{q=0}^{m-s}\sum_{r=0}^1 \frac{1}{h^{q+1}}K^{(q)}\bigg(\frac{Y+(-1)^r y_c}{h}\bigg)M_{q,r}^{[m-s]}(X,x,c)\Bigg] \nonumber \\
    &=\E\Bigg[ \sum_{s=0}^m  \binom{m}{s}\big(\partial_{\eta}^s \omega_c w(y_c)\big) \sum_{q=0}^{m-s}\sum_{r=0}^1 \frac{1}{h^{q+1}}K^{(q)}\bigg(\frac{Y+(-1)^r y_c}{h}\bigg) \E\Big[M_{q,r}^{[m-s]}(X,x,c)\Big|Y\Big]\Bigg] \nonumber \\
    &=\E \Bigg[\sum_{s=0}^m  \binom{m}{s}\big(\partial_{\eta}^s \omega_c w(y_c)\big) \sum_{q=0}^{m-s}\sum_{r=0}^1 \frac{1}{h^{q+1}}K^{(q)}\bigg(\frac{Y+(-1)^r y_c}{h}\bigg)\E\Big[M_{q,r}^{[m-s]}(X,x,c)\Big|Y\Big]\Bigg] \nonumber \\
    &=\sum_{s=0}^m \binom{m}{s}\big(\partial_{\eta}^s \omega_c w(y_c)\big) \sum_{q=0}^{m-s}\sum_{r=0}^1 \frac{1}{h^{q+1}}\int K^{(q)}\bigg(\frac{y+(-1)^r y_c}{h}\bigg)\E\Big[M_{q,r}^{[m-s]}(X,x,c)\Big|Y=y\Big]g(y;\theta) dy\nonumber \\
    &= f^{[m]}(x,c;\eta)+r^{[m]}(x,c,\eta), \label{A.0.3}
\end{align}
where $\sup_{\{(x,c,\eta) \in \mathcal{X}_{n} \times \mathcal{C} \times \mathcal{H}\}}\Vert r^{[m]}(x,c;\eta)\Vert =O(h^2)$, $m=0,1,2$, under assumption {A5}.
Since $\partial_{\eta}^m w(y_{c})$ and $\partial_{\eta}^m K_{h}(Y,y_{c})$ have bounded variation in $x$ uniformly over $(c,\eta)\in\mathcal{C} \times \mathcal{H}$ for $m=0,1,2$, Theorem 22 in \cite{nolan1987u} implies that the classes
\begin{align*}
    \bigg\{\partial_{\eta}^m \big(\omega_{c} w(y_{c}) K_h(Y,y_{c})\big):(x,c,\eta) \in \mathcal{X}_{n} \times \mathcal{C} \times \mathcal{H}\bigg\},~m=0,1,2,~r=0,1, 
\end{align*}
and 
\begin{align*}
    \bigg\{\partial_{\eta}^m \big(\omega_{c} w(y_{c}) K_h(Y_{\ell},y_{c})\big):(x,c,\ell,\eta) \in \mathcal{X}_{n} \times \mathcal{C} \times \mathcal{C} \times \mathcal{H}\bigg\},~m=0,1,2,~r=0,1, 
\end{align*}
are Euclidean.
Combining this fact with (\ref{A.0.1.7})--(\ref{A.0.3}) and applying Theorem II.37 in \cite{pollard1984convergence} together with the triangle inequality yield
 \begin{align}
    &\sup_{\{(x,c,\eta) \in \mathcal{X}_{n} \times \mathcal{C} \times \mathcal{H}\}}\big\Vert\partial_\eta^m\bar{f}_h(x,c;\eta)-f^{[m]}(x,c;\eta)\big\Vert \leq O(h^{2})+o_p\Bigg(\sqrt{\frac{\log n}{nh^{1+2m}}}\,\Bigg)\label{A.0.4}
\end{align}
and
\begin{align}
&    \sup_{\{(x,c,\ell,\eta) \in \mathcal{X}_{n} \times \mathcal{C} \times  \mathcal{C}\times \mathcal{H}\}}\Big\Vert\frac{1}{n} \sum_{i=1}^{n} \sum_{c=1}^k \partial_{\eta}^m \Big(\omega_{c} w(y_{c})K_{h}(Y_{i\ell},y_{c}) \Big) \Big\Vert \nonumber\\
&\leq O\Big(\frac{1}{h^m}\Big)+o_p\Bigg(\sqrt{\frac{\log n}{nh^{1+2m}}}\,\Bigg),~m=0,1,2, \label{A.0.r.4}
\end{align}
where $\bar{f}_h(x,c;\eta)=\omega_{c}w(y_{c}) \sum_{i=1}^{n} K_{h}(Y_{i},y_{c})/n$.
By the latent cluster recovery result in (\ref{3.1.7_0}), we have
\begin{align}
   \sup_{\{i,c\}} \big| I\big(i\in \widetilde{\mathcal{G}}_{c}^{\,\textsc{sp}}\big)-I\big(i\in \mathcal{G}^{\text{o}}_{c}\big) \big|=o_{p}\Big(\frac{1}{n^2}\Big). \label{A.0.5}
\end{align}
Using (\ref{A.0.r.4}) and (\ref{A.0.5}) and applying the triangle inequality, it follows that
\begin{align}
\sup_{\{(x,c,\eta) \in \mathcal{X}_{n} \times \mathcal{C} \times \mathcal{H}\}}\big\Vert \partial_\eta^m\hat{f}_h(x,c;\eta)- \partial_\eta^m \bar{f}_h(x,c;\eta) \big\Vert=o_p\Big(\frac{1}{n}\Big),~ m=0,1,2.\label{A.0.8}
\end{align}
Lemma \ref{lem1} follows directly from 
(\ref{A.0.4}) and (\ref{A.0.8}).

\end{proof}

\subsection{Proof of Theorem \ref{Thm3.1}}

\label{Thm3.1pf}

Under assumptions A2 and A3, the separation penalty estimator $(\tilde{\beta}^{\textsc{sp}},\tilde{\mu}^{\textsc{sp}})$ possesses the local oracle property if, with probability converging to one, \begin{align}
    \textsc{SS}_{\text{sp}}(\beta,\mu;\lambda) > \textsc{SS}_{\text{sp}}\big(\tilde{\beta}^{\text{or}},\tilde{\mu}^{\text{or}};\lambda\big)~\forall (\beta,\mu) \neq (\tilde{\beta}^{\text{or}},\tilde{\mu}^{\text{or}}).\label{A.1.1}
\end{align}
It suffices to show that, with probability converging to one,
\begin{align}
 \textsc{SS}_{\text{sp}}(\beta^*,\mu^*;\lambda) > \textsc{SS}_{\text{sp}}\big(\tilde{\beta}^{\text{or}},\tilde{\mu}^{\text{or}};\lambda\big)  \text{ and }
   \textsc{SS}_{\text{sp}}(\beta,\mu;\lambda) \geq \textsc{SS}_{\text{sp}}(\beta^*,\mu^*;\lambda),  \label{A.1.3}
\end{align}
 where 
\begin{align*}
    \mu^*_{c}=\frac{1}{|\mathcal{G}_c^{\text{o}}|}\sum_{\{i \in \mathcal{G}_{c}^{\text{o}}\}} \beta_i\text{ and }\beta_i^*=\sum_{c=1}^k \mu^*_{c} I\big(i \in \mathcal{G}_{c}^{\text{o}}\big).
\end{align*}

By $\lambda \sum_{i=1}^n \min_{\{c\in \mathcal{C}\}} \| W^{1/2} (\beta^*_i - \mu^*_c) \|_1=0$,
\begin{align}
&\textsc{SS}_{\text{sp}}\big(\beta^*,\mu^*; \lambda\big) =
\frac{1}{2} \sum_{i=1}^n \big(X_{i}-\beta_{i}^*\big)^{\top}W(X_{i}-\beta_{i}^*) \nonumber \\
&=\frac{1}{2} \sum_{i=1}^n \sum^{k}_{c=1}I(C_{i}=c)(X_{i}-\mu_{c}^*)^{\top}W(X_{i}-\mu_{c}^*)\stackrel{\triangle}{=}\textsc{SS}_{w}(\mu^*).\label{A.1.4}
\end{align}
Since $\tilde{\mu}^{\text{or}}$ is the unique minimizer of $\textsc{SS}_{w}(\mu^*)$, it follows that
\begin{align}
\textsc{SS}_{\text{sp}}\big(\beta^*,\mu^*; \lambda\big) > \textsc{SS}_{\text{sp}}\big(\tilde{\beta}^{\text{or}},\tilde{\mu}^{\text{or}}; \lambda\big)~\forall (\beta^{*},\mu^{*}) \neq (\tilde{\beta}^{\text{or}},\tilde{\mu}^{\text{or}}). \label{A.1.5}
\end{align} 
We expand the first term in $\textsc{SS}_{\text{sp}}(\beta,\mu;\lambda)$ as a first-order Taylor series around $\beta = \beta^*$, yielding the following expression:
\begin{align}
    \textsc{SS}_{\text{sp}}(\beta,\mu;\lambda) - \textsc{SS}_{\text{sp}}(\beta^*,\mu^*;\lambda) &= - \sum_{i=1}^n\big(X_i - \bar{\beta}_i\big)^{\top} W \big(\beta_i - \beta_i^*\big) + \lambda \sum_{i=1}^n  \min_{c} \big\| W^{\frac{1}{2}} (\beta_i - \mu_c)\big\|_1 \nonumber \\ 
    &\overset{\Delta}{=} I_1 + I_2, \label{A.1.6}
\end{align}
where $\bar{\beta}_i = \varsigma \beta_i + (1 - \varsigma) \beta_i^*$ for some $\varsigma \in (0,1)$, $i = 1, \dots, n$.  

A direct calculation leads to the following expression for $I_{1}$:
\begin{align}
   & I_1=\sum_{c=1}^k \sum_{\{i \in \mathcal{G}_{c}^{\text{o}}\}} \big(-X_i+\bar{\beta_i}\big)^{\top} W \sum_{\{j \in \mathcal{G}_{c}^{\text{o}}\}} \frac{\beta_i-\beta_j}{| \mathcal{G}_{c}^{\text{o}} |}\nonumber  \\
    &=-\sum_{c=1}^k\sum_{\{i,j \in \mathcal{G}_{c}^{\text{o}}, i<j\}} \frac{\big(-X_i+\bar{\beta_i}-\big(-X_j+\bar{\beta_j}\big)\big)^{\top} W (\beta_i-\beta_j)}{| \mathcal{G}_{c}^{\text{o}} |}\label{A.1.7}.
\end{align}
By the Cauchy-Schwarz inequality and Markov's inequality, under assumption {A1}, we obtain
\begin{align}
P\big(\big|I_1\big| > n^{\alpha} c_{n}\big)
 &\leq P\Bigg(\sum_{c=1}^k\sum_{\{i,j \in \mathcal{G}_{c}^{\text{o}}, i<j\}} \frac{\big\Vert W^{\frac{1}{2}}\big(-X_i+\bar{\beta_i}-\big(-X_j+\bar{\beta_j}\big)\big)\big\Vert \big\Vert W^{\frac{1}{2}} (\beta_i-\beta_j)\big\Vert}{| \mathcal{G}_{c}^{\text{o}} |} \nonumber >n^{\alpha} c_{n}\Bigg) \nonumber \\
&\leq \frac{1}{n^{\alpha}c_{n}}\sum_{c=1}^k\sum_{\{i,j \in \mathcal{G}_{c}^{\text{o}}, i<j\}} \frac{\E\Big[\big\Vert W^{\frac{1}{2}}\big(-X_i+\bar{\beta_i}-\big(-X_j+\bar{\beta_j}\big)\big)\big\Vert \Big]\big\Vert W^{\frac{1}{2}} (\beta_i-\beta_j)\big\Vert}{| \mathcal{G}_{c}^{\text{o}} |} \nonumber\\
&\longrightarrow 0 \text{ as }n \longrightarrow \infty \text{ for every fixed } \alpha>0,\label{A.1.8}
\end{align}
where
\begin{align*}
c_{n}=\sum_{c=1}^k \sum_{\{i,j \in \mathcal{G}_{c}^{\text{o}}, i<j\}} \frac{\big\Vert W^{\frac{1}{2}}(\beta_i-\beta_j) \big\Vert}{| \mathcal{G}_{c}^{\text{o}}|}.
\end{align*} 
Consequently, by (\ref{A.1.8}), 
\begin{align}
    I_1 \geq -n^{\alpha} c_{n}(1+o_p(1)). \label{A.1.10}
\end{align}
For $i \in \mathcal{G}^{\text{o}}_{c_1}$ and $c_1 \neq c_2$, it follows from the triangle inequality that under assumption {A2}
\begin{align}
      & \big\Vert W^{\frac{1}{2}}(\beta_{i} - \mu_{c_2}) \big\Vert_1  \geq  \big\Vert W^{\frac{1}{2}}(\mu_{c_1}^{\text{o}} - \mu_{c_2}^{\text{o}}) \big\Vert_1 - \big\Vert W^{\frac{1}{2}}(\mu_{c_2} - \mu_{c_2}^{\text{o}}) \big\Vert_1 -\big\Vert W^{\frac{1}{2}}(\beta_{i} - \beta_{i}^{\text{o}}) \big\Vert_1 > \frac{D_{\mu}}{2} ~ \text{and}  \nonumber  \\
      & \big\Vert W^{\frac{1}{2}}(\beta_{i} - \mu_{c_1}) \big\Vert_1 \leq \big\Vert W^{\frac{1}{2}}(\mu_{c_1} - \mu_{c_1}^{\text{o}})\big\Vert_1
       +  \big\Vert W^{\frac{1}{2}}(\beta_{i} - \beta_{i}^{\text{o}}) \big\Vert_1 < \frac{D_\mu}{2}.\label{A.1.11}
\end{align}
The bounds in (\ref{A.1.11}) imply
\begin{eqnarray}
    I_2=\lambda \sum_{c=1}^k \sum_{\{i \in \mathcal{G}_{c}^{\text{o}}\}} \big\Vert W^{\frac{1}{2}}(\beta_i-\mu_{c}) \big\Vert_1
    =\lambda \sum_{c=1}^k \sum_{\{i,j \in \mathcal{G}_{c}^{\text{o}}\}} \frac{ \big\Vert W^{\frac{1}{2}}(\beta_i-\mu_{c}) \big\Vert_1 +\big\Vert W^{\frac{1}{2}}(\beta_j-\mu_{c}) \big\Vert_1 }{2| \mathcal{G}_{c}^{\text{o}} |}.~~~\label{A.1.12} 
\end{eqnarray}
Applying the triangle inequality to (\ref{A.1.12}) and the inequality $\| \cdot \|_1 \geq \| \cdot \|$, we obtain
\begin{align}
   & I_2 \geq  \lambda \sum_{c=1}^k \sum_{\{i,j \in \mathcal{G}_{c}^{\text{o}}\}} \frac{\big\Vert W^{\frac{1}{2}}(\beta_i-\beta_j) \big\Vert_1}{2 | \mathcal{G}_{c}^{\text{o}} |}   \geq \lambda \,c_{n}.
    \label{A.1.13} 
\end{align}\vspace{-0.01in}Substituting (\ref{A.1.10}) and (\ref{A.1.13}) into (\ref{A.1.6}) and invoking assumption {A3}, we conclude that, with probability converging to one,
\begin{align}
    \textsc{SS}_{\text{sp}}(\beta,\mu;\lambda) - \textsc{SS}_{\text{sp}}(\beta^*,\mu^*;\lambda) 
    \geq \big(\lambda - n^{\alpha}  \big)c_{n}(1+o_p(1))\geq 0. \label{A.1.14}
\end{align}
Combining (\ref{A.1.5}) and (\ref{A.1.14}) establishes (\ref{A.1.3}) and, hence, Theorem \ref{Thm3.1}.

\subsection{Proof of Theorems \ref{Thm3.2}}
\label{Thm3.2-3.pf}

Define 
\begin{align*}
&pS_s(\eta)=\partial_{\eta}\,p\ell_{s}(\eta), ~pI_{s}(\eta)=-\partial^{2}_{\eta}\,p\ell_{s}(\eta),~ s=1,2,\\
&\ell_1(\eta)=\sum^{n}_{i=1}\sum^{k}_{c=1}I(i\in\mathcal{G}_{c}^{\text{o}})\log \big(f^{[0]}(X_{i},c;\eta)\big), \text{ and }\ell_2(\eta)=\sum^{n}_{i=1}\log \big(f^{[0]}(X_{i};\eta)\big).
\end{align*} 
Applying the triangle inequality, we obtain
\begin{align}
  \sup_{\{\eta \in \mathcal{H}\}} \bigg|\frac{1}{n}p{\ell}_{1}(\eta)-\frac{1}{n}\ell_1(\eta)\bigg|\leq &~\sup_{\{i,c\}}\Big( \Big| I(i\in \widetilde{\mathcal{G}}_{c}^{\textsc{sp}})-I(i\in \mathcal{G}_c^{\text{o}}) \Big|\Big) \bigg(\frac{1}{n} \sum_{i=1}^n \sum_{c=1}^k \sup_{\{\eta \in \mathcal{H}\}}  \Big| \log(\hat{f}_h(X_{i},c;\eta)) \Big| \bigg) \nonumber \\
   &~+ \sup_{\{(x,c,\eta) \in \mathcal{X}_{n} \times \mathcal{C} \times \mathcal{H}\}}\bigg|\log\bigg(1+\frac{\hat{f}_h(x,c;\eta)-f^{[0]}(x,c;\eta)}{f^{[0]}(x,c;\eta)}\bigg)\bigg|\text{ and}
    \label{A.2.1}\\
 \sup_{\{\eta \in \mathcal{H}\}} \bigg|\frac{1}{n}p{\ell}_{2}(\eta)-\frac{1}{n}\ell_2(\eta)\bigg|\leq &~ \sup_{\{(x,\eta) \in \mathcal{X}_{n} \times \mathcal{H}\}}\bigg|\log\bigg(1+\frac{\hat{f}_h(x;\eta)-f^{[0]}(x;\eta)}{f^{[0]}(x;\eta)}\bigg)\bigg|. \label{A.3.1}
\end{align}
By (\ref{A.0.5}), the uniform consistency of $\hat{f}_h(x,c;\eta)$ to $f^{[0]}(x,c;\eta)$ in Lemma \ref{lem1}, and $\log(1+x)=O(x)$ for $x=o(1)$, it follows that under assumptions {A4}--{A6},
\begin{align}
    &\sup_{\{(x,c,\eta) \in \mathcal{X}_{n} \times \mathcal{C} \times \mathcal{H}\}}\bigg|\log\bigg(1+\frac{\hat{f}_h(x,c;\eta)-f^{[0]}(x,c;\eta)}{f^{[0]}(x,c;\eta)}\bigg)\bigg| \nonumber \\
    &=O\Bigg(\frac{\sup_{\{(x,c,\eta) \in \mathcal{X}_{n} \times \mathcal{C} \times \mathcal{H}\}}\big|\hat{f}_h(x,c;\eta)-f^{[0]}(x,c;\eta)\big|}{\inf_{\{(x,c,\eta) \in \mathcal{X}_{n} \times \mathcal{C} \times \mathcal{H}\}} f^{[0]}(x,c;\eta)}\Bigg)=o_p(1),\label{A.2.3} \\
    & \sup_{\{(x,\eta) \in \mathcal{X}_{n} \times \mathcal{H}\}}\bigg|\log\bigg(1+\frac{\hat{f}_h(x;\eta)-f^{[0]}(x;\eta)}{f^{[0]}(x;\eta)}\bigg)\bigg|   \nonumber\\
    &=O\Bigg(\frac{\sup_{\{(x,\eta) \in \mathcal{X}_{n} \times \mathcal{H}\}}\big|\hat{f}_h(x;\eta)-f^{[0]}(x,\eta)\big|}{\inf_{\{(x,\eta) \in \mathcal{X}_{n} \times \mathcal{H}\}} f^{[0]}(x;\eta)}\Bigg) =o_p(1)\text{ and} \label{A.3.2}\\
    & \sup_{\{i,c\}} \Big| I(i\in \widetilde{\mathcal{G}}_{c}^{\textsc{sp}})-I(i\in \mathcal{G}_c^{\text{o}}) \Big| \bigg( \frac{1}{n} \sum_{i=1}^n \sum_{c=1}^k \sup_{\{\eta \in \mathcal{H}\}} \Big| \log(\hat{f}_h(X_{i},c;\eta)) \Big|\bigg) =o_p\bigg(\frac{1}{n}\bigg). \label{A.2.2} 
\end{align}
Substituting (\ref{A.2.2}) and (\ref{A.2.3}) into (\ref{A.2.1}) and substituting (\ref{A.3.2}) into (\ref{A.3.1}), we deduce that
\begin{align}
    \sup_{\{\eta \in \mathcal{H}\}} \bigg|\frac{1}{n}p{\ell}_{s}(\eta)-\frac{1}{n}\ell_s(\eta)\bigg|=o_p(1), s=1,2. \label{A.2.4}
\end{align}
Under assumption {A5}, $f^{[0]}(x,c;\eta)$ has bounded variation in $x$, uniformly over $(c,\eta) \in\mathcal{C}\times \mathcal{H} $.
It follows from Theorem 22 in \cite{nolan1987u} that the classes
\begin{align*}
    \big\{f^{[0]}(x,c;\eta):(c,\eta) \in\mathcal{C}\times \mathcal{H}  \big\}\text{ and }  \big\{f^{[0]}(x;\eta): \eta \in \mathcal{H} \big\}
\end{align*} 
are Euclidean. Applying Corollary 4 in \cite{sherman1994maximal}, we obtain
\begin{align}
 &\sup_{\{\eta \in \mathcal{H}\}} \bigg|   \frac{1}{n} \ell_{1}(\eta)-\E\big[\log( f^{[0]}(X,C;\eta))\big]  \bigg|=O_p\Big(\frac{1}{\sqrt{n}}\Big)\text{ and }\nonumber\\
 &\sup_{\{\eta \in \mathcal{H}\}} \bigg|   \frac{1}{n} \ell_{2}(\eta)-\E\big[\log( f^{[0]}(X;\eta))\big]  \bigg|=O_p\Big(\frac{1}{\sqrt{n}}\Big).\label{A.2.5}
\end{align}
Thus, combining (\ref{A.2.4}) and (\ref{A.2.5}) yields
\begin{align}
&\sup_{\{\eta \in \mathcal{H}\}} \bigg|   \frac{1}{n} p\ell_{1}(\eta)-\E\big[\log( f^{[0]}(X,C;\eta))\big]  \bigg|=o_p(1)\text{ and }\nonumber\\
&\sup_{\{\eta \in \mathcal{H}\}} \bigg|   \frac{1}{n} p\ell_{2}(\eta)-\E\big[\log( f^{[0]}(X;\eta))\big]  \bigg|=o_p(1).\label{A.2.6}
\end{align}

By Jensen's inequality, it follows that
\begin{align}
 &\E\big[\log(f^{[0]}(X,C;\eta))\big]-\E\big[\log( f^{[0]}(X,C;\eta^{\text{o}}))\big]\leq \log\E\Bigg[\frac{f^{[0]}(X,C;\eta)}{f^{[0]}(X,C;\eta^{\text{o}})}\Bigg]=\log 1=0\text{ and}\nonumber\\
  &\E\big[\log(f^{[0]}(X;\eta))\big]-\E\big[\log( f^{[0]}(X;\eta^{\text{o}}))\big]\leq \log\E\Bigg[\frac{f^{[0]}(X;\eta)}{f^{[0]}(X;\eta^{\text{o}})}\Bigg]=\log 1=0, \label{A.2.7}
  \end{align}
with strict inequalities for $\eta \neq \eta^{\text{o}}$. Hence, under assumption A6, $\eta^{\text{o}}$ is the unique maximizer of $Q_1(\eta)=\E\big[\log( f^{[0]}(X,C;\eta))\big]$ and $Q_2(\eta)=\E\big[\log (f^{[0]}(X;\eta))\big]$. Combining (\ref{A.2.6}) and (\ref{A.2.7}) yields
\begin{align*}
&\frac{1}{n}p\ell_1(\tilde{\eta}^{\textsc{pml}})- \frac{1}{n} p\ell_1(\eta^{\text{o}})=Q_1(\tilde{\eta}^{\textsc{pml}})-Q_1(\eta^{\text{o}})+o_p(1)\text{ and}\\
&\frac{1}{n}p\ell_2(\tilde{\eta}^{\textsc{pmml}})- \frac{1}{n} p\ell_2(\eta^{\text{o}})=Q_2(\tilde{\eta}^{\textsc{pmml}})-Q_2(\eta^{\text{o}})+o_p(1),
\end{align*}
where $Q_1(\tilde{\eta}^{\textsc{pml}})-Q_1(\eta^{\text{o}})<0$ and $Q_2(\tilde{\eta}^{\textsc{pmml}})-Q_2(\eta^{\text{o}})<0$. On the other hand, by the definition of $\tilde{\eta}^{\textsc{pml}}$ and $\tilde{\eta}^{\textsc{pmml}}$ as maximizers of $p\ell_1(\eta)$ and $p\ell_2(\eta)$, respectively, we have $ p\ell_1(\tilde{\eta}^{\textsc{pml}}) \geq p\ell_1(\eta^{\text{o}})$ and $p\ell_2(\tilde{\eta}^{\textsc{pmml}}) \geq p\ell_2(\eta^{\text{o}})$. Therefore,
\begin{align*}
Q_1(\eta^{\text{o}})-Q_1(\tilde{\eta}^{\textsc{pml}})=o_p(1)\text{ and }Q_2(\eta^{\text{o}})-Q_2(\tilde{\eta}^{\textsc{pmml}})=o_p(1).
\end{align*}
The consistency of $\tilde{\eta}^{\textsc{pml}}$ to $\eta^{\text{o}}$ and $\tilde{\eta}^{\textsc{pmml}}$ to $\eta^{\text{o}}$ then follows from the uniqueness of the maximizers of  $Q_1(\eta)$ and $Q_2(\eta)$ over $ \mathcal{H}$.

Let $v \in \mathbb{R}^{d^{\eta}}$ be an arbitrary constant vector, where $d^{\eta}=\dim(\mathcal{H})$. A
first-order Taylor expansion of $v^{\top}pS_1(\tilde{\eta}^{\textsc{pml}})/\sqrt{n}$ around $\eta^{\text{o}}$ yields
\begin{align}
\frac{1}{\sqrt{n}}v^{\top}pS_1(\tilde{\eta}^{\textsc{pml}})=\frac{1}{\sqrt{n}}v^{\top}pS_1(\eta^{\text{o}})-v^{\top}\Big(\frac{1}{n}pI_1(\bar{\eta}^{\textsc{pml}})\Big)\sqrt{n}(\tilde{\eta}^{\textsc{pml}}-\eta^{\text{o}}), \label{A.2.8}
\end{align}
where $\bar{\eta}^{\textsc{pml}}$ lies on the line segment between $\tilde{\eta}^{\textsc{pml}}$ and $\eta^{\text{o}}$.
Similarly, a first-order expansion of $v^{\top}pS_2$ $(\tilde{\eta}^{\textsc{pmml}})/\sqrt{n}$ around $\eta^{\text{o}}$ gives
\begin{align}
\frac{1}{\sqrt{n}}v^{\top}pS_2(\tilde{\eta}^{\textsc{pmml}})=\frac{1}{\sqrt{n}}v^{\top}pS_2(\eta^{\text{o}})-v^{\top}\Big(\frac{1}{n}pI_2(\bar{\eta}^{\textsc{pmml}})\Big)\sqrt{n}(\tilde{\eta}^{\textsc{pmml}}-\eta^{\text{o}}),\label{A.3.7}
\end{align}
where $\bar{\eta}^{\textsc{pmml}}$ lies on the line segment between $\tilde{\eta}^{\textsc{pmml}}$ and $\eta^{\text{o}}$. 
Rewriting the normalized pseudo-score vectors in (\ref{A.2.8}) and (\ref{A.3.7}), we obtain the decomposition
\begin{align}
    \frac{1}{\sqrt{n}}pS_{1}(\eta^{\text{o}})=&\frac{1}{\sqrt{n}}\sum^{n}_{i=1}\sum^{k}_{c=1}I(i\in\mathcal{G}_{c})\frac{f^{[1]}(X_{i},c;\eta^{\text{o}})}{f^{[0]}(X_{i},c;\eta^{\text{o}})}+     \frac{1}{\sqrt{n}} \sum_{i=1}^n \sum_{c=1}^k \big(I\big(i\in \widetilde{\mathcal{G}}_{c}^{\,\textsc{sp}}\big)-I(i\in \mathcal{G}_c)\big)\frac{\partial_{\eta}\hat{f}_h(X_{i},c;\eta^{\text{o}})}{\hat{f}_h(X_{i},c;\eta^{\text{o}})}  \nonumber \\
    &+\frac{1}{\sqrt{n}}\sum_{i=1}^n \sum_{m=0}^{1} \frac{(-1)^{1-m}f^{[1-m]}(X_i,C_i;\eta^{\text{o}})\big(\partial_{\eta}^{m}\hat{f}_h(X_i,C_i;\eta^{\text{o}})-f^{[m]}(X_i,C_i;\eta^{\text{o}})\big)} {f^{[0]}(X_i,C_i;\eta^{\text{o}})\big(f^{[0]}(X_i,C_i;\eta^{\text{o}})+\hat{f}_h(X_i,C_i;\eta^{\text{o}})-f^{[0]}(X_i,C_i;\eta^{\text{o}})\big)} \nonumber \\
      \stackrel{\triangle}{=}& \frac{1}{\sqrt{n}}S_{1}(\eta^{\text{o}})+I\!I_1+I\!I_2\text{ and}\label{A.2.9} 
      \end{align}
      \begin{align}
    \frac{1}{\sqrt{n}}pS_{2}(\eta^{\text{o}})=&\frac{1}{\sqrt{n}}\sum_{i=1}^n \sum_{m=0}^{1} \frac{(-1)^{1-m}f^{[1-m]}(X_i;\eta^{\text{o}})\big(\partial_{\eta}^{m}\hat{f}_h(X_i;\eta^{\text{o}})-f^{[m]}(X_i;\eta^{\text{o}})\big)} {f^{[0]}(X_i;\eta^{\text{o}})\big(f^{[0]}(X_i;\eta^{\text{o}})+\hat{f}_h(X_i;\eta^{\text{o}})-f^{[0]}(X_i;\eta^{\text{o}})\big)} \nonumber \\  
    &+\frac{1}{\sqrt{n}}\sum^{n}_{i=1}\frac{f^{[1]}(X_{i};\eta^{\text{o}})}{f^{[0]}(X_{i};\eta^{\text{o}})} 
      \stackrel{\triangle}{=}I\!I_{3}+ \frac{1}{\sqrt{n}}S_2(\eta^{\text{o}}). \label{A.3.8}
\end{align}
By the triangle inequality, (\ref{A.0.5}), and Lemma \ref{lem1} with $m=0,1$, it follows that under assumptions {A4} -- {A7},
\begin{align}
    &|I\!I_1| 
    \leq k \sup_{\{i,c\}} \big|\sqrt{n}\big( I\big(i\in \widetilde{\mathcal{G}}_{c}^{\textsc{sp}}\big)-I\big(i\in \mathcal{G}_c^{\text{o}}\big)\big) \big| \sup_{\{(x,c) \in \mathcal{X}_{n} \times \mathcal{C}\}}  \Bigg|   \frac{\partial_{\eta}\hat{f}_h(x,c;\eta^{\text{o}})}{\hat{f}_h(x,c;\eta^{\text{o}})}\Bigg|\nonumber\\
    &\leq k \sup_{\{i,c\}} \big|\sqrt{n}\big( I\big(i\in \widetilde{\mathcal{G}}_{c}^{\textsc{sp}}\big)-I\big(i\in \mathcal{G}_c^{\text{o}}\big)\big) \big| \nonumber \\
    &  \cdot\frac{\sup_{\{(x,c) \in \mathcal{X}_{n} \times \mathcal{C}\}}\big|\partial_{\eta}\hat{f}_h(x,c;\eta^{\text{o}})-f^{[1]}(x,c;\eta^{\text{o}})\big|+\sup_{\{(x,c) \in \mathcal{X}_{n} \times \mathcal{C}\}}\big|f^{[1]}(x,c;\eta^{\text{o}})\big|}{\inf_{\{(x,c) \in \mathcal{X}_{n} \times \mathcal{C}\}}f(x,c;\eta^{\text{o}})-\sup_{\{(x,c) \in \mathcal{X}_{n} \times \mathcal{C}\}}\big|\hat{f}_h(x,c;\eta^{\text{o}})-f(x,c;\eta^{\text{o}})\big|}=o_p\bigg(\frac{1}{\sqrt{n}}\bigg),\label{A.2.11}\\
&I\!I_2
=\sqrt{n}\sum^{1}_{m=0}(-1)^{1-m}\bigg(\frac{1}{n^{2}}\sum_{i=1}^n\sum_{j=1}^{n}\frac{f^{[1-m]}(X_i,C_i;\eta^{\text{o}})}{f^{2}(X_i,C_i;\eta^{\text{o}})}\xi^{[m]}_{j}(X_i,C_i;\eta^{\text{o}})\bigg) +O_p\bigg(\sqrt{n}h^4+\sqrt{\frac{1}{nh^{4}}}\,\bigg) \nonumber \\
&= \sqrt{n}\sum^{1}_{m=0}(-1)^{1-m}\bigg(\frac{1}{n^{2}}\sum_{i=1}^n\sum_{j=1}^{n}U_{1,ij}^{[m]}+\frac{1}{n}\sum_{i=1}^n V_{1,i}^{[m]}\bigg)+o_p(1)\text{ and} \label{A.2.12}\\
&I\!I_{3}
=\sqrt{n}\sum^{1}_{m=0}(-1)^{1-m}\bigg(\frac{1}{n^{2}}\sum_{i=1}^n\sum_{j=1}^{n}\frac{f^{[1-m]}(X_i;\eta^{\text{o}})}{f^{2}(X_i;\eta^{\text{o}})}\xi^{[m]}_{j}(X_i;\eta^{\text{o}})\bigg) +O_p\bigg(\sqrt{n}h^4+\sqrt{\frac{1}{nh^{4}}}\,\bigg) \nonumber \\
&= \sqrt{n}\sum^{1}_{m=0}(-1)^{1-m}\bigg(\frac{1}{n^{2}}\sum_{i=1}^n\sum_{j=1}^{n}U_{2,ij}^{[m]}+\frac{1}{n}\sum_{i=1}^n V_{2,i}^{[m]}\bigg)+o_p(1), \label{A.3.9}
\end{align}
where
\begin{align*}
   & U_{1,ij}^{[m]}=\frac{f^{[1-m]}(X_i,C_i;\eta^{\text{o}})}{f^{2}(X_i,C_i;\eta^{\text{o}})}\xi_{j,m}(X_i,C_i;\eta^{\text{o}})-V_{1,i}^{[m]},~
    V_{1,i}^{[m]}=\E\bigg[\frac{f^{[1-m]}(X_i,C_i;\eta^{\text{o}})}{f^{2}(X_i,C_i;\eta^{\text{o}})}\xi^{[m]}_{j}(X_i,C_i;\eta^{\text{o}})\bigg|i\bigg],\\
     &\xi^{[m]}_{j}(x;\eta^{\text{o}})=\sum_{c=1}^k \xi_{j,m}(x,c;\eta^{\text{o}}),~ U_{2,ij}^{[m]}=\frac{f^{[1-m]}(X_i;\eta^{\text{o}})}{f^{2}(X_i;\eta^{\text{o}})}\xi^{[m]}_{j}(X_i;\eta^{\text{o}})-V_{2,i}^{[m]}\text{ and} \\
& V_{2,i}^{[m]}=\E\bigg[\frac{f^{[1-m]}(X_i;\eta^{\text{o}})}{f^{2}(X_i;\eta^{\text{o}})}\xi^{[m]}_{j}(X_i;\eta^{\text{o}})\bigg|i\bigg],~i,j=1,\dots,n,~m=0,1.
\end{align*}
By the definition of $f^{[1]}(x,c;\eta)$ in (\ref{A.0.3}) and the identity $\E\big[(X-\mu_{c})^{\otimes 2}\mid Y_{c},C=c \big]=\big[(1+Y_{c})^{p/2}-1\big]^{2/p}\Sigma_{c}/p$, it follows that
\begin{align}
    &\E\big[f^{[1]}(X,C;\eta^{\text{o}})\mid Y_{c},C=c \big]=0.  \label{A.3.0}
\end{align}
Moreover, the conditional mean-zero properties $\E\Big[U^{[0]}_{1,ij}\Big| j\Big]=0$ and $\E\Big[U^{[1]}_{1,ij}\Big| j\Big]=0$ are implied by the identity $\E[\xi^{[1]}_{j}(X_i,C_i;\eta^{\text{o}})/f(X_i,C_i;\eta^{\text{o}})| X_{j},C_{j}]=0$, together with (\ref{A.3.0}).
Consequently, $U^{[m]}_{s,ij}$, $s=1,2$, are degenerate $U$-statistics with variances of order $O(h^{-2m-1})$ for $m=0,1$. By Theorem 8.1 in \cite{hoeffding1948probability}, we obtain
\begin{align}
   \frac{1}{n^2}\sum_{i=1}^n \sum_{j=1}^n U_{s,ij}^{[m]}=O_p\bigg(\frac{1}{n\sqrt{h^{2m+1}}}\bigg),
~s=1,2,~,m=0,1.\label{A.2.13}
  \end{align}
Since $\E \big[V_{s,i}^{[m]} \big] = 0$ and $\var \big(V_{s,i}^{[m]} \big) = O(h^4)$ for $s = 1, 2$ and $m = 0, 1$, it is implied by the central limit theorem that
\begin{align}
    \frac{1}{n}\sum_{i=1}^n V_{s,i}^{[m]}= O_p\bigg(\frac{h^{2}}{\sqrt{n}}\bigg),
   ~s=1,2,~ m=0,1.\label{A.2.14}
     \end{align}
Combining (\ref{A.2.13}) and (\ref{A.2.14}) yields
\begin{align}
  I\!I_2=o_p(1) \text{ and }I\!I_{3}=o_p(1). \label{A.2.t.15}
   \end{align}
Substituting the first result in (\ref{A.2.t.15}) and (\ref{A.2.11}) into (\ref{A.2.9}), and the second result in (\ref{A.2.t.15}) into (\ref{A.3.8}), and applying the central limit theorem, we conclude that
\begin{align}
& \frac{1}{\sqrt{n}}pS_{s}(\eta^{\text{o}})=\frac{1}{\sqrt{n}}S_{s}(\eta^{\text{o}})+o_p(1) 
  \stackrel{d}{\longrightarrow} N\big(0,V_{s}\big),~s=1,2.\label{A.2.15}
  \end{align}

The normalized $pI_1(\eta^{\text{o}})$ admits the following decomposition:
\begin{align}
\frac{1}{n}pI_1(\eta^{\text{o}})=&~\frac{1}{n}\sum_{i=1}^n \sum_{c=1}^k I\big(i \in \mathcal{G}_{c}^{\text{o}}\big)\bigg[\bigg(\frac{\partial_{\eta}\hat{f}_{h}(X_i,c;\eta^{\text{o}})} {\hat{f}_{h}(X_i,c;\eta^{\text{o}})} \bigg) ^{\otimes 2}-\frac{\partial_{\eta}^{2} \hat{f}_{h}(X_i,c;\eta^{\text{o}})} {\hat{f}_{h}(X_i,c;\eta^{\text{o}})}\bigg] \nonumber \\
&~+\frac{1}{n}\sum_{i=1}^n \sum_{c=1}^k \big(I\big(i \in \widetilde{\mathcal{G}}_{c}^{\textsc{sp}}\big)- I\big(i \in \mathcal{G}_{c}^{\text{o}}\big)\big)\bigg[\bigg(\frac{\partial_{\eta}\hat{f}_{h}(X_i,c;\eta^{\text{o}})} {\hat{f}_{h}(X_i,c;\eta^{\text{o}})} \bigg) ^{\otimes 2}-\frac{\partial_{\eta}^{2} \hat{f}_{h}(X_i,c;\eta^{\text{o}})} {\hat{f}_{h}(X_i,c;\eta^{\text{o}})}\bigg]. \label{A.2.16}
\end{align}\vspace{-0.01in}
By (\ref{A.0.5}) and Lemma \ref{lem1}, it follows that, under assumptions {A4} and {A6},
\begin{align}
  &\frac{1}{n}\sum_{i=1}^n \sum_{c=1}^k I\big(i \in \mathcal{G}_{c}^{\text{o}}\big)\bigg[\bigg(\frac{\partial_{\eta}\hat{f}_{h}(X_i,c;\eta^{\text{o}})} {\hat{f}_{h}(X_i,c;\eta^{\text{o}})} \bigg) ^{\otimes 2}-\frac{\partial_{\eta}^{2} \hat{f}_{h}(X_i,c;\eta^{\text{o}})} {\hat{f}_{h}(X_i,c;\eta^{\text{o}})}\bigg] \nonumber \\
    &=\frac{1}{n}\sum_{i=1}^n \sum_{c=1}^k I\big(i \in \mathcal{G}_{c}^{\text{o}}\big)\bigg[\bigg(\frac{f^{[1]}(X_i,c;\eta^{\text{o}})} {f(X_i,c;\eta^{\text{o}})} \bigg) ^{\otimes 2}-\frac{f^{[2]}(X_i,c;\eta^{\text{o}})} {f(X_i,c;\eta^{\text{o}})}\bigg]+o_p(1), \label{A.2.17}\\
   &\bigg| \frac{1}{n}\sum_{i=1}^n \sum_{c=1}^k \big(I\big(i \in \widetilde{\mathcal{G}}_{c}^{\textsc{sp}}\big)- I\big(i \in \mathcal{G}_{c}^{\text{o}}\big)\big)\bigg[\bigg(\frac{\partial_{\eta}\hat{f}_{h}(X_i,c;\eta^{\text{o}})} {\hat{f}_{h}(X_i,c;\eta^{\text{o}})} \bigg) ^{\otimes 2}-\frac{\partial_{\eta}^{2} \hat{f}_{h}(X_i,c;\eta^{\text{o}})} {\hat{f}_{h}(X_i,c;\eta^{\text{o}})}\bigg] \bigg| \nonumber \\
     &\leq \sup_{\{i,c\}} \big|\big( I\big(i\in \widetilde{\mathcal{G}}_{c}^{\textsc{sp}}\big)-I\big(i\in \mathcal{G}_c^{\text{o}}\big) \big| \Bigg| \frac{1}{n}\sum_{i=1}^n \sum_{c=1}^k \Bigg[\Bigg(\frac{\partial_{\eta}\hat{f}_{h}(X_i,c;\eta^{\text{o}})} {\hat{f}_{h}(X_i,c;\eta^{\text{o}})} \Bigg) ^{\otimes 2}-\frac{\partial_{\eta}^{2} \hat{f}_{h}(X_i,c;\eta^{\text{o}})} {\hat{f}_{h}(X_i,c;\eta^{\text{o}})}\Bigg] \Bigg| =o_p\bigg(\frac{1}{n}\bigg) \label{A.2.18}
\end{align}
and
\begin{eqnarray}
 \hspace{-0.1in} \frac{1}{n}\sum_{i=1}^n \bigg[\bigg(\frac{\partial_{\eta}\hat{f}_{h}(X_i;\eta^{\text{o}})} {\hat{f}_{h}(X_i;\eta^{\text{o}})} \bigg) ^{\otimes 2}-\frac{\partial_{\eta}^{2} \hat{f}_{h}(X_i;\eta^{\text{o}})} {\hat{f}_{h}(X_i;\eta^{\text{o}})}\bigg] =\frac{1}{n}\sum_{i=1}^n \bigg[\bigg(\frac{f^{[1]}(X_i;\eta^{\text{o}})} {f(X_i;\eta^{\text{o}})} \bigg) ^{\otimes 2}-\frac{f^{[2]}(X_i;\eta^{\text{o}})} {f(X_i;\eta^{\text{o}})}\bigg]+o_p(1).~ \label{A.3.15}
\end{eqnarray}
Substituting (\ref{A.2.17}) and (\ref{A.2.18}) into (\ref{A.2.16}), (\ref{A.3.15}), and invoking the law of large numbers yield
\begin{align*}
    &\frac{1}{n}pI_1(\eta^{\text{o}})\stackrel{p}{\longrightarrow}\E\bigg[I\big(i \in \mathcal{G}_{c}^{\text{o}}\big)\bigg(\bigg(\frac{f^{[1]}(X_i,c;\eta^{\text{o}})} {f(X_i,c;\eta^{\text{o}})} \bigg) ^{\otimes 2}-\frac{f^{[2]}(X_i,c;\eta^{\text{o}})} {f(X_i,c;\eta^{\text{o}})}\bigg)\bigg] \text{ and}\\
         &  \frac{1}{n}pI_2(\eta^{\text{o}})\stackrel{p}{\longrightarrow}\E\bigg[\bigg(\bigg(\frac{f^{[1]}(X_i;\eta^{\text{o}})} {f(X_i;\eta^{\text{o}})} \bigg) ^{\otimes 2}-\frac{f^{[2]}(X_i;\eta^{\text{o}})} {f(X_i;\eta^{\text{o}})}\bigg)\bigg].
\end{align*}
By Lemma \ref{lem1} with $m=2$, the uniform integrability of $I(i \in \mathcal{G}_{c}^{\text{o}})\partial_{\eta}^{2} \hat{f}_h(X,c;\eta^{\text{o}})/f(X,c;\eta^{\text{o}})$ and $\partial_{\eta}^{2} \hat{f}_h(X;\eta^{\text{o}})/f(X;\eta^{\text{o}})$, 
\begin{align*}
& \E\bigg[I\big(i \in \mathcal{G}_{c}^{\text{o}}\big)\frac{\partial_{\eta}^{2} \hat{f}_h(X,c;\eta^{\text{o}})}{f(X,c;\eta^{\text{o}})}\bigg]=\int \frac{\partial_{\eta}^{2} \hat{f}_h(x,c;\eta^{\text{o}})}{f(x,c;\eta^{\text{o}})} f(x,c;\eta^{\text{o}}) dx
 =\partial_{\eta}^{2} \int \hat{f}_h(x,c;\eta^{\text{o}}) dx=0\text{ and}\\
 &\E\bigg[\frac{\partial_{\eta}^{2} \hat{f}_h(X;\eta^{\text{o}})}{f(X;\eta^{\text{o}})}\bigg]=\int \frac{\partial_{\eta}^{2} \hat{f}_h(x;\eta^{\text{o}})}{f(x;\eta^{\text{o}})} f(x;\eta^{\text{o}}) dx
 =\partial_{\eta}^{2} \int \hat{f}_h(x;\eta^{\text{o}}) dx=0,
\end{align*}
it follows that
\begin{align*}
     \E\bigg[\frac{f^{[2]}(X,C;\eta^{\text{o}})} {f(X,C;\eta^{\text{o}})}\bigg]=0\text{ and }  \E\bigg[\frac{f^{[2]}(X;\eta^{\text{o}})} {f(X;\eta^{\text{o}})}\bigg]=0.
     \end{align*}
Thus, by the continuity of $pI_1(\eta)$ and $pI_2(\eta)$ in $\eta$ and the consistency properties of $\tilde{\eta}^{\textsc{pml}}$ and $\tilde{\eta}^{\textsc{pmml}}$, we conclude that 
\begin{align}
   & \frac{1}{n}pI_1(\bar{\eta}^{\textsc{pml}}) \stackrel{p}{\longrightarrow} V_1 \label{A.2.20}
   \end{align}
   and
   \begin{align}
    & \frac{1}{n}pI_2(\bar{\eta}^{\textsc{pmml}}) \stackrel{p}{\longrightarrow} V_2. \label{A.3.17}
     \end{align}
Substituting (\ref{A.2.15}) with $s=1$ and (\ref{A.2.20}) into (\ref{A.2.8}), and (\ref{A.2.15}) with $s=2$ and (\ref{A.3.17}) into (\ref{A.3.7}), and invoking Slutsky's theorem under assumption {A8} together establish the asymptotic normality of $\tilde{\eta}^{\textsc{pml}}$ and $\tilde{\eta}^{\textsc{pmml}}$.

\subsection{Derivation of the Semiparametric Efficiency Bound}
\label{speff.pf}

To derive the nuisance tangent space $\Lambda$ associated with $\{f(x, c; \eta)\}$, we consider the collection of nuisance score vectors corresponding to all possible parametric submodels $\{f_{\gamma}(x, c; \eta)\}$, where $\gamma$ denotes a $q \times 1$ parameter vector. A generic submodel takes the form 
\begin{align*}
     f_{\gamma}(x,c;\eta^{\text{o}})=\omega_c w(y_{c})g_{\gamma}(y_{c})\text{ with } g_{\gamma}(y)=\lim_{h \rightarrow 0^{+}} \E\big[K_{h}\big(Y_{\gamma},y\big)\big],
\end{align*}
where $Y_{\gamma}$ is defined analogously to $Y$ in Section \ref{sec:sec2}, with $\theta$ replaced by $\gamma$. By construction, we have $f_{\eta^{\text{o}}}(x,c;\eta^{\text{o}})=f(x,c;\eta^{\text{o}})$ for each $c \in \mathcal{C}$.
Following the formulation in \citet{tsiatis2006semiparametric}, the nuisance tangent space $\Lambda$ is characterized as the mean-square closure of the tangent spaces generated by all such submodels. That is, for each parametric submodel, the associated tangent space is given by
\begin{align*}
    \Lambda=\Bigg\{D\sum_{c=1}^k I(C=c)  \frac{\partial_{\gamma}f_{\gamma}(X,c;\eta^{\text{o}})\mid_{\gamma=\gamma^{\text{o}}}}{f_{\gamma^{\text{o}}}(X,c;\eta^{\text{o}})} \Bigg\},
\end{align*}
where $D$ is an arbitrary constant matrix of dimension $d \times q$.

Denote by $\Lambda^{\perp}$ the orthogonal complement of $\Lambda$.
Let $\ell(\eta)$ denote the log-likelihood contribution
\begin{align*}
    \ell(\eta)=\sum^{k}_{c=1}I(C=c)\log f^{[0]}(X,c;\eta),
\end{align*}    
and define the score-type function as

\begin{align*}
    pS(\eta)=\sum^{k}_{c=1}I(C=c)\frac{f^{[1]}(X,c;\eta)}{f^{[0]}(X,c;\eta)}.
\end{align*}
According to the semiparametric efficiency theory of \citet{BKRW1998Efficient}, the efficient score for estimating $\eta$ is given by the orthogonal projection of $\partial_{\eta} \ell(\eta^{\text{o}})$ onto $\Lambda^{\perp}$.
Therefore, the function $pS(\eta^{\text{o}})$ coincides with the efficient score if it equals this projection.

For any $v \in \Lambda$, it depends on the entire collection $\{Y_1, \dots, Y_k\}$, whereas the function $f(X, c; \eta^{\text{o}})$ depends solely on $Y_c$. This structural distinction, together with (\ref{A.3.0}), implies that
\begin{align*}
    \E \big[ v^{\top} pS(\eta^{\text{o}}) \big]=\E\bigg[\sum_{c=1}^k I(C=c)\frac{v^{\top}}{f(X,c;\eta^{\text{o}})}\E\big[f^{[1]}(X,c;\eta^{\text{o}})|Y_c,C=c\big]\bigg]=0.
\end{align*}
It follows that
\begin{align}
    pS(\eta^{\text{o}})\in \Lambda^{\perp}.\label{A.5.1}
\end{align}
Next, observe that
\begin{align*}
    pS(\eta^{\text{o}})-\partial_{\eta}\ell(\eta^{\text{o}})=\sum_{c=1}^k I(C=c)\frac{\partial_{\gamma}f_{\gamma}(X,c;\eta^{\text{o}})|_{\gamma=\eta^{\text{o}}}}{f(X,c;\eta^{\text{o}})},
\end{align*}
from which it follows that
\begin{align}
    pS(\eta^{\text{o}})-\partial_{\eta}\ell(\eta^{\text{o}})\in \Lambda.\label{A.5.2}
\end{align}
Combining (\ref{A.5.1}) and (\ref{A.5.2}), we conclude that $pS(\eta^{\text{o}})$ is the orthogonal projection of $\partial_{\eta} \ell(\eta^{\text{o}})$ onto $\Lambda^{\perp}$. Consequently, $pS(\eta^{\text{o}})$ coincides with the efficient score in the semiparametric model. The efficiency of the estimator $\tilde{\eta}^{\textsc{pml}}$ thus follows as a direct consequence.

\subsection{Proof of Theorem \ref{Thm3.4}}
 \label{Thm3.4.pf}

Following an argument analogous to that used in the proof of Lemma \ref{lem1}, we obtain
\begin{eqnarray}
 \sup_{\{(x,\eta) \in \mathcal{X}_{n} \times \mathcal{H}\}}\Big\Vert\hat{f}_{\tilde{h}}(x;\eta)-f^{[0]}(x;\eta)\Big\Vert
 =\sup_{\{(x,\eta) \in \mathcal{X}_{n} \times \mathcal{H}\}}\bigg\Vert \frac{1}{n}\sum_{i=1}^n \xi^{[0]}_{i}(x;\eta)  \bigg\Vert=o_p\bigg(\frac{\sqrt{\log n}}{n^{\frac{2}{5}}}\bigg), \label{A.4.m.1}
\end{eqnarray}
where $\xi^{[0]}_{i}(x;\eta)=\sum^{k}_{c=1}\xi^{[0]}_{i}(x,c;\eta)$, $i=1,\dots,n$.
From the proof of Theorem  \ref{Thm3.2}, it follows that $ \tilde{\eta}^{\textsc{pml}}-\eta^{\text{o}}= O_{p}\big( n^{-1/2} \big) $ whenever $\mathcal{G} \in \mathcal{C}_{\mathcal{G}}$. In contrast, when $\mathcal{G} \notin \mathcal{C}_{\mathcal{G}}$, $\tilde{\eta}^{\textsc{pml}}$ remains consistent to $\eta^{\text{o}}$, but does not attain the $\sqrt{n}$-consistency due to the presence of a nonzero expectation in the score function. Accordingly, the convergence rate of $\tilde{\eta}^{\textsc{pml}}$ is summarized as
\begin{align}
    \tilde{\eta}^{\textsc{pml}}-\eta^{\text{o}}= \left\{
\begin{array}{ll}
O_{p}\Big( n^{-\frac{1}{2}} \Big)  
& \mathcal{G} \in \mathcal{C}_{\mathcal{G}}, \\
 o_p(1)
& \text{otherwise.}
\end{array}\right. \label{A.4.2}
\end{align}
Using the properties in (\ref{A.4.m.1}) and (\ref{A.4.2}) with a Taylor expansion and assumption {A6}, we obtain
\begin{align}
&\sup_{\{x\in\mathcal{X}_n\}}|\hat{f}_{\tilde{h}}(x;\tilde{\eta}^{\textsc{pml}})-f(x;\eta^{\text{o}})| \nonumber \\
   & \leq \sup_{\{x\in\mathcal{X}_n\}} |\hat{f}_{\tilde{h}}(x;\tilde{\eta}^{\textsc{pml}})-f^{[0]}(x;\tilde{\eta}^{\textsc{pml}})|+\sup_{\{(x,\eta) \in \mathcal{X}_{n} \times \mathcal{H}\}} \Vert \partial_{\eta} f^{[0]}(x;\eta) \Vert \Vert \tilde{\eta}^{\textsc{pml}}-\eta^{\text{o}} \Vert \nonumber \\
    &= \left\{
\begin{array}{ll}
o_p\Big(\frac{\sqrt{\log n}}{n^{\frac{2}{5}}}\Big) + O_{p}\Big( n^{-\frac{1}{2}} \Big)
& \mathcal{G}^{\text{o}} \in \mathcal{C}_{\mathcal{G}}, \\
o_p\Big(\frac{\sqrt{\log n}}{n^{\frac{2}{5}}}\Big) + O_{p}\big( \| \tilde{\eta}^{\textsc{pml}} - \eta^{\text{o}} \| \big)
& \text{otherwise.}
\end{array}\right.\label{A.4.m.3}
\end{align}

By applying a Taylor expansion and the uniform convergence property in (\ref{A.2.3}), the log-pseudo-likelihood function admits the following decomposition:
\begin{align}
    -\frac{1}{n}p\ell(k) =&-\frac{1}{n}\sum_{i=1}^n \log f(X_i;\eta^{\text{o}})-\frac{1}{n}\sum_{i=1}^n  \frac{\hat{f}^{-i}_{\tilde{h}}(X_i;\tilde{\eta}^{\textsc{pml}})-f(X_i;\eta^{\text{o}})}{f(X_i;\eta^{\text{o}})}
    +O_{p}\big( \| \tilde{\eta}^{\textsc{pml}} - \eta^{\text{o}} \|^{2}\big)+o_p\bigg(\frac{\log n}{n^{\frac{4}{5}}}\bigg)\nonumber \\
   \stackrel{\triangle}{=}& \left\{
\begin{array}{ll}
I\!I\!I_1(k)+I\!I\!I_2(k)+o_p\Big(\frac{\log n}{n^{\frac{4}{5}}}\Big)
& \mathcal{G}^{\text{o}} \in \mathcal{C}_{\mathcal{G}}, \\
I\!I\!I_1(k)+I\!I\!I_2(k)+o_p\Big(\frac{\log n}{n^{\frac{4}{5}}}\Big) + O_{p}\big( \| \tilde{\eta}^{\textsc{pml}} - \eta^{\text{o}} \|^{2}\big)
& \text{otherwise.}
\end{array}\right.\label{A.4.m.4}
\end{align}
When $\mathcal{G}^{\text{o}} \in \mathcal{C}_{\mathcal{G}}$, it follows by construction that $I\!I\!I_1(k)=I\!I\!I_1(k^*)$. Otherwise, applying the law of large numbers yields $I\!I\!I_1(k)-I\!I\!I_1(k^*)=b^2(k)(1+o_p(1))$. Thus, we summarize
\begin{align}
  &I\!I\!I_1(k)-I\!I\!I_1(k^*)= \left\{
\begin{array}{ll}
0
& \mathcal{G}^{\text{o}} \in \mathcal{C}_{\mathcal{G}}, \\
b^2(k)(1+o_p(1))
& \text{otherwise.}
\end{array}\right. \label{A.4.m.5}
\end{align}
Substituting the $i.i.d.$ representation from (\ref{A.4.m.1}) into $I\!I\!I_2(k)$, we obtain
\begin{eqnarray}
 I\!I\!I_2(k)=-\frac{1}{n(n-1)}\sum_{i=1}^n\sum_{j \neq i}\frac{\xi^{[0]}_{j}(X_i;\tilde{\eta}^{\textsc{pml}})}{f^{[0]}(X_i;\eta^{\text{o}})}-\E\bigg[\frac{\xi^{[0]}_{j}(X_i;\tilde{\eta}^{\textsc{pml}})}{f^{[0]}(X_i;\eta^{\text{o}})}\bigg| i \bigg] -\frac{1}{n}\sum_{i=1}^n \E\bigg[\frac{\xi^{[0]}_{j}(X_i;\tilde{\eta}^{\textsc{pml}})}{f^{[0]}(X_i;\eta^{\text{o}})}\bigg| i \bigg].~\label{A.4.6}
\end{eqnarray}
Since $\xi^{[0]}(x;\eta)/f^{[0]}(x;\eta^{\text{o}})$ is of bounded variation in $x$, uniformly over $c\in\mathcal{C}$ and $\eta\in\mathcal{H}$, Lemma 22 in \cite{nolan1987u} implies that the class
\begin{align*}
\bigg\{\frac{\xi^{[0]}_{j}(X_i;\eta)}{f^{[0]}(X_i;\eta^{\text{o}})}:\eta \in \mathcal{H}\bigg\} 
\end{align*}
is Euclidean.
When $\mathcal{G}^{\text{o}} \in \mathcal{C}_{\mathcal{G}}$, the second-order $U$-process in (\ref{A.4.6}) is degenerate,
and the variances of the summands in the two terms of (\ref{A.4.6}) are of orders $O(n^{1/5})$ and $O(n^{-4/5})$, respectively. By Corollary 4 in \cite{sherman1994maximal}, 
\begin{align*}
  -\frac{1}{n(n-1)}\sum_{i=1}^n\sum_{j \neq i} \bigg(\frac{\xi^{[0]}_{j}(X_i;\tilde{\eta}^{\textsc{pml}})}{f^{[0]}(X_i;\eta^{\text{o}})}-\E\bigg[\frac{\xi^{[0]}_{j}(X_i;\tilde{\eta}^{\textsc{pml}})}{f^{[0]}(X_i;\eta^{\text{o}})}\bigg| i \bigg]\bigg) =&~O_p \Big(n^{-\frac{9}{10}}\Big) \text{ and }\\
  \frac{1}{n}\sum_{i=1}^n \E\bigg[\frac{\xi^{[0]}_{j}(X_i;\tilde{\eta}^{\textsc{pml}})}{f^{[0]}(X_i;\eta^{\text{o}})}\bigg| i \bigg]=&~O_p \Big(n^{-\frac{9}{10}}\Big).
  \end{align*}
Thus,
\begin{align}
   I\!I\!I_2(k)=O_p \Big(n^{-\frac{9}{10}}\Big).\label{A.4.7}
   \end{align}
When $\mathcal{G}^{\text{o}} \notin \mathcal{C}_{\mathcal{G}}$, the corresponding $U$-process is non-degenerate. Together with the bound in (\ref{A.4.m.3}), this yields
\begin{align}
    I\!I\!I_2(k)=o_p\bigg(\frac{\sqrt{\log n}}{n^{\frac{2}{5}}}\bigg)\text{ if } \mathcal{G}^{\text{o}} \notin \mathcal{C}_{\mathcal{G}}. \label{A.4.8}
\end{align}
Combining the results in (\ref{A.4.m.5}), (\ref{A.4.7}), and (\ref{A.4.8}) with the expression of $\text{SPIC}(k)$, we obtain 
\begin{align}
&\text{SPIC}(k)-\text{SPIC}(k^*) = \left\{
\begin{array}{ll}
\max\Big\{ b^{2}(k),(k-k^*)\frac{\log n }{ 2 n^{\frac{4}{5}}}\Big\}(1+o_{p}(1))
& k > k^*, \\
 b^{2}(k) (1+o_{p}(1))
& k < k^*.
\end{array}\right. \label{A.4.m.9}
\end{align}
The assertion in Theorem \ref{Thm3.4} follows directly from (\ref{A.4.m.9}).

\subsection{Supplementary Figures and Tables}
\label{table-figure}

Figure \ref{fig:Sindian_mean_D} and Tables \ref{tab:addlabel}--\ref{tab:label3} provide additional descriptions of the empirical applications. Table \ref{tab:label4} reports the 90/10 cluster-imbalance results. Table \ref{tab:lam_ari} and Figures \ref{fig:lam_res}--\ref{fig:lam_res_1} report the exploratory SP-stage tuning comparison, based on ten replications per configuration.

\begin{figure}[H]
\centering
 \includegraphics[width=15cm,height=15cm]{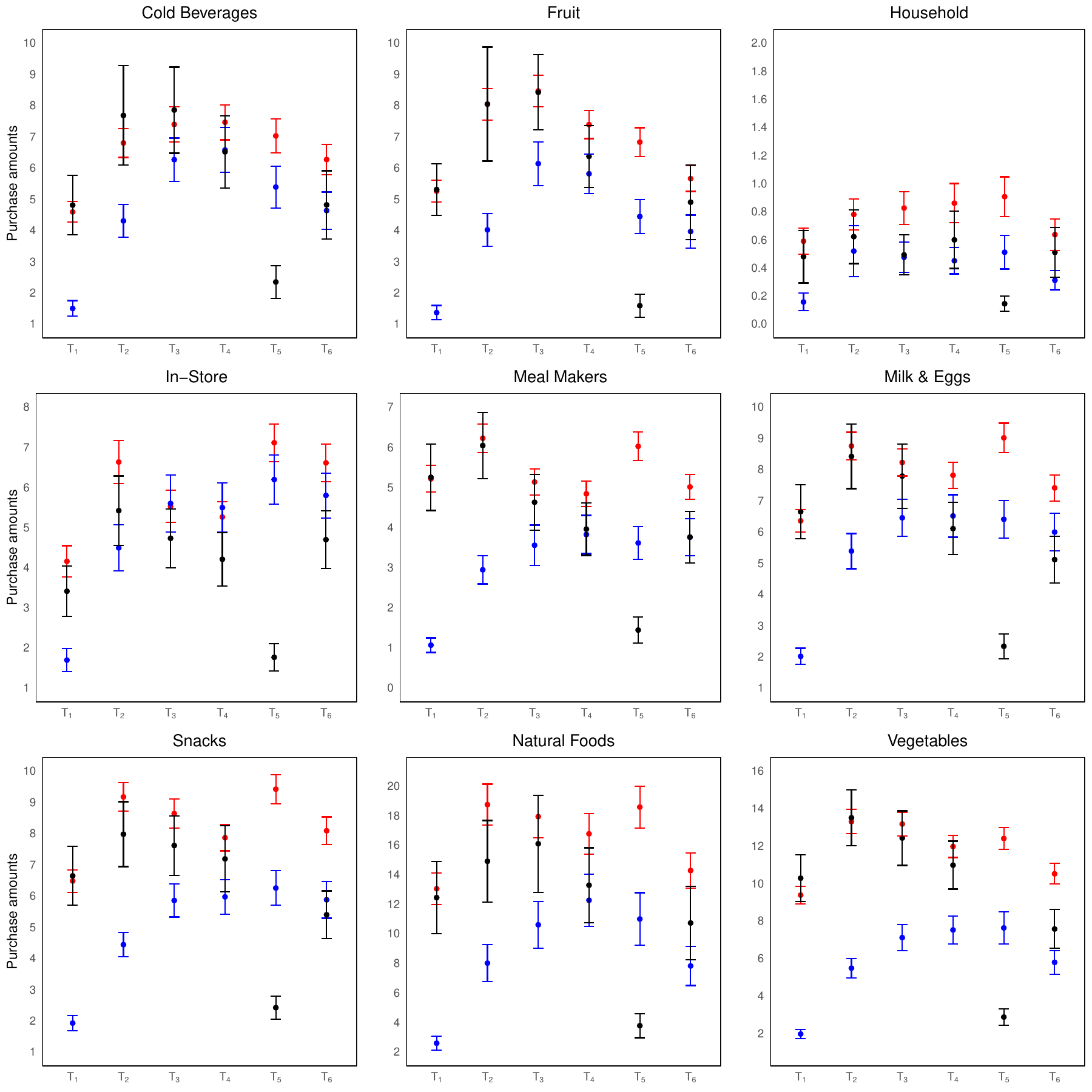}
 \caption{{\footnotesize Average purchase amounts and approximate 95\% confidence intervals for the nine major product categories, stratified by cluster. Results for Clusters 1, 2, and 3 are shown in red, blue, and black, respectively. 
 }} \label{fig:Sindian_mean_D} 
\end{figure}

\begin{table}[htbp]
  \centering
  \caption{Average weekly purchase amounts (AWPA) and average purchase amounts (APA), with standard errors in parentheses, for the nine major product categories, stratified by cluster.}
  \begin{adjustbox}{max width=\linewidth}
      \begin{tabular}{ccccccc}
            \toprule
& \multicolumn{3}{c}{AWPA} & \multicolumn{3}{c}{APA}  \\
\cmidrule(rl){2-4} \cmidrule(rl){5-7}
          & \multicolumn{3}{c}{Cluster} & \multicolumn{3}{c}{Cluster}  \\
\cmidrule(rl){2-4} \cmidrule(rl){5-7}
    Product category & 1     & 2     & 3 & 1     & 2     & 3 \\
\midrule
    Cold Beverages & 0.67  & 0.44  & 0.56  & 1.21  & 0.80   & 1.03 \\
          & (0.020) & (0.021) & (0.048) & (0.036) & (0.039) & (0.091) \\
    Fruit & 0.71  & 0.38  & 0.57  & 1.27  & 0.70   & 1.06 \\
          & (0.017) & (0.017) & (0.042) & (0.030) & (0.031) & (0.080) \\
    Household & 0.08  & 0.04  & 0.05  & 0.15  & 0.07  & 0.09 \\
          & (0.004) & (0.003) & (0.004) & (0.008) & (0.005) & (0.008) \\
    In-Store & 0.58  & 0.46  & 0.39  & 1.04  & 0.84  & 0.73 \\
          & (0.017) & (0.021) & (0.022) & (0.030) & (0.038) & (0.042) \\
    Meal Makers & 0.56  & 0.27  & 0.41  & 1.01  & 0.51  & 0.76 \\
          & (0.012) & (0.012) & (0.022) & (0.022) & (0.023) & (0.041) \\
    Milk \& Eggs & 0.84  & 0.51  & 0.62  & 1.51  & 0.94  & 1.15 \\
          & (0.018) & (0.021) & (0.032) & (0.031) & (0.038) & (0.060) \\
    Snacks & 0.90   & 0.52  & 0.67  & 1.62  & 0.96  & 1.25 \\
          & (0.017) & (0.018) & (0.031) & (0.031) & (0.033) & (0.057) \\
    Natural Foods & 1.47  & 0.61  & 1.04  & 2.67  & 1.12  & 1.94 \\
          & (0.050) & (0.041) & (0.091) & (0.091) & (0.075) & (0.169) \\
    Vegetables & 1.32  & 0.60   & 1.03  & 2.38  & 1.11  & 1.92 \\
          & (0.024) & (0.023) & (0.047) & (0.043) & (0.042) & (0.086) \\
\bottomrule
    \end{tabular}%
  \label{tab:addlabel}%
  \end{adjustbox}
  \end{table}%

% Table generated by Excel2LaTeX from sheet 'Sheet 1'
\begin{table}[htbp]
  \centering
  \caption{Sample means (standard errors) of six biomarkers and $\textit{age}$, stratified by diabetes status, gravidity status, and $(\mathit{ds}, \mathit{gs})$-based grouping.}
   \begin{adjustbox}{max width=\linewidth}
    \begin{tabular}{cccccccccc}
\toprule
          &    \multicolumn{2}{c}{\textit{ds}}        & \multicolumn{2}{c}{\textit{gs}}       & \multicolumn{4}{c}{(\textit{ds},\textit{gs})}    \\
\cmidrule(rl){2-3} \cmidrule(rl){4-5} \cmidrule(rl){6-9}
     Variable     &  0     & 1     & 0     & 1     & (0,0) & (1,0) & (0,1) & (1,1) \\
\midrule
    $\textit{gtt}$   & -0.36 & 0.72  & -0.20 & 0.24  & -0.47 & 0.63  & -0.19 & 0.78 \\
          & (0.054) & (0.075) & (0.066) & (0.074) & (0.066) & (0.123) & (0.09) & (0.095) \\
    $\textit{dbp}$   & -0.14 & 0.27  & -0.16 & 0.18  & -0.25 & 0.14  & 0.04  & 0.36 \\
          & (0.06) & (0.088) & (0.071) & (0.07) & (0.077) & (0.166) & (0.096) & (0.098) \\
    $\textit{tsft}$  & -0.18 & 0.37  & -0.04 & 0.05  & -0.21 & 0.50  & -0.13 & 0.28 \\
          & (0.064) & (0.072) & (0.072) & (0.071) & (0.083) & (0.112) & (0.101) & (0.094) \\
    $\textit{si}$    & -0.24 & 0.49  & -0.12 & 0.14  & -0.29 & 0.41  & -0.18 & 0.54 \\
          & (0.06) & (0.077) & (0.071) & (0.071) & (0.077) & (0.145) & (0.096) & (0.087) \\
    $\textit{bmi}$   & -0.19 & 0.39  & 0.01  & -0.01 & -0.20 & 0.69  & -0.19 & 0.20 \\
          & (0.064) & (0.073) & (0.077) & (0.063) & (0.086) & (0.133) & (0.091) & (0.077) \\
    $\textit{dpf}$   & -0.14 & 0.29  & 0.02  & -0.02 & -0.08 & 0.32  & -0.25 & 0.27 \\
          & (0.06) & (0.088) & (0.069) & (0.075) & (0.078) & (0.145) & (0.096) & (0.11) \\
    $\textit{age}$   & -0.26 & 0.53  & -0.49 & 0.58  & -0.62 & -0.06 & 0.32  & 0.90 \\
          & (0.053) & (0.094) & (0.044) & (0.078) & (0.039) & (0.118) & (0.098) & (0.117) \\
\bottomrule
    \end{tabular}%
      \end{adjustbox}
        \label{tab:label2}%
\end{table}%

\begin{table}[htbp]
  \centering
   \caption{Sample means (standard errors) of six biomarkers and $\textit{age}$, stratified by diabetes status within each cluster (top panel) and by gravidity status within each cluster (bottom panel).}
     \begin{adjustbox}{max width=\linewidth}
     \begin{tabular}{ccccccccc}
\toprule
        Cluster  & \multicolumn{2}{c}{1} & \multicolumn{2}{c}{2} & \multicolumn{2}{c}{3} & \multicolumn{2}{c}{4}  \\
\cmidrule(rl){1-1} \cmidrule(rl){2-3} \cmidrule(rl){4-5} \cmidrule(rl){6-7} \cmidrule(rl){8-9}  
        Variable  & $\textit{ds}=0$ & $\textit{ds}=1$ & $\textit{ds}=0$ & $\textit{ds}=1$ &$\textit{ds}=0$ & $\textit{ds}=1$ &$\textit{ds}=0$ & $\textit{ds}=1$  \\
\midrule
    $\textit{gtt}$   & -0.90 & -0.61 & 0.49  & 0.95  & -0.39 & 0.57  & -0.03 & 0.84 \\
          & (0.066) & (0.164) & (0.082) & (0.09) & (0.084) & (0.276) & (0.154) & (0.11) \\
    $\textit{dbp}$   & -0.29 & -0.04 & 0.34  & 0.17  & -0.40 & -0.11 & 0.25  & 0.59 \\
          & (0.097) & (0.257) & (0.123) & (0.151) & (0.109) & (0.219) & (0.134) & (0.12) \\
    $\textit{tsft}$  & 0.38  & 0.41  & 0.57  & 0.74  & -1.33 & -1.27 & 0.24  & 0.31 \\
          & (0.063) & (0.084) & (0.076) & (0.074) & (0.061) & (0.229) & (0.146) & (0.098) \\
    $\textit{si}$    & -0.66 & -0.34 & 0.72  & 0.65  & -0.39 & 0.02  & -0.12 & 0.65 \\
          & (0.084) & (0.178) & (0.109) & (0.116) & (0.097) & (0.193) & (0.149) & (0.121) \\
    $\textit{bmi}$   & 0.24  & 0.42  & 0.43  & 0.71  & -1.01 & -0.62 & -0.11 & 0.26 \\
          & (0.085) & (0.172) & (0.132) & (0.106) & (0.082) & (0.16) & (0.159) & (0.107) \\
    $\textit{dpf}$   & -0.25 & -0.12 & 0.45  & 0.48  & -0.30 & 0.43  & -0.30 & 0.13 \\
          & (0.099) & (0.303) & (0.132) & (0.126) & (0.098) & (0.27) & (0.172) & (0.147) \\
    $\textit{age}$   & -0.58 & -0.34 & -0.36 & -0.15 & -0.58 & -0.07 & 1.52  & 1.75 \\
          & (0.041) & (0.125) & (0.072) & (0.064) & (0.046) & (0.174) & (0.128) & (0.091) \\
\bottomrule
    \end{tabular}%
      \end{adjustbox}

  \centering
  \caption*{}
    \begin{adjustbox}{max width=\linewidth}
    \begin{tabular}{ccccccccc}
\toprule
      Cluster    & \multicolumn{2}{c}{1} & \multicolumn{2}{c}{2} & \multicolumn{2}{c}{3} & \multicolumn{2}{c}{4}  \\
\cmidrule(rl){1-1} \cmidrule(rl){2-3} \cmidrule(rl){4-5} \cmidrule(rl){6-7} \cmidrule(rl){8-9}  
       Vairable   & $\textit{gs}=0$ & $\textit{gs}=1$ & $\textit{gs}=0$ & $\textit{gs}=1$ &$\textit{gs}=0$ & $\textit{gs}=1$ &$\textit{gs}=0$ & $\textit{gs}=1$  \\
\midrule
    $\textit{gtt}$   & -0.86 & -0.87 & 0.59  & 0.93  & -0.37 & -0.11 & 0.79  & 0.43 \\
          & (0.072) & (0.121) & (0.085) & (0.097) & (0.107) & (0.143) & (0.303) & (0.108) \\
    $\textit{dbp}$   & -0.27 & -0.23 & 0.21  & 0.3   & -0.47 & -0.16 & 0.45  & 0.45 \\
          & (0.105) & (0.184) & (0.148) & (0.122) & (0.118) & (0.179) & (0.276) & (0.097) \\
    $\textit{tsft}$  & 0.41  & 0.33  & 0.72  & 0.59  & -1.31 & -1.35 & 0.36  & 0.27 \\
          & (0.069) & (0.096) & (0.075) & (0.075) & (0.068) & (0.119) & (0.187) & (0.092) \\
    $\textit{si}$    & -0.62 & -0.62 & 0.64  & 0.73  & -0.39 & -0.24 & 0.71  & 0.27 \\
          & (0.091) & (0.147) & (0.114) & (0.11) & (0.113) & (0.146) & (0.302) & (0.108) \\
    $\textit{bmi}$   & 0.31  & 0.13  & 0.75  & 0.35  & -1.04 & -0.8  & 0.31  & 0.07 \\
          & (0.099) & (0.108) & (0.122) & (0.1) & (0.095) & (0.121) & (0.225) & (0.101) \\
    $\textit{dpf}$   & -0.24 & -0.23 & 0.57  & 0.34  & -0.15 & -0.34 & -0.34 & 0 \\
          & (0.114) & (0.17) & (0.117) & (0.142) & (0.121) & (0.152) & (0.24) & (0.125) \\
    $\textit{age}$   & -0.66 & -0.28 & -0.43 & -0.01 & -0.71 & -0.14 & 1.63  & 1.65 \\
          & (0.038) & (0.083) & (0.055) & (0.072) & (0.037) & (0.094) & (0.217) & (0.082) \\
\bottomrule
    \end{tabular}%
      \end{adjustbox}
           \label{tab:label3}%
\end{table}%

\begin{table}[H]
  \centering
   \caption{Results under 90/10 cluster imbalance for model {M1} with $(p,k)=(6,2)$. ARI and RSE values are scaled by $10^2$. Each RSE pair lists PML and PMML, respectively. The generating number of clusters is two; selected cluster counts and correct-selection proportions refer to SPIC.}
     \begin{adjustbox}{max width=\linewidth}
\begin{tabular}{ccccccccc}
\toprule
$\sigma$ & $n$ & Mean OC ARI
& \multicolumn{2}{c}{Mean-vector RSE}
& \multicolumn{2}{c}{Covariance RSE}
& Mean selected $k$
& Correct-selection proportion \\
\cmidrule(lr){4-5}
\cmidrule(lr){6-7}
& & & PML & PMML & PML & PMML & & \\
\midrule
1.2 &  250  & 13.04 & 6.43 & 6.46 &  7.10 &  7.13 & 3.01 & 0.146 \\
1.2 &  500  &  3.29 & 6.41 & 6.49 &  6.29 &  6.35 & 3.09 & 0.032 \\
1.2 & 1000  &  0.40 & 6.11 & 6.16 &  5.15 &  5.22 & 3.18 & 0.004 \\
\midrule
1.6 &  250  &  4.08 & 9.40 & 9.45 & 12.83 & 12.86 & 2.71 & 0.456 \\
1.6 &  500  &  1.30 & 9.19 & 9.28 & 11.09 & 11.12 & 2.71 & 0.390 \\
1.6 & 1000  &  0.22 & 8.55 & 8.56 &  9.45 &  9.48 & 3.07 & 0.174 \\
\bottomrule
    \end{tabular}%
      \end{adjustbox}

           \label{tab:label4}%
\end{table}%

% Table generated by Excel2LaTeX from sheet 'Sheet 1'
\begin{table}[H]
  \centering
  \caption{Mean ARIs, scaled by \(10^2\), between the realized latent partition and SP-estimated partitions under models {M1} and {M2}, based on ten replications per configuration. The data-driven SP tuning rule is compared with the first, fifth, and tenth candidates, \(\lambda_1,\lambda_5,\lambda_{10}\), in the shrinkage sequence with \(J=10\). These are fixed positions in a data-dependent sequence, not fixed numerical values across samples.}
  \begin{adjustbox}{max width=0.9\linewidth, max totalheight=0.9\textheight}
    \begin{tabular}{ccccccccccc}
    \toprule
    \multicolumn{11}{c}{M1}\\
    \midrule
        \multicolumn{2}{c}{$(p,k)$}&\multicolumn{4}{c}{(6,2)} & &\multicolumn{4}{c}{(10,2)} \\
        \cmidrule(rl){1-2} \cmidrule(rl){3-6}  \cmidrule(rl){8-11} 
    $\sigma$ & $n$     & SP    & $\lambda_1$ & $\lambda_5$ & $\lambda_{10}$ &       & SP    & $\lambda_1$ & $\lambda_5$ & $\lambda_{10}$ \\
    \midrule
    1.2   & 250   & 95.55 & 95.25 & 95.56 & 95.71 &       & 98.72 & 98.40 & 98.72 & 98.72 \\
          & 500   & 96.50 & 95.80 & 96.11 & 96.19 &       & 98.88 & 98.72 & 98.88 & 98.88 \\
          & 1000  & 95.25 & 94.31 & 94.70 & 94.97 &       & 99.36 & 98.84 & 99.16 & 99.36 \\
          \midrule
    1.6   & 250   & 68.83 & 63.04 & 65.20 & 66.73 &       & 87.09 & 84.74 & 86.54 & 87.19 \\
          & 500   & 71.09 & 66.53 & 68.10 & 69.20 &       & 89.92 & 88.42 & 89.09 & 89.47 \\
          & 1000  & 72.75 & 67.69 & 70.59 & 71.78 &       & 89.85 & 86.58 & 88.12 & 89.24 \\
          \midrule
        \multicolumn{2}{c}{$(p,k)$}&\multicolumn{4}{c}{(10,3)} & &\multicolumn{4}{c}{(15,7)} \\
        \cmidrule(rl){1-2} \cmidrule(rl){3-6}  \cmidrule(rl){8-11} 
    $\sigma$ & $n$     & SP    & $\lambda_1$ & $\lambda_5$ & $\lambda_{10}$ &       & SP    & $\lambda_1$ & $\lambda_5$ & $\lambda_{10}$ \\
    \midrule
    1.2   & 250   & 95.78 & 94.37 & 95.06 & 95.20 &       & 95.18 & 95.06 & 95.33 & 95.22 \\
          & 500   & 95.36 & 94.23 & 94.72 & 95.13 &       & 98.03 & 97.89 & 97.99 & 98.04 \\
          & 1000  & 95.84 & 93.72 & 94.86 & 95.55 &       & 99.65 & 99.51 & 99.57 & 99.65 \\
          \midrule
    1.6   & 250   & 71.72 & 68.48 & 69.28 & 70.13 &       & 66.24 & 64.70 & 65.59 & 65.71 \\
          & 500   & 69.98 & 67.19 & 68.47 & 69.02 &       & 86.48 & 84.37 & 85.45 & 86.24 \\
          & 1000  & 78.38 & 73.58 & 75.03 & 77.48 &       & 92.82 & 89.44 & 91.27 & 92.26 \\
          \midrule
          \multicolumn{11}{c}{M2}\\
    \midrule
        \multicolumn{2}{c}{$(p,k)$}&\multicolumn{4}{c}{(6,2)} & &\multicolumn{4}{c}{(10,2)} \\
        \cmidrule(rl){1-2} \cmidrule(rl){3-6}  \cmidrule(rl){8-11} 
    $\sigma$ & $n$     & SP    & $\lambda_1$ & $\lambda_5$ & $\lambda_{10}$ &       & SP    & $\lambda_1$ & $\lambda_5$ & $\lambda_{10}$ \\
    \midrule
    1.2   & 250   & 91.74 & 89.62 & 90.22 & 91.13 &       & 97.78 & 97.93 & 98.09 & 98.09 \\
          & 500   & 91.92 & 90.31 & 91.16 & 91.54 &       & 98.01 & 97.53 & 97.93 & 97.69 \\
          & 1000  & 92.57 & 91.45 & 91.72 & 91.87 &       & 98.13 & 97.14 & 97.57 & 97.85 \\
              \midrule
    1.6   & 250   & 67.50 & 68.77 & 67.95 & 68.04 &       & 84.85 & 85.39 & 85.40 & 85.02 \\
          & 500   & 74.19 & 71.72 & 73.20 & 73.77 &       & 87.97 & 86.65 & 87.16 & 87.53 \\
          & 1000  & 76.76 & 74.04 & 75.11 & 76.19 &       & 90.09 & 88.09 & 89.18 & 89.67 \\
          \midrule
        \multicolumn{2}{c}{$(p,k)$}&\multicolumn{4}{c}{(10,3)} & &\multicolumn{4}{c}{(15,7)} \\
        \cmidrule(rl){1-2} \cmidrule(rl){3-6}  \cmidrule(rl){8-11} 
    $\sigma$ & $n$     & SP    & $\lambda_1$ & $\lambda_5$ & $\lambda_{10}$ &       & SP    & $\lambda_1$ & $\lambda_5$ & $\lambda_{10}$ \\
    \midrule
    1.2   & 250   & 93.45 & 92.49 & 93.20 & 93.69 &       & 89.11 & 88.99 & 88.86 & 88.98 \\
          & 500   & 94.24 & 92.87 & 93.69 & 94.18 &       & 97.34 & 97.11 & 97.28 & 97.30 \\
          & 1000  & 93.75 & 92.08 & 93.10 & 93.45 &       & 98.78 & 97.95 & 98.32 & 98.60 \\
              \midrule
    1.6   & 250   & 72.80 & 70.35 & 71.62 & 72.56 &       & 64.76 & 64.81 & 64.41 & 65.00 \\
          & 500   & 77.54 & 73.77 & 75.06 & 76.33 &       & 89.44 & 87.25 & 88.37 & 89.05 \\
          & 1000  & 79.61 & 75.90 & 77.54 & 78.51 &       & 90.49 & 87.59 & 89.03 & 89.86 \\
          \bottomrule
    \end{tabular}%
    \end{adjustbox}
  \label{tab:lam_ari}%
\end{table}%

\begin{sidewaysfigure}[p]
    \centering
    \includegraphics[width=23cm,height=14cm]{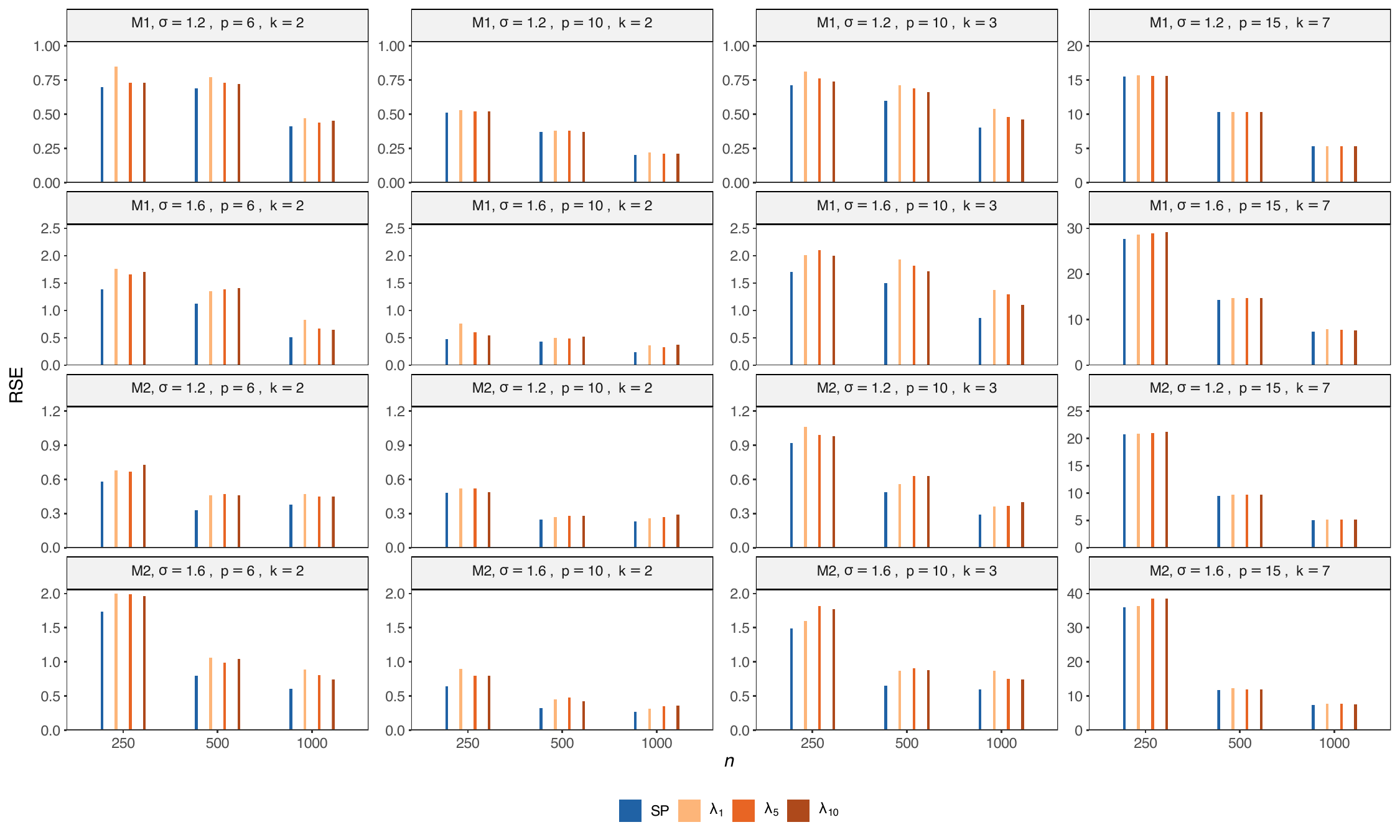}  
    \caption{\scriptsize{Mean RSEs of the SP-stage mean-vector estimates, scaled by $10^2$, under models {M1} and {M2}, based on ten replications per configuration. The data-driven tuning rule is compared with the first, fifth, and tenth candidates, $\lambda_1$, $\lambda_5$, $\lambda_{10}$, in the shrinkage sequence with $J=10$. These are fixed positions in a data-dependent sequence, not fixed numerical values across samples.}} 
    \label{fig:lam_res}  
\end{sidewaysfigure}

\begin{sidewaysfigure}[p]
    \centering
    \includegraphics[width=23cm,height=14cm]{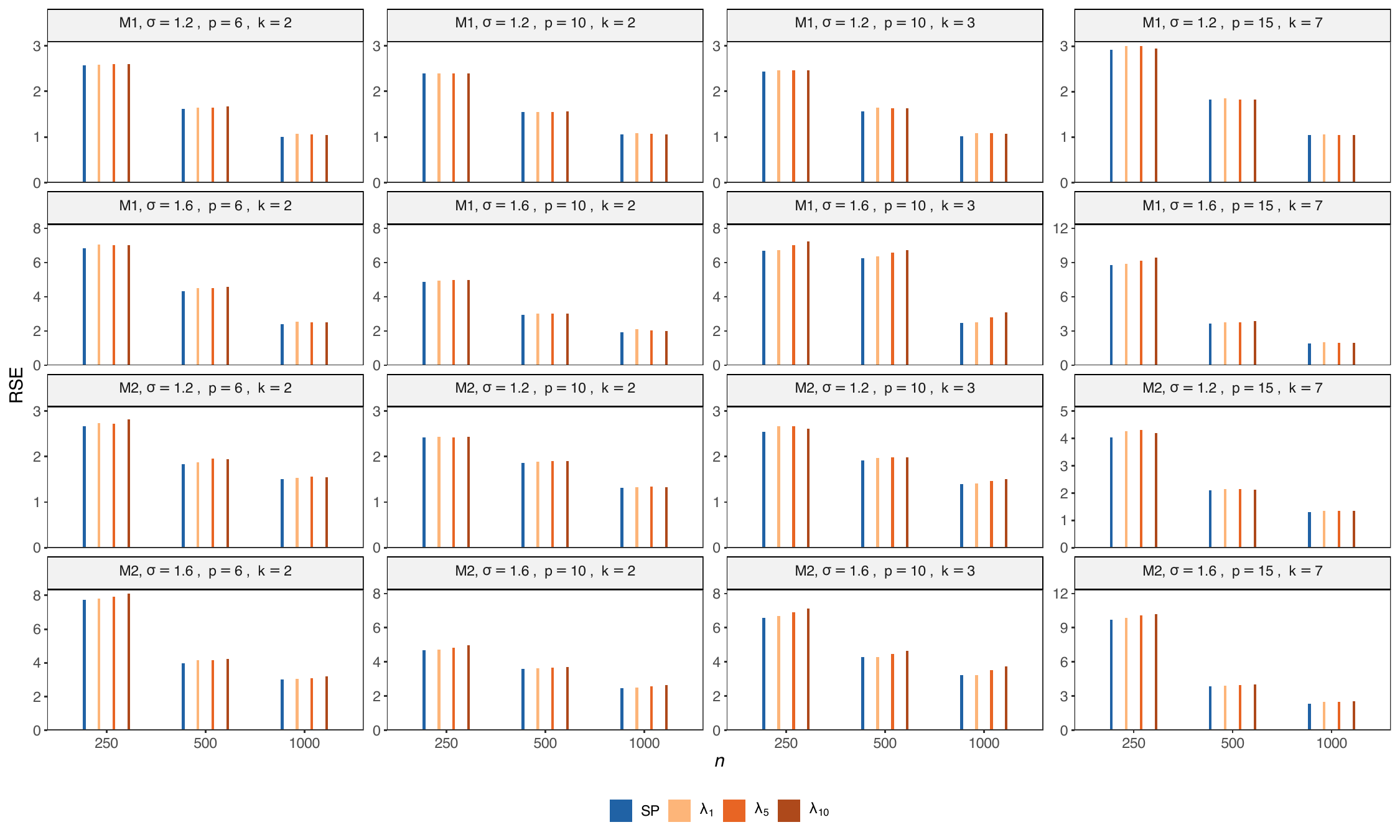}  
    \caption{\scriptsize{Mean RSEs of the SP-stage covariance-matrix estimates, scaled by $10^2$, under models {M1} and {M2}, based on ten replications per configuration. The data-driven tuning rule is compared with the first, fifth, and tenth candidates, $\lambda_1$, $\lambda_5$, $\lambda_{10}$, in the shrinkage sequence with $J=10$. These are fixed positions in a data-dependent sequence, not fixed numerical values across samples.}} 
    \label{fig:lam_res_1}  
\end{sidewaysfigure}

\end{document}